\documentclass[a4paper, reqno]{amsart}
\usepackage{a4wide}
\usepackage{amsmath,amsfonts,amssymb,amsthm,color}
\pdfoutput=1 
\usepackage{proofread}
\usepackage{comment}
\usepackage{enumitem}
\usepackage[rightcaption]{sidecap}
\sidecaptionvpos{figure}{t}
\usepackage{amsaddr} 
\usepackage{tikz}
\usetikzlibrary{decorations.markings}

\renewcommand\labelenumi{(\roman{enumi})}
\renewcommand\theenumi\labelenumi

\newcommand\dotw{\mathbin{\vcenter{\hbox{\scalebox{0.4}{$\bullet$}}}}}

\theoremstyle{plain}
\newtheorem{theorem}{Theorem}[section]
\newtheorem{proposition}[theorem]{Proposition}
\newtheorem{corollary}[theorem]{Corollary}
\newtheorem{lemma}[theorem]{Lemma}

\theoremstyle{remark}
\newtheorem{remark}[theorem]{Remark}

\newcommand{\e}{{\mathrm e}}

\newcommand{\bb}{\mathfrak{B}}

\newcommand{\xx}{{\underline x}}

\definecolor{green}{rgb}{0,0.6,0}

\newcommand{\R}{{\mathbb R}}

\newcommand{\x}{{\bf x}}

\newcommand{\yy}{{\underline y}}
\newcommand{\y}{{\bf y}}
\newcommand{\C}{{\mathbb C}}

\numberwithin{equation}{section}

\newcommand{\Par}{\Theta}
\newcommand{\Pro}{\mathbb{P}}
\newcommand{\E}{\mathbb{E}}

\newcommand{\N}{{\mathbb N}}
\newcommand{\Z}{{\mathbb Z}}

\newcommand{\iu}{\mathrm{i}}
\newcommand{\di}{\mathrm{d}}

\DeclareMathOperator{\Tr}{Tr}

\newcommand{\Reminder}{\mathfrak{R}}

\usepackage{hyperref}

\begin{document}

\title[From orbital magnetism to bulk-edge correspondence ]{From orbital magnetism\\ to bulk-edge correspondence}
\author{Horia D. Cornean, Massimo Moscolari and Stefan Teufel} 


\begin{abstract}

By extending the gauge covariant magnetic perturbation theory to operators defined on half-planes, we prove that for $2d$ random ergodic magnetic Schr\"odinger operators, the zero-temperature bulk-edge correspondence can be obtained from a general bulk-edge duality at positive temperature involving the bulk magnetization and the total edge current.

\noindent Our main result is encapsulated in a formula, which states that the derivative of a large class of bulk partition functions with respect to the external constant magnetic field, equals the expectation of a corresponding edge distribution function of the velocity component which is parallel to the edge. Neither spectral gaps, nor mobility gaps, nor topological arguments are required. 

\noindent The equality between the bulk and edge indices, as stated by the conventional bulk-edge correspondence, is obtained as a corollary of our purely analytical arguments by imposing a gap condition and by taking a ``zero-temperature" limit.

\end{abstract}

\maketitle

\tableofcontents

\section{Introduction and main results}

We consider  random magnetic Schr\"odinger operators $H_{\omega,b}$ with $b:=-e\bb$, where $\bb\in\R$ is the strength of a constant magnetic field and $e$ is the positive elementary charge. These operators act in $L^2(\R^2)$. The restriction of $H_{\omega,b}$  to the half-space $E := \{ (x_1,x_2) \in \R^2 \,|\, x_2  \geq  0\}$ with a Dirichlet boundary condition at $x_2=0$ is denoted by $H_{\omega, b}^E$. It is well known that under a suitable covariance condition with respect to magnetic $\Z^2$-translations and for a large  class of smooth distribution functions $F:\R\to \R$ which decay fast enough at infinity, the corresponding bulk density
\begin{equation}
\label{BFBirkhoff}
{B}_F(b):=  \mathbb{E}\left(\Tr \left( \chi_{\Lambda_1} F(H_{\dotw,b})\right) \right)
\stackrel{\,\mathrm{a.s.}}{=}  \lim_{L\to\infty} {L}^{-2} \Tr \left( \chi_{{\Lambda_L}} F(H_{\omega,b})\right )
\end{equation}
with $\Lambda_L := [-L/2,L/2]^2$  is well defined and almost surely equal to the average density. Here $\mathbb{E}$ denotes the expectation with respect to the randomness $\omega$. 
In this work we prove that $b \mapsto {B}_F(b)$ is  differentiable and that its derivative equals minus the total particle current parallel to the edge  for the  edge Hamiltonian $H_{\omega,b}^E$ in the state $F'(H_{\omega,b }^E)$:

\begin{equation}\label{MainEquation}
 B_F'(b) = -
\lim_{L\to \infty} \mathbb{E} \left(\Tr  \big(  \tilde \chi_L\,\iu \left[H_{\dotw,b }^E,X_1\right]  F'(H_{\dotw,b }^E)\big)\right).
\end{equation}

Here $\tilde \chi_L(x)$  is a suitable regularisation of the characteristic function of 
 the strip $[0,1]\times [0,L]$ to be discussed below and $\iu  [H_{\omega,b }^E,X_1 ]$ is the velocity operator for the edge Hamiltonian in the $x_1$-direction. 
Also,  in \eqref{MainEquation} the expectation over disorder can be replaced by a spatial average almost surely:
\begin{equation}\label{current ergodic}
\mathbb{E}\left(\Tr  \big(  \tilde \chi_L\,\iu \left[H_{\dotw,b }^E,X_1\right]  F'(H_{\dotw,b }^E)\big)\right) \stackrel{\,\mathrm{a.s.}}{=}  \lim_{M\to \infty} M^{-1}   \Tr  \big(  \tilde \chi_{L,M}\,\iu \left[H_{\omega,b }^E,X_1\right]  F'(H_{\omega,b }^E)\big) \,,
\end{equation}
where $\tilde \chi_{L,M}$  is supported on the strip $[0,M]\times [0,L]$.

While we will discuss the physical implications of equation~\eqref{MainEquation} in detail in Section \ref{sec:physics}, let us mention that for 
 $F(x)=-T \ln \big (1+ \e^{-(x-\mu)/T}\big )$, with chemical potential
 $\mu\in\R$ and temperature $T>0$, equation~\eqref{MainEquation} states that the grandcanonical bulk magnetization $m^B(b,T,\mu)$, namely the left-hand side of equation~\eqref{MainEquation} (see also \eqref{eq:bulkM}) equals the total edge charge-current  $I_1^E(b,T,\mu)$ flowing  along the edge, that is the right-hand side of equation \eqref{MainEquation} (see also \eqref{eq:TotalCurrent}), i.e.
 \begin{equation}\label{mequalsI}
m^B(b,T,\mu) =I_1^E(b,T,\mu)\qquad\mbox{for all } b,\mu \in \R\,,\; T>0\,.
\end{equation}
Under the further assumption that $\mu$ lies in an almost sure spectral gap of the bulk Hamiltonian,  by taking the limit $T\to 0$ we obtain the equality of the transverse zero-temperature bulk conductivity $\sigma_{12}(b,T=0,\mu)$ and the edge zero-temperature conductance $\sigma_{\mathrm{E}}(b,T=0,\mu)$, i.e.
\begin{equation}\label{cond-equality}
\sigma_{12}(b,0,\mu) =\sigma_{\mathrm{E}}(b,0,\mu) .
\end{equation}
As a by-product of our main result, in Appendix \ref{appendix:Kubo} we extend for the pure Landau-Hamiltonian the bulk-edge correspondence of  transport coefficients also to the case of positive temperatures. In Section \ref{sec:physicsT} we discuss how to possibly extend this result to more general settings. Furthermore, in Appendix \ref{subsec:Discrete} we briefly sketch how to apply our strategy to tight-binding models, using both position operators and switch functions.

Regarding the development of new mathematical techniques, we extend for the first time the gauge covariant magnetic perturbation theory \cite{CorneanNenciu1998, Nenciu02, CorneanNenciu09, Cornean10} to non-relativistic magnetic Schr\"odinger operators defined on infinite domains with boundary.  

\subsection{The setting} 

Each electron moves in a $\mathbb{Z}^2$-periodic scalar potential $V$  under the influence of a $\mathbb{Z}^2$-periodic magnetic field which is orthogonal to the plane. If the flux per unit cell of the periodic magnetic field is zero, then it is generated by a periodic magnetic potential; otherwise, the magnetic potential can be divided in two components: one, $\mathcal{A}$, that is periodic, and another one, $A$, that has a linear growth. We assume that $V:\R^2\to\R$ and  $\mathcal{A},A:\R^2\to\R^2$ are smooth.   Moreover, we include the presence of a disordered background. Let $\nu$ be a probability measure on $[-1,1]$, we define $\Par=[-1,1]^{\Z^2}$ and we consider $\Pro=\bigotimes_{\Z^2} \nu$ to be a probability measure on $\Par$. We denote by $\E$ the expectation on the probability space $(\Par,\Pro)$. For every $\omega \in \Par$ we define a random potential $V_{\omega}$, by 
$$
V_{\omega}(x)=\sum_{\gamma \in \Z^2} \omega_\gamma u(x-\gamma)
$$
where $u \in C_0^{\infty}(\R^2)$. By construction, $\{\omega_i\}_{i \in \Z^2}$ is a family of independent identically distributed (with probability distribution $\nu$) random variables taking values in $[-1,1]$. Then, the bulk dynamics in $L^2(\mathbb{R}^2)$
is described by the Hamiltonian
\begin{align}\label{Gizas}H_{\omega,b}=\frac{1}{2}\left(-\iu \nabla - \mathcal{A} - b A \right)^2+V+V_{\omega},
\end{align}
where $A$ is the magnetic potential of a constant unit magnetic field perpendicular to the plane,  which in the Landau gauge is given by $A(\x):=\left(-x_2,0\right)$, and $b:=-e \mathfrak{B}$ is a real parameter representing the strength of the constant magnetic field $\mathfrak{B}$.

\tikzset{->-/.style={decoration={
			markings,
			mark=at position .57 with {\arrow{>}}},postaction={decorate}}}

\tikzset{-<-/.style={decoration={
			markings,
			mark=at position .57 with {\arrow{<}}},postaction={decorate}}}
\begin{SCfigure}[1.4]
    \begin{tikzpicture}[thick,scale=0.7, every node/.style={transform shape}]

	\filldraw[fill=green!25!white, draw=white, opacity=0.2] (-5,-1.5) -- (1.7,-1.5) -- (1.7,4.3) -- (-5,4.3) -- (-5,-1.5);

	\draw[->] (-5,-1.5)--(1.7,-1.5) node[right]{};
	\draw[->] (-2.5,-3.5)--(-2.5,4.3) node[above]{};

	\foreach \x in {-4,-2.5,...,1.7}{
		\foreach \y in {-3,-1.5,...,4.3}{
			\node[draw,circle,inner sep=2pt,fill] at (\x,\y) {};}
	}
	\draw[blue,thick] (-2.5,-1.5) -- (-2.5,0) -- (-1,0) -- (-1,-1.5) -- (-2.5,-1.5);
	\filldraw[fill=blue!25!white, draw=gray, opacity=0.2](-2.5,-1.5) -- (-2.5,2.2) -- (-1,2.2) -- (-1,-1.5) -- (-2,-1.5);
	\draw[thick] (-2.6,2.2) node at (-3.2,2.2){$L$} -- (-2.4,2.2);

	\draw[black,dashed, thick] (-1,-1) -- (-1,4.3) ;
	\node at (-1.75,-0.75) {$\Omega$};
	\node at (-1.75,3.2) {$\mathcal{S}_{\infty}$};
	\node at (-1.75,1.2) {$\mathcal{S}_{L}$};

	\node at (1.3,3.7) {$E$};

	\end{tikzpicture}
	\caption{The figure represents the  horizontal edge setting.  The Hamiltonian $H^{E}_{\omega,b}$ is the Dirichlet realization of $H_{\omega,b}$ on the infinite half-plane $E$, that is the pale green area in the figure.}
	\label{fig:Setting4}
\end{SCfigure}
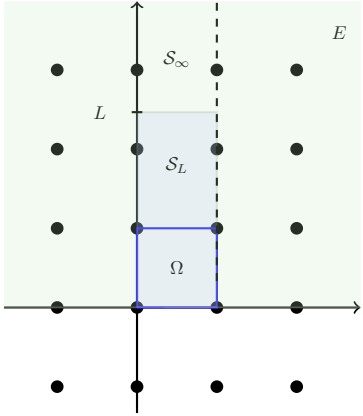

The family of selfadjoint operators $\{H_{\omega,b}\}_{\omega \in \Par}$ satisfies the following covariant relation
\begin{equation}
\label{eq:covariance2D}
\tau_{b,\gamma}H_{\omega,b}\tau^*_{b,\gamma}=H_{T(\gamma)\omega,b}, \quad \forall \, \gamma \in \Z^2,
\end{equation}
where $\tau_{b,\gamma}$ is a family of magnetic translations compatible with the Landau gauge, namely the family of operator defined for every $\gamma \in \Z^2$ by $(\tau_{b,\gamma} \psi )(x)= \e^{\mathrm{i} b \left(\gamma_1-x_1\right) \gamma_2} \psi(x-\gamma)$ for every $\psi \in \C^\infty(\R^2)$, while $T(\gamma)$ is the canonical action of $\Z^2$ on $\Theta$. Hence $(H_{\omega,b})_{\omega \in \Par}$ is ergodic with respect to the lattice $\Z^2$.

The edge dynamics is described by an Hamiltonian acting on $L^2(E)$, where
$$E:=\left\{(x_1,x_2)\in \R^2 |\;  x_2 \geq  0 \right\}.$$ 
The edge Hamiltonian is denoted by
$H^E_{\omega,b}$ and equals the Dirichlet realization of $H_{\omega,b}$ in $E$. 

It is straightforward to check that  $(H^E_{\omega,b})_{\omega \in \Par}$ is still ergodic with respect to the one-dimensional lattice generated by the vector $(1,0)$, thus in the presence of a boundary this  family is only $x_1$-covariant, in the sense that
\begin{equation}
\label{eq:covariance1D}
\tau_{b,\gamma}H^E_{\omega,b}\tau^*_{b,\gamma}=H^E_{T(\gamma)\omega,b} \qquad \forall \, \gamma=(\gamma_1,0) \in \Z^2 \,.
\end{equation}

Under the previous hypothesis, we have that, for every $\omega \in \Par$, $H_{\omega,b}$ is essentially selfadjoint on $C^\infty_0(\R^2)$, while $H^E_{\omega,b}$ is essentially self-adjoint on  $C^{\infty}(E)$ (up to the boundary) with a Dirichlet boundary condition at $x_2=0$ and with compact support. Furthermore, both operators are bounded from below and we denote by $\Sigma(b)$, and $\Sigma_{E}(b)$ respectively, their almost sure spectrum. Note that $\inf \Sigma(0) \leq \inf \Sigma(b)\leq \inf \Sigma_E(b)$ due to the diamagnetic inequality and the min-max principle.

\subsection{General bulk-edge correspondence at positive temperature}

We are now ready to state our main result regarding a general correspondence between bulk and edge quantities that holds true at any positive temperature.

\begin{theorem}\label{thm-positive}Let $\Omega=[0,1]\times [0,1]\subset \R^2$ and denote by $\chi_{\Omega}$ its indicator function. Let also $\chi_{L}$ be the indicator function of the strip $\mathcal{S}_{L}:=[0,1]\times [0,L]$, $L\geq 1$.  
Let $F$ be any real valued function whose restriction to $[\inf \Sigma(0),\infty)$ can be extended to a Schwartz function.  Then:

\begin{enumerate}[label=(\roman*)]
\item  Both $\chi_{\Omega} F(H_{\omega,b})$ and $\chi_{L}  F(H_{\omega,b}^E)$ are trace class and we may define: 
\begin{equation*}
\rho_{L,\omega}(b):=\frac{1}{L} \Tr \left(\chi_{L}  F(H_{\omega,b}^E)\right),\quad B_F(b):=\mathbb{E}\left({\rm Tr}\big (\chi_{\Omega} F(H_{\dotw,b})\big )\right).
\end{equation*}
Also, $B_F(\cdot)$ and $\rho_{L,\omega}(\cdot)$ are  differentiable  and
\begin{equation}\label{hc1}
\lim_{L\to\infty}\rho_{L,\omega}(b)=B_F(b),\quad \lim_{L\to\infty}\frac{\di\rho_{L,\omega}}{\di b}(b)=\frac{\di B_F}{\di b}(b)\quad\mbox{for almost all $\omega\in \Par$}\,.
\end{equation}

\item 
Let $g\in C^1([0,1])$ be any real function such that $g(0)=1$ and $g(1)=0$. Define $\widetilde{\chi}_L(\x):= \chi_{L}(\x) g(x_2/L)$. Then:
\begin{equation}\label{ar1}
    \frac{\di B_F}{\di b}(b)= -\lim_{L\to \infty} \mathbb{E}\left(\Tr  \left( \widetilde{\chi}_{L}\iu \left[H_{\dotw,b}^E,X_1\right]  F'(H_{\dotw,b}^E)\right)\right).
\end{equation}
\end{enumerate}
\end{theorem}

\vspace{0.2cm}

\medskip

The quantity $B_F(b)$ is a standard generalization of the integrated density of states for the bulk Hamiltonian and, in the same way, the quantity $\rho_{L,\omega}(b)$ can
be interpreted as a generalized integrated density of states of the edge Hamiltonian corresponding to a strip of
width $1$ and height $L$, which starts at the edge and goes far up into the bulk, see Figure \ref{fig:Setting4}. The trace class properties of $\chi_{\Omega} F(H_{\omega,b})$ and $\chi_{L}  F(H_{\omega,b}^E)$ are just a consequence of the localization of the integral kernels of $F(H_{\omega,b})$ and $F(H_{\omega,b}^E)$, see Lemma \ref{lemma:IntTrClassEstimateF}. Moreover, a standard application of Birkhoff's ergodic theorem allows to replace the disorder by a spatial average almost surely, see for example \eqref{BFBirkhoff} and \eqref{current ergodic}.  

Then, while the first limit in \eqref{hc1} is a consequence of the locality of the elliptic operators we work with, the second limit in \eqref{hc1} is more subtle because it requires a careful analysis and control of the magnetic dependence of $\rho_{L,\omega}$, which we do by extending the gauge covariant magnetic perturbation theory to operators defined on infinite domains with boundary. Finally, \eqref{ar1} connects the magnetic derivative of the bulk generalized integrated density of states with the total edge current. The physical interpretation of Theorem \ref{thm-positive} in terms of bulk and edge quantities is discussed in Section \ref{sec:physics}, while the definition and role of the total edge current and bulk magnetization is thoroughly discussed in Section \ref{sec:currentDef} and Section \ref{sec:magnetization} respectively.  

\begin{remark}
The results of Theorem \ref{thm-positive} still hold true, with essentially the same proof, in the case where the edge Hamiltonian is perturbed by a smooth potential $W_\omega$ supported in a finite strip near the edge, that is $\textrm{supp}(W_\omega) \subseteq \R \times [0,d]$, $d>0$. More precisely, let $H^{E,W}_{\omega,b}=H^E_{\omega,b} + W_\omega$ be the perturbed edge Hamiltonian densely defined on $L^2(E)$ with Dirichlet boundary condition at $x_2=0$. We assume that $(H^{E,W}_{\omega,b})_{\omega \in \Par}$ is still ergodic on the one-dimensional lattice generated by $(1,0)$. Then for all $g$ satisfying the assumptions listed in Theorem \ref{thm-positive} we have that 
\begin{equation}\label{potentialEQ}
   \lim_{L\to \infty} \mathbb{E}\left(\Tr  \left( \widetilde{\chi}_{L}\iu \left[H_{\dotw,b}^{E,W},X_1\right]  F'(H_{\dotw,b}^{E,W})\right)\right)=  \lim_{L\to \infty} \mathbb{E}\left(\Tr  \left( \widetilde{\chi}_{L}\iu \left[H_{\dotw,b}^E,X_1\right]  F'(H_{\dotw,b}^E)\right)\right)\,,
\end{equation}
which generalizes \cite[Theorem 1]{CombesGerminet} to positive temperature.
Physically \eqref{potentialEQ} means that the total edge current is independent of the potential landscape near the edge. This is important,  since in real samples  the interaction between the electrons leads to variations of the density near the edge resulting in a complicated  self-consistent landscape of the electro-chemical potential, see e.g.\ \cite{weis} and references therein. 
In particular, the microscopic edge current density is very sensitive to these microscopic details as well as to the exact location of the edge, while the total edge current is not. We take this as an indication that while the precise form of edge-states of specific  one-body Hamiltonians is of mathematical interest, it does not play a major role in the physical explanation of the universality of edge currents. 
\end{remark}

\begin{remark}
Our results in  \eqref{hc1} are related to the thermodynamic limit for the Landau diamagnetism  \cite{AngelescuNenciuBundaru,HeSj, Cornean00, Resta2010, BS}. They are also connected with the analysis of a higher order magnetic response of the bulk system, which consists of proving asymptotic series in $b-b_0$ of the type 
$$B_F(b)\sim B_F(b_0)+\sum_{n\geq 1} (b-b_0)^n \Xi_n(b_0).$$
In this kind of expansions, our gauge covariant magnetic perturbation theory plays a crucial role and has been adapted by physicists to various discrete bulk models. A typical physical application is the study of the zero-field orbital susceptibility \cite{BS, Savoie}, but also higher order expansions are recently considered to be of physical interest \cite{Haldane}.    
\end{remark}

\begin{remark}
\label{rmk:genericUnitcells}
 In Theorem~\ref{thm-positive},  
 $(H_{\omega,b})_{\omega \in \Par}$ is ergodic with respect to the lattice $\Z^2$. However, the results of Theorem~\ref{thm-positive} hold true for a generic two-dimensional lattice $\Gamma \subset \R^2$ generated by two linearly independent vectors ${\bf a_1},{\bf a_2}$ (where we assume without loss of generality that ${\bf a_1}=(a_{11},0)$ with $a_{11}>0$ and ${\bf a_2}=(a_{21},a_{22})$ with $a_{22}>0$), provided that one makes the following changes:
\begin{enumerate}
    \item The fundamental cell $\Omega$ is replaced by the fundamental cell of the new lattice $\Gamma$, that is $\Omega':=\{\x \in \R^2 \, | \, \x=\beta_1 {\bf a_1} + \beta_2 {\bf a_2}, \; \beta_1,\beta_2 \in [0,1] \}$. Also, the set $\mathcal{S}_L$ has to be replaced by $\mathcal{S}_L':=\{\x \in \R^2 \, | \, \x=\beta_1 {\bf a_1} + \beta_2 (0,a_{22}), \; \beta_1\in [0,1], \; \beta_2\in [0,L] \}$.
    \item The function $\chi_{\Omega}$ is replaced by $\frac{1}{|\Omega'|}\chi_{\Omega'}$, where $|\Omega'|$ denotes the area of $\Omega'$.
    \item The function $\widetilde{\chi}_{L}$ has to be replaced by $\chi_{\mathcal{S}_L'}(\x) g(x_2/(La_{22}))$. 
\end{enumerate}
In particular, this implies that our main result \eqref{MainEquation} holds true also for more general half-spaces of the form $E_{(\alpha_1,\alpha_2)}:=\left\{ (x_1,x_2) \in \R^2 \, | \,   x_2  \geq  \alpha_1 x_1 + \alpha_2\right\}$ with $(\alpha_1,\alpha_2)\in \mathbb{Q}\times\mathbb{R}$, {\it i.e.} when the cut defining the edge Hamiltonian follows any shifted crystallographic plane. A detailed presentation of such generalization involves a careful use of unitary rotation maps and it will be provided elsewhere.
\end{remark}

\subsection{Physical implications of Theorem \ref{thm-positive}}
\label{sec:physics}
In this section we give a physical interpretation of Theorem \ref{thm-positive}. The more mathematically inclined readers can skip this part on their first read.

Depending on the precise form of $F$, equation~\eqref{MainEquation} acquires different physical meanings, but in all cases,  $B_F$ is a thermodynamic bulk density and the right-hand side of \eqref{MainEquation} is some form of an edge current.

The most important situation is the case in which $F=F_\mathrm{GCP}(x):=-T \ln \big (1+ \e^{-(x-\mu)/T}\big )$ with chemical potential
 $\mu\in\R$ and temperature $T>0$. According to   standard equilibrium quantum statistical mechanics, 
 the bulk quantity  $-B_{F_\mathrm{GCP}}(b)$ equals the grand canonical pressure $p^B(b,T,\mu)$. The regularity properties of this function in terms of $b$ and $\mu$ were thoroughly studied in \cite{BS}. In particular, it is $C^2$ in $b$ and $\mu$ when $T>0$. Its partial derivatives define the (orbital) magnetization $m^B(b,T,\mu) := \partial_{\bb}\, p^B(b,T,\mu)$ (here $\bb=-b/e$) and the particle density $n^B(b,T,\mu):=\partial_\mu\, p^B(b,T,\mu) $.
 
 We notice that  $F_\mathrm{GCP}'(x)=\left(\e^{(x-\mu)/T}+1\right)^{-1} = F_\mathrm{FD}(x)$ is the Fermi-Dirac distribution and it turns out that 
 $$n^B(b,T,\mu)=\partial_{\mu}\, p^B(b,T,\mu)=B_{F_\mathrm{FD}}(b).$$
 
 Differentiating $-B_{F_\mathrm{GCP}}(b)$ with respect to  $\bb=-b/e$  and multiplying both sides with $e$, the new left-hand side of \eqref{MainEquation}  becomes the magnetization while the new right-hand side  becomes the total edge  charge-current $I^E_1(b,T,\mu)$ (see also Section \ref{sec:currentDef})
 in  the   Gibbs state $F_\mathrm{FD}(H_{\omega, b}^E)$.  In conclusion, if  $F=F_\mathrm{GCP}$, our main equation~\eqref{MainEquation} 
 leads to  \eqref{mequalsI}. This relation together with \eqref{hc1} is also connected with the well-known problem of the thermodynamic limit of Landau diamagnetism \cite{AngelescuNenciuBundaru,HeSj, Cornean00, BS} and diamagnetic currents \cite{MacrisMartinPule, Kunz}. Formula \eqref{mequalsI} extends similar relations derived in \cite{MacrisMartinPule, Kunz} for the case of pure Landau Hamiltonians with  Boltzmann partition functions,  {\it i.e.}  $F(x)=e^{-x/T}$. There, the authors can only handle a pure exponential function because their approach is based on the Feynman-Kac formula applied to magnetic Gibbs semigroups.

 \smallskip

Differentiating with respect to $\mu$ in formula \eqref{mequalsI}, followed by a zero-temperature limit and coupled with the well-known St\v reda formula, leads to our formula \eqref{cond-equality}. This provides a new alternative proof to the bulk-edge correspondence between the zero-temperature bulk conductivity and edge conductance, as it is explained in Section \ref{sec:ZeroT}. However, in order to understand the implications of \eqref{mequalsI} for the positive temperature quantum transport coefficients, we have to do a preliminary analysis on the bulk magnetization and the total edge current.

\subsubsection{The bulk magnetization}
\label{sec:magnetization}
Assume for simplicity that we only have a scalar potential $V$ that is smooth and $\Z^2$-periodic and that $\mathcal{A}=0$. Then, consider $H_{\Lambda_\ell}(b)=2^{-1}(-i\partial_1+bx_2)^2-2^{-1}\partial_2^2+V$ with Dirichlet boundary conditions in some big square box $\Lambda_\ell$ centered at the origin and with size length $\ell$. Let  $F_\mathrm{GCP}(x)=-T \ln \big (1+ \e^{-(x-\mu)/T}\big )$ with chemical potential
 $\mu\in\R$ and temperature $T>0$, and $F_{\rm FD}(x)=F_\mathrm{GCP}'(x)=1/(1+ \e^{(x-\mu)/T}\big )$. The grandcanonical magnetization can be written, up to physical constants, as a thermodynamic limit of the form 
\begin{align}
\label{eq:bulkM}
& m^B(b,T,\mu)=\lim_{\ell\to \infty} \frac{1}{|\Lambda_\ell|}{\rm Tr} \Big ( (\partial_b H_{\Lambda_\ell}(b)) F_{\rm FD}(H_{\Lambda_\ell}(b))\Big ) \\
&\approx\frac{1}{|\Lambda_L|}{\rm Tr} \Big ( (\partial_b H_{\Lambda_L}(b)) F_{\rm FD}(H_{\Lambda_L}(b))\Big )=-\frac{1}{|\Lambda_L|}\int_{\Lambda_L}\mathrm{d}\x \;  x_2  \Big ( P_1(b)\; F_{\rm FD}(H_{\Lambda_L}(b))\Big )(\x,\x),\quad L\gg 1.
\end{align}
Using the gauge covariance of the bulk Hamiltonian defined on the whole space, it is natural to expect that the expression of the magnetization contains two terms \cite{ThonhauserCeresoliVanderbiltResta,Resta2010}: one term which involves a bulk contribution coming from $P_1(b)F_{\rm FD}(H(b))$ defined on the whole space, that we call $m^{B, {\rm circ}}(b,T,\mu)$, and one extra term due to the presence of extended edge states and the fact that the  multiplication operator with $x_2$ is of order $L$ near the boundary of $\Lambda_L$, that we call $m^{B, {\rm res}}(b,T,\mu)$. Since the bulk magnetization, being a bulk quantity, can be expressed solely in term of the bulk Hamiltonian, we have that the edge states contribution to the thermodynamic limit can {\it also} be expressed by using the bulk Hamiltonian. This fact can be explicitly illustrated for pure Landau Hamiltonians, see Appendix \ref{appendix:Kubo}. In the discrete case, under the extra assumption that the spectrum of the bulk Hamiltonian is composed by $M$ simple Bloch bands $\{E_{l}(k)\}^M_{l=1}$, it has been shown in \cite[Corollary 4]{SchulzBaldesTeufel13} and \cite[Corollary 3]{StiepanTeufel} that one can split the bulk magnetization into two terms
\begin{equation}
\label{eq:magnetizationSplitting}
\begin{aligned}
m^B(b,T,\mu)=m^{B, \rm circ}(b,T,\mu) + m^{B, \rm res}(b,T,\mu)\\
m^{B, \rm circ}(b,T,\mu):=  \sum_{l=1}^{M} \int_{\mathbb{B}_b} \frac{\mathrm{d} k}{(2 \pi)^{2}}F_{ \rm FD}\left(E_{l}(k)\right) R_{1, 2}^{(l)}(k) \\
m^{B, \rm res}(b,T,\mu):= \sum_{l=1}^{M} \int_{\mathbb{B}_b} \frac{\mathrm{d} k}{(2 \pi)^{2}} F_{\rm GCP}(E_{l}(k)) \Omega_{1, 2}^{(l)}(k)
\end{aligned}
\end{equation}
where $\mathbb{B}_b$ is the enlarged magnetic Brillouin zone, $R_{1, 2}^{(l)}(k)$ is associated with the Rammal-Wilkinson tensor and $\Omega_{1, 2}^{(l)}(k)$ is the Berry curvature, see \cite{SchulzBaldesTeufel13} for a precise statement. We expect that a similar splitting also holds true in the continuum case, but the extension to the continuum setting is not trivial and we postpone its analysis to a future work \cite{CorneanMoscolariTeufelFuture}. 

By comparing the thermodynamic definition of the grandcanonical magnetization and the expression \eqref{eq:magnetizationSplitting} for the magnetization, we deduce that the \emph{residual} part of the magnetization, $m^{B, \rm res}(b,T,\mu)$, is the one associated with the edge states, hence it is the one relevant for the quantum Hall effect. Indeed, at zero temperature the Hall conductivity satisfies (see for example \cite{Resta2010}) 
\begin{equation}
\label{eq:SigmaM}
\sigma_H(b,0,\mu)=-e\partial_\mu m^{B, \rm res} (b,0,\mu).
\end{equation}
 
For the pure Landau Hamiltonian, see \cite{CorneanNenciuPedersen06} and Appendix \ref{appendix:Kubo}, we have that formula \eqref{eq:SigmaM} is valid also at positive temperature. The formal semiclassical result of \cite{XiaoYaoFangNiu,XiaoChangNiu} (see also the recent book \cite{Vanderbilt}) suggests that \eqref{eq:SigmaM}, when $\sigma_H$ is the antisymmetric part of the conductivity tensor, should be valid also at positive temperature under specific assumptions and in suitable limiting regimes, rigorous results in this direction are in preparation \cite{DeNittisMoscolariPolo}. 
Furthermore, notice that in \cite{XiaoChangNiu} the splitting of the magnetization is discussed together with the relation of $m^{B,{\rm res}}$ to the edge states by using wave packet dynamics. Instead, the term $m^{B, \rm circ}(b,T,\mu)$ is associated to the internal motion of the wave packet and, as one can see in the Landau example, this is due to the fact that electrons in the Landau levels carry a certain angular momentum. This can be also explained by analyzing the semiclassical \emph{circular} electronic orbits. \\Heuristically, we can consider $m^{B, \rm circ}(b,T,\mu)$ to be the density of the angular momentum of the bulk system. However, as this analysis suggests, the total magnetization of a quantum system is not equal to the density of angular momentum but there is some extra term due to the appearance of edge states in the thermodynamic limit procedure. This extra term is given by $m^{B, \rm res}(b,T,\mu)$, which cannot be understood in terms of any (local) density, in view of the delocalized nature of the edge states.

\subsubsection{The total edge current}
\label{sec:currentDef}
Consider $F$ to be the Fermi-Dirac distribution, namely $F_\mathrm{FD}(x):=\left(\e^{(x-\mu)/T}+1\right)^{-1}$.  But the  same arguments apply for any function $F$ satisfying the hypothesis of Theorem \ref{thm-positive}.  According to Lemma~\ref{AR2}, the so-called bulk ``averaged equilibrium current per unit area" vanishes when we average over any translated unit cell $\Omega+{\bf a}$, ${\bf a}\in\R^2$:
$$
\mathbb{E}\left( \Tr \left( \chi_{\Omega+{\bf a}} \iu \left[H_{\dotw,b },X_1\right] F_\mathrm{FD}(H_{\dotw,b})\right)\right)=0.
$$  
The main contribution to the right-hand side of \eqref{MainEquation} is thus expected to come from a particle-current flowing near the edge. However, the current density operator $  \iu [H_{\omega,b }^E,X_1]  F_{FD}'(H_{\omega,b }^E)$  
 restricted to the strip $[0,1]\times [0,\infty)$ is not trace class by itself, but only when subtracting the oscillating contributions in the bulk.
 
 Indeed, let us define the averaged density of the bulk and edge particle-current by evaluating the integral kernels of the corresponding current density operators on the diagonal and then averaging over one period in the direction parallel to the edge, 
 \[
 j_1^{B/E}(x_2) := \int_0^1 \,\mathrm{d} x_1\; \E\left(\iu \left[H^{\;\;/E}_{\dotw,b },X_1\right]  F_\mathrm{FD}(H^{\;\;/E}_{\dotw, b })\right)((x_1,x_2),(x_1,x_2))\,.
 \]
Clearly, $j_1^B$ is $\mathbb{Z}$-periodic and, by construction, $j_1^E$ is supported on $\{x_2\geq 0\}$. The vanishing of the averaged equilibrium current per unit area implies  that $j_1^B$ is zero when averaged over any interval of length one. 
By exploiting geometric perturbation theory, see Appendix \ref{appendix1}, it is possible to show \cite{MoscolariStottrup} that $j_1^E(x_2) -j_1^B(x_2) = \mathcal{O}(x_2^{-\infty})$ for $x_2\to +\infty$, hence $j_1^E -j_1^B$ is integrable. Thus   $j_1^E$  is integrable if and only if $j_1^B\equiv0$.
 The general situation is sketched in Figure~\ref{fig:edgecurrent}.  
 
 \begin{figure}
	\color{black}
	\begin{tikzpicture}[scale=1]
	
	\draw[thick,->] (-1,0)--(10,0) node[anchor= north ]{$x_2$};
	\draw[thick,->] (0,-1)--(0,5);
	\draw[thick] (7,-.1) -- (7,.1) node[anchor= south ]{$L$};
	\draw[thick] (-.1,1) -- (.1,1) node[anchor= east ]{\hspace{-5mm}$1$};
	\draw[purple] (0,0) -- (7,1) -- (10.5,1) node[anchor=south]{$1-g(\cdot/L)$};
	\draw[purple,dashed] (0,1) node[anchor=south west]{$g(\cdot/L)$} -- (7,0) -- (10.5,0);

	\draw[thick, blue] (0,4.5)  node[anchor=north west]{\quad$j_1^E$} .. controls (2.7,-2) .. (3.5,0)  .. controls (3.9,1.3) .. (4.5,0)  .. controls (4.9,-1) .. (5.5,0) .. controls (6,.9) .. (6.5,0)  .. controls (7,-.85) .. (7.5,0)  .. controls (8,.8) .. (8.5,0)  .. controls (9,-.8) .. (9.5,0)  .. controls (10,.8) .. (10.5,0) ;	
	\draw[red, dashed] (-1.5,0).. controls (-1,-.8) ..  (-0.5,0) node[anchor=south east]{\;$j_1^B$}  .. controls (0,.8)  ..  (0.5,0)  .. controls (1,-.8) .. (1.5,0).. controls (2,.8) .. (2.5,0).. controls (3,-.8) .. (3.5,0) .. controls (4,.8) .. (4.5,0) .. controls (5,-.8) .. (5.5,0) .. controls (6,.8) .. (6.5,0)  .. controls (7,-.8) .. (7.5,0)  .. controls (8,.8) .. (8.5,0)  .. controls (9,-.8) .. (9.5,0)  .. controls (10,.8) .. (10.5,0);

	\end{tikzpicture}
	\caption{The averaged density of the edge particle-current $j_1^E$ (blue curve) as a function of the distance $x_2$ to the edge approaches the oscillatory  persistent averaged density of the bulk particle-current $j_1^B$ (dashed red curve). 
	 }
	\label{fig:edgecurrent}
\end{figure}
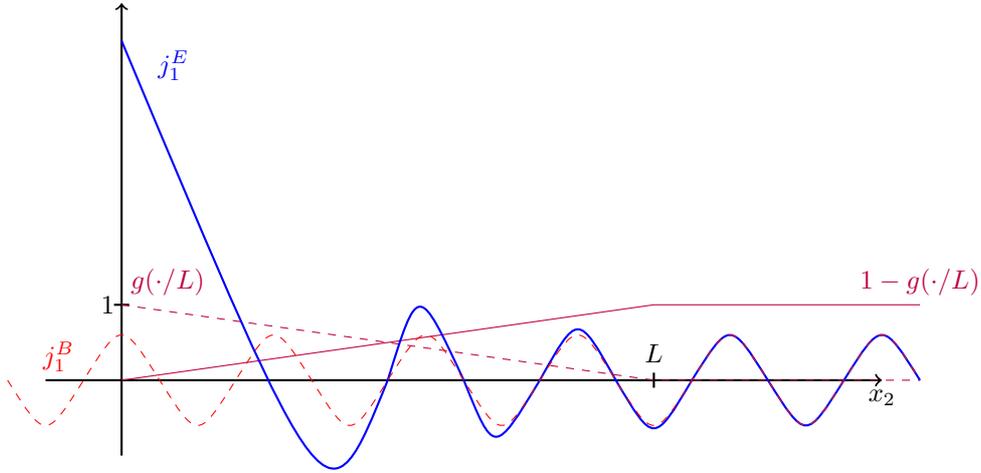
 
 It turns out that the quantity  $\int_0^\infty \left(j_1^E(x_2)- j_1^B(x_2)\right)\mathrm{d} x_2 $ is not a good definition for the total edge particle-current, since   subtracting $j_1^B$ all the way near the edge produces an uncontrolled offset. 
 A natural solution to this problem is to  subtract $j_1^B$ from $j_1^E$ only in the bulk by using a suitable cut-off function. Consider a   $C^1$-function    $g:[0,1]\to [0,1]$   with $g(0)=1$ and $g(1)=0$, and let us define the total edge charge-current as
	\begin{equation}
 \label{eq:TotalCurrent}
	I^E_1 :=-e \lim_{L\to \infty} \int_0^L \mathrm{d} x_2\; \left(j_1^E(x_2)-  \big(1- g(x_2/L)\big) j_1^B(x_2)\right)\, =
	-e \lim_{L\to \infty} \int_0^L \,\mathrm{d} x_2 \;  g(x_2/L) j_1^E(x_2) \,,
	\end{equation}
    where the last equality follows from $
	   \lim_{L\to \infty} \int_0^L \mathrm{d} x_2\;  (1- g(x_2/L))\left(  j_1^E(x_2)- j_1^B(x_2)\right) =0
	$. As we will show in Lemma~\ref{CurrentIndipG}, the value of $I^E_1$  is independent of the specific cut-off function $g$, indicating that $I^E_1$ is only encoding what is happening at the cut. 
	Two other equivalent expressions for the total edge current $I^E_1$ defined above, which could also  be considered   natural alternative  definitions, are 
\[
 I^E_1 = -e\lim_{L\to \infty} \int_L^{L+1} \mathrm{d}t  \int_0^{t}\mathrm{d}x_2 \, j_1^E(x_2) \,,
 \]
 where one  averages over the location of a sharp cutoff in the bulk, and $I_1^E = e \,J_1^E(0)$, where $J_1^E$ is the anti-derivative of $j_1^E$ that averages to zero at infinity, i.e.\   $\lim_{L\to \infty} \int_L^{L+1} J_1^E(x_2)\,\mathrm{d}x_2 =0$.
 
As already noticed in the literature \cite{CooperHalperinRuzin,XiaoYaoFangNiu,ElgartGrafSchenker,Resta2010, HaiduGummich}, not all currents are transport currents\footnote{‘‘\emph{For a quantum mechanical system in the presence of an applied magnetic field, however, there may be nonzero circulating currents even in a situation of thermodynamic equilibrium, as was noted above. We shall find it convenient to break the currents into a ‘‘transport’’ part
and a ‘‘magnetization’’ part...}’’ \cite{CooperHalperinRuzin}}. One has to make a distinction between the \emph{total} edge current and the \emph{transport} edge current. The total and transport currents differ by the so-called \emph{magnetization currents}, which are heuristically due to the the inhomogeneity of the materials and the fact that electrons might perform closed orbits that produce magnetic moments but do not contribute to transport. For this reason we define the transport {\it edge} current by subtracting the {\it bulk} circular magnetization from the {\it total} edge current:
\begin{equation}
\label{eq:ITr}
I_1^{E,{\rm tr}}:=I_1^E - m^{B,{\rm circ}}(b,T,\mu) \,.
\end{equation}

We will give an heuristic motivation to \eqref{eq:ITr} in a moment. As we explained in Section \ref{sec:magnetization}, $m^{B,{\rm circ}}(b,T,\mu)$ is associated to the density of magnetic momentum of the bulk system, hence, at a classical level, in view of Ampère's law, \eqref{eq:ITr} accounts for the subtraction of the electronic currents which induce the magnetic momentum in the bulk. 

Now let us detail our heuristic motivation of \eqref{eq:ITr}. Denote by $m^{E,{\rm circ}}(b,T,\mu)$ the density of magnetic moment of the {\it edge} system. Due to the presence of the boundary, $m^{E,{\rm circ}}(b,T,\mu)$ must be position dependent and, in view of the symmetry of the system, it should only depend on $x_2$. We have two ``boundary conditions" given by $\lim_{x_2 \to -\infty}m^{E,{\rm circ}}(b,T,\mu)(x_2) = 0$ and $\lim_{x_2 \to +\infty}m^{E,{\rm circ}}(b,T,\mu)(x_2) = m^{B,{\rm circ}}(b,T,\mu)$, which are a consequence of the fact that in the lower half-plane there is no material, and far from the cut the material can be approximated with the bulk homogeneous system. 

In the general $3d$ case, both the magnetization and the  magnetization current density are vectors, and they are related by the equation $\mathbf{j}^{E,\rm mag}:=\nabla \times \mathbf{m}^{E, \rm circ}$. 
Taking into account our particular $2d$ setting, the   magnetization    current density has only one non-zero component, parallel to the edge, and is given by
\begin{equation}
\label{eq:CurrentMag}
j^{E, {\rm mag}}_1:= \partial_2 m^{E,{\rm circ}}(b,T,\mu). 
\end{equation}
  Then, in order to compute the magnetization current that flows across a fiducial line perpendicular to the boundary we simply have to integrate \eqref{eq:CurrentMag}:
\begin{align}\label{july2022}
I_1^{E,{\rm mag}}=\int_{-\infty}^{+\infty} \mathrm{d} x_2 \; j^{E, {\rm mag}}_1(x_2) = \int_{-\infty}^{+\infty} \mathrm{d} x_2 \; \partial_2 m^{E,{\rm circ}}(b,T,\mu)(x_2) = m^{B,{\rm circ}}(b,T,\mu). 
\end{align}
By following \cite{CooperHalperinRuzin,Resta2010}, the total edge current can be split as
\begin{equation*}
I_1^{E}= I_1^{E, {\rm tr}} + I_1^{E, {\rm mag}}
\end{equation*}
which together with \eqref{july2022} it immediately implies \eqref{eq:ITr}. 

One of the main consequence of \eqref{eq:ITr} is that the edge conductance is defined by the derivative of the transport edge current with respect to the chemical potential
\begin{equation}
\label{eq:EdgeC}
\sigma_E(b,T,\mu):=-e\partial_\mu I_1^{E,{\rm tr}} = -e\partial_\mu I_1^{E} -e \partial_\mu m^{B,{\rm circ}}(b,T,\mu).
\end{equation}
Notice that \eqref{eq:EdgeC} is in accordance with the usual definition of the edge conductance whenever we are at zero temperature and the chemical potential $\mu$ belongs to an almost sure spectral gap of the bulk Hamiltonian: in such cases we have that $ \partial_\mu m^{B,{\rm circ}}(b,0,\mu)=0$, hence distinguishing between transport and total edge currents is not important. On the other hand, this distinction becomes relevant at positive temperature, as we explicitly show in the case of pure Landau Hamiltonians (see Appendix \ref{appendix:Kubo}). 

There exists another situation where subtracting the magnetization currents is crucial, even in the zero-temperature case; this happens when $\mu$ lies in a region of a mobility gap, as it has  already been recognized in \cite{ElgartGrafSchenker}. Indeed, even though the setting and the language are different, in \cite{ElgartGrafSchenker} the authors
use a definition of edge conductance similar, at least in spirit, to \eqref{eq:EdgeC}.

\medskip

\subsubsection{Bulk-edge correspondence of transport coefficients} 
\label{sec:physicsT}
 
Let us again consider $F=F_\mathrm{FD}$. Hence we have $B_{F_\mathrm{FD}}(b)=n^B(b,T,\mu)$ and Theorem \ref{thm-positive} implies:
\begin{equation}\label{sept1}
-e \frac{\partial}{\partial \bb } \,n^B(b,T,\mu)= - e^2\lim_{L\to \infty} \mathbb{E}\left(\Tr  \big \{ \widetilde{\chi}_{L}\iu \left[H_{\dotw,b}^E,X_1\right]  F_{\rm FD}'(H_{\dotw,b}^E)\big \}\right). 
\end{equation}
Using that $\partial_{\bb}n^B=\partial_\mu m^B$ together with \eqref{mequalsI} we also have that 
\begin{equation}\label{sept2}
-e\partial_\mu I_1^E(b,T,\mu)=-e \frac{\partial}{\partial \bb } \,n^B(b,T,\mu). 
\end{equation}
Formula \eqref{sept2} is a \emph{generalized} St{\v r}eda formula which is valid at any positive temperature and without any spectral assumptions. As the analysis of the paradigmatic example of the Landau Hamiltonian shows (see Appendix \ref{appendix:Kubo}), the right-hand side of \eqref{sept2} coincides with the \emph{Fermi see contribution to the transverse bulk Hall conductivity}. When $\mu$ lies in a spectral gap of the almost sure spectrum of $H_{\omega,b}$, the limit $T\searrow 0$ is further investigated in Proposition~\ref{prophc1} and Corollary~\ref{thm:Main}, where we show how the ``usual" bulk-edge correspondence of transport coefficients is recovered with an exponential rate of convergence.

Therefore, our result shows that the conventional zero-temperature bulk-edge correspondence, usually presented as an equality of certain topological indices, is just a particular manifestation of a  more general paradigm connecting bulk and edge quantities, mathematically expressed by \eqref{MainEquation}. 

\medskip

The extension of the bulk-edge correspondence of transport coefficient to positive temperature requires, at least, the knowledge of the bulk conductivity at positive temperature. As discussed in Section \ref{sec:magnetization}, a rigorous derivation and evaluation of the Kubo formula at positive temperature for quantum Hall systems is still an active field of research, some results can be found in \cite{CorneanNenciuPedersen06,AizenmanGraf,BoucletGerminetKleinSchenker, DeNittisLein}. By assuming, as it is customary in the physics literature and as it is in the pure Landau case, that
\eqref{eq:SigmaM} is valid also at positive temperature, we can obtain a positive temperature bulk-edge correspondence between quantum transport coefficient. Indeed, from \eqref{mequalsI} we have that $m^B(b,T,\mu)=I^E_1(b,T,\mu)$ and by subtracting the contribution of the magnetization current we have
\begin{equation}
    m^{B, {\rm res}}(b,T,\mu)=I^{E,{\rm tr}}_1(b,T,\mu)
\end{equation}
which is an equality between the bulk residual magnetization and the transport edge current. Finally, by taking the derivative with respect to $\mu$, we get an equality between the Hall conductivity, see \eqref{eq:SigmaM}, and the edge conductance \eqref{eq:EdgeC}. Further analysis of the positive temperature coefficient in particular cases, like the case of simple isolated Bloch bands, see for example \cite{SchulzBaldesTeufel13}, or the case of $\mu$ lying in a mobility gap, see for example \cite{ElgartGrafSchenker}, are postponed to future works \cite{CorneanMoscolariTeufelFuture}.

\subsection{The zero-temperature limit} \label{sec:ZeroT}
We now show how to recover the ``usual" zero-temperature bulk-edge correspondence for gapped systems from our general result that connects the bulk magnetization to the total edge current. Let us assume that the almost sure spectrum $\Sigma(b_0)$ of the bulk Hamiltonian $H_{\omega,b_0}$ for a certain $b_0 \in \R$  has a  gap that includes the interval   $[E_-,E_+]$   with $E_-<E_+$. As a by-product of gauge covariant magnetic perturbation theory \cite{Cornean10}, $[E_-,E_+]$ stays away from the almost sure spectrum $\Sigma(b)$ of $H_{\omega,b}$ for $|b-b_0|$ small enough. The spectral island $\sigma_0(b):=\Sigma(b) \cap (-\infty,E_-)$ is bounded and non-empty, and denote by $P_{\omega,b}$ the spectral projection onto~$\sigma_0(b)$.

\begin{proposition}\label{prophc1}
Consider any smooth function $0\leq F_0\leq 1$ which equals $1$ on $(-\infty, E_-]$ and $0$ on $[E_+,\infty)$.
Fix $\mu\in [E_-,E_+]$. For every $T>0$ we define $F_{T,\mu}(x):=F_{FD}(x)=\big ( \e^{\frac{x-\mu}{T}} +1\big)^{-1}$.  Then there exist two constants $C_1,C_2>0$ such that
\begin{align}\label{hc4}
\left| B'_{F_{T,\mu}}(b_0)-B'_{F_{0}}(b_0)\right| \leq C_1\; \e^{-C_2/T}. 
\end{align}

In particular, $\lim_{T \searrow 0}B'_{F_{T,\mu}}(b_0)$ is independent of which $F_0$ we choose with the above properties. Moreover, if $\chi_{\infty}$ denotes the indicator function of the strip $\mathcal{S}_{\infty}:=[0,1]\times (0,\infty)$, then we have:
\begin{equation}
\label{eq:ZeroTLimit}
\lim_{T \searrow 0}B'_{F_{T,\mu}}(b_0)=B'_{F_{0}}(b_0)=- \E\left({\rm{Tr}}\left\{ \chi_{\infty} \iu \left[H^E_{\dotw,b_0},X_1\right] F_0'(H^E_{\dotw,b_0}) \right\}\right) .
\end{equation}
\end{proposition}

\begin{remark}
Because $F_0'$ is supported in the gap of $H_{\omega,b_0}$, for every $\omega \in \Par$, we have that 
$$\chi_{\infty} \iu \left[H^E_{\omega,b_0},X_1\right] F_0'(H^E_{\omega,b_0})$$
is trace class even though the strip is semi-infinite, see the proof of Proposition \ref{prophc1} for more details on this issue. 
\end{remark}

\begin{remark}
The quantity $B_{F_0}(b_0)$ is nothing but the integrated density of states of the projection $(P_{\omega,b_0})_{\omega \in \Par}$  averaged over the random potential.  According to the St{\v r}eda formula (see e.g.  \cite{Bellissard1986,Bellissard1986II,RammalBellissard,Bellissard1992,Kellendonk,SchulzBaldesTeufel13}  and \cite{CorneanNenciuPedersen06,CorneanMonacoMoscolari2018} for an analysis using covariant magnetic perturbation theory as we do in this paper) we have: $$C_0:=2\pi B_{F_0}'(b_0)=2\pi \mathbb{E} \big (\Tr \left( \chi_{\Omega}P_{\dotw,b_0} \iu \big [ [X_1,P_{\dotw,b_0}],[X_2,P_{\dotw,b_0}]\big ] \right) \big )$$ where $C_0$ is the Chern character \cite{MarcelliMonacoMoscolariPanati,MarcelliMoscolariPanati,Bellissard1986,BES} of the family $(P_{\omega,b_0})_{\omega \in \Par}$, which is proportional to the zero-temperature transverse Hall conductivity computed in the linear response regime (see for example the recent papers \cite{DeRoeckElgartFraas,MarcelliMonaco} and references therein). Combining this with Proposition~\ref{prophc1} we arrive at: 

\begin{corollary}
\label{thm:Main}
In the gapped case as explained above:
\begin{equation}
\label{EdgeCurrent}
C_0= - 2\pi \E\left({\rm{Tr}}\left\{ \chi_{\infty} \iu \left[H^E_{\dotw,b_0},X_1\right] F_0'(H^E_{\dotw,b_0}) \right\} \right).
\end{equation}
\end{corollary}
\end{remark}

As mentioned in the introduction, Corollary~\ref{thm:Main} is exactly the usual bulk-edge correspondence at zero temperature, see \cite{SchulzBaldesKellendonkRichter, KellendonkRichterSchulzBaldes, ElbauGraf,ElgartGrafSchenker,ProdanSchulzBaldes, DeNittisSchulzBaldes, KellendonkSchulzBaldes,AllridgeMaxZirnbauer}.

For completeness, in Appendix~\ref{appendix1} we sketch a simple and well-known argument (see for example \cite{CombesGerminet,GrafPorta,DeNittisSchulzBaldes}), adapted to our general magnetic setting, based on the so-called edge states and in the absence of disorder, for why the quantity on the right-hand side of \eqref{EdgeCurrent} must be an integer.

\subsection{Open questions}

\begin{enumerate}
   \item If $b_0=0$ and $H_{\omega,0}$ is time reversal invariant (i.e.\ it commutes with the anti-unitary involution given by complex conjugation), then one can prove  that $B_F(b)=B_F(-b)$. Hence, \eqref{ar1} is reduced to $0=0$ regardless of which $F$ we use. In this case, the relevant quantity is $B_F''(0)$. Moreover, if $F(x)=-T\ln(1+\e^{(x-\mu)/T})$ and $b_0=0$ then $B_F''(0)$ equals (up to a universal constant) the bulk magnetic susceptibility at zero magnetic field. An interesting open question is thus:  is there a generalized bulk-edge correspondence also in this case, which might link the bulk magnetic susceptibility with some edge current-current correlations?

    \item 
    A second open question, which seems to be mathematically quite challenging, is to allow for a larger class of magnetic Hamiltonians arising from general elliptic magnetic pseudodifferential operators, including pseudorelativistic and Dirac operators. Recently, our results have been extended \cite{CorneanMoscolariSorensen} to the case of Dirac operators with constant magnetic field and infinite-mass boundary condition, however the general case is still open.
    First, there is the boundary issue which is even more  difficult than for second order elliptic differential operators, see for example the discussion in \cite{CorneanMoscolariSorensen}. Second, the gauge covariant magnetic perturbation theory has not yet been extended to general pseudodifferential operators defined on domains with boundary. Furthermore, it would be interesting to investigate the extension of our approach to non-elliptic pseudodifferential operators, for example analysing models similar to those in \cite{GrafJudTauber}. 
    
    \item In Section \ref{sec:ZeroT} we prove that our main formula \eqref{ar1} reduces, in the zero-temperature limit, to the usual bulk-edge correspondence \eqref{eq:ZeroTLimit} when the chemical potential $\mu$ lies in a spectral gap of the almost sure spectrum of the bulk Hamiltonian. An interesting open problem concerns the zero-temperature limit when $\mu$ lies either in a region of mobility gap or in the absolutely continuous part of the bulk almost sure spectrum. While we know of no previous attempts to the zero-temperature limit in those two cases, the mobility gap case in the setting of tight-binding models at zero temperature has been discussed in \cite{ElgartGrafSchenker,BolsSchenkerShapiro}, (see also the discussion in Appendix \ref{subsec:Discrete}), where  also   the quantization of bulk conductivity and edge conductance is shown. Instead, the case of $\mu$ lying in the absolutely continuous spectrum seems to have not yet been investigated. 
    
    \item  Our main results can be easily extended to the setting of the quantum spin Hall effect with conserved spin, \textit{i.e.} $\Z_2$ time-reversal symmetric topological insulators, just by defining suitable spin density states, like it was done in \cite{MonacoMoscolari} for the so-called spin St\v reda formula. This would provide a positive temperature generalization of \cite{GrafPorta,AvilaSchulzBaldesVillegasBlas}. However, a more challenging open problem would be to extend our analysis beyond the spin conserved case thus providing a bulk-edge correspondence proof for a general bulk quantum spin Hall model like the one analyzed in \cite{MPS, MonacoUlcakar}.
    
    \item Last but not least, we would like to achieve a deeper understanding of the fact that at positive temperature, the Hall conductivity (computed within the "classical" linear response theory) does not coincide with the bulk quantity provided by the right-hand side of the (generalized) St\v reda formula \eqref{sept2}. This fact has already been pointed out in the physics literature \cite{Streda,Resta2010}, and we also comment on it in Section \ref{sec:physicsT} and especially in Appendix \ref{appendix:Kubo}, where we completely analyze the case of a pure Landau Hamiltonian. However, for more general operators, the problem is still open.

    \end{enumerate}

\subsection{Two key technical results}

We now formulate two technical results that, together with the aforementioned gauge covariant magnetic perturbation theory, play a central role in the proof of Theorem \ref{thm-positive} and Proposition \ref{prophc1}. The first one regards trace class and kernel regularity property of operators involved in the definition of the total edge current, while the second one states that the ``averaged equilibrium current per unit area" of the bulk Hamiltonian equals zero.

\begin{proposition}\label{AR1} 
The following statements hold true: 
\begin{enumerate}[label=(\roman*)]
\item 
\label{AR1:point1}The operator 
\begin{equation}\label{hc2}
\mathcal{I}_{\omega,b}:=\chi_{\infty}\big \{ \iu \left[H_{\omega,b}^E,X_1\right]  F'(H_{\omega,b}^E)-\chi_{E} \iu \left[H_{\omega,b},X_1\right]F'(H_{\omega,b})\chi_{E} \big\}
\end{equation}
is trace class for every $\omega \in \Par$, $b \in \R$, and its trace is uniformly bounded with respect to~$\omega$. 

\item \label{AR1:point2} Both operators $ \iu \left[H_{\omega,b}^E,X_1\right]  F'(H_{\omega,b}^E)$ and $\iu \left[H_{\omega,b},X_1\right]F'(H_{\omega,b})$ have jointly continuous integral kernels on $\mathcal{S}_{\infty}$, uniformly in $\omega \in \Par$.

\item \label{AR1:point3} Denote the diagonal value of the integral kernel of $\mathbb{E}\left(\mathcal{I}_{\dotw,b}\right)$ by $I_b(x_1,x_2)$. Then $I_b(x_1,x_2)$ is continuous on $\mathcal{S}_\infty$ and 
\begin{equation}\label{eq:currentL1}
\int_0^1 \di x_1 \int_{0}^\infty \di x_2 \; |I_b(x_1,x_2)|\; <\infty. 
\end{equation}
\end{enumerate}

\end{proposition}

\begin{remark}
\label{rmk:comparisonRegularizationCurrent}
 If $H_{\omega,b}$ has a gap in the spectrum and $F'$ is supported in such a gap (thus reducing to a zero-temperature setting), then $-\Tr \left( \mathcal{I}_{\omega,b}\right)= -\Tr \left(\chi_{\infty}\big \{ \iu \left[H^E_{\omega,b},X_1\right]  F'(H^E_{\omega,b}) \big \}\right)$, which coincides with the usual edge conductivity and the result of Proposition~\ref{AR1}~\ref{AR1:point1} is essentially known. Instead, when $F'$ is not supported on a gap of $H_{\omega,b}$, the subtraction of the bulk term is crucial to get the trace class property. Indeed, by exploiting geometric perturbation theory as we do in the proof of Lemma~ \ref{AR1}, one can prove that the first summand on the right-hand side of \eqref{hc2} is not even a compact operator on its own. 
 \end{remark}

\begin{lemma}\label{AR2} The ``averaged equilibrium current per unit area" vanishes, in the sense that given any real valued function $G$ which coincides with a Schwartz function on the almost sure spectrum of $H_{\omega,b}$ we have 
\begin{equation}\label{eq:EqNoCurrent}
  \E\left(\Tr \big ( \chi_{\Omega} \iu [H_{\dotw,b},X_i]G(H_{\dotw,b})\big )\right)=0, \qquad i \in \{1,2\}.
\end{equation}
Moreover, as a consequence of Birkhoff's ergodic theorem the ``equilibrium current per unit area" vanishes for almost all $\omega\in \Par$, that is
\begin{equation}
\label{eq:TDLcurrent}
\lim_{L \to \infty} \frac{\Tr \big ( \chi_{\Lambda_L} \iu [H_{\omega,b},X_i]G(H_{\omega,b})\big )}{L^2}=0, \qquad i \in \{1,2\}.
\end{equation}
\end{lemma}

\subsection{The content of the rest of the paper} $\qquad\qquad\qquad$

In Section~\ref{sect4} we prove Proposition~\ref{AR1}~\ref{AR1:point1} and \ref{AR1:point3}, while the proof of Proposition~\ref{AR1}~\ref{AR1:point2} is postponed to Appendix~\ref{appendix:IntKernels}. The proof of \ref{AR1:point1} is based on uniform exponential decay estimates and geometric perturbation theory, both briefly explained in Appendices \ref{subsec:GPT} and \ref{appendix:IntKernels} for completeness. Even though we follow a standard strategy by showing that  $\mathcal{I}_{\omega,b}$ can be written as a finite sum of products of Hilbert-Schmidt operators, the fact that $F'$ might not have compact support makes the proof much more involved.  

In Section~\ref{sect5} we show, among other things, that the ``averaged equilibrium current per unit area" is  zero. This fact enters in an essential way in the proof that the total edge current, namely the right-hand side of \eqref{ar1}, is well-defined and it is independent of $g$, see Lemma \ref{CurrentIndipG}. A previous (rather involved) proof which covers our case of unbounded magnetic random operators can be found in \cite{BoucletGerminetKleinSchenker}. We chose, for completeness, to give a different and somewhat simpler proof which reduces the unbounded setting to the bounded case, which was already treated in \cite{SchulzBaldesTeufel13}.

Section~\ref{sect6} contains the most important new ideas coming from the extension of gauge covariant magnetic perturbation theory to half-spaces, and also contains the core of the proof of our main result Theorem \ref{thm-positive}.

The proof of  Theorem \ref{thm-positive}~(i) consists of a thorough analysis of the ``$F$"- integrated density of states
$\rho_{L,\omega}(b)$. In Proposition~\ref{step1} we prove that when $L\to\infty$, this quantity is self-averaging and converges to the pure bulk quantity $B_F(b)$; the result is not  unexpected, taking into account the locality of the elliptic operators we work with, which makes that the influence of the boundary condition is localized near the cut. 

In Proposition~\ref{prop:TDLMagneticD} we prove a deeper and more difficult result which states that not only $\rho_{L,\omega}(b)$, but also its derivative with respect to $b$ has the same thermodynamic limit property and converges to $B_F'{(b)}$. The proof is challenging due to the linear growth at infinity of the magnetic potential, and the gauge covariant magnetic perturbation theory adapted to operators defined on domains with boundary plays a crucial role. We note that when $F(x)=-T \ln(1+\e^{(x-\mu)/T})$, the problem is related to the thermodynamic limit of the bulk magnetization \cite{Cornean00}.

The proof of  Theorem \ref{thm-positive}~(ii), i.e. formula \eqref{ar1}, follows directly from Propositions~\ref{prop:TDLMagneticD} and \ref{prop:TDLMagneticEdge}. 

Section~\ref{sect9} deals with the gapped case and the zero-temperature limit, as explained in Proposition~\ref{prophc1}. 

\subsection{Connections with the existing literature}

The history of bulk-edge correspondence in topological insulators goes back to the eighties, right after the discovery of the quantum Hall effect. Indeed, as it was first realized in the pioneering work by Halperin \cite{Halperin}, a comprehensive description of the integer quantum Hall effect at zero temperature requires the introduction of charge-carrying edge states, whose conductivity is quantized and exactly equal to the transverse bulk Hall conductivity. Around ten years later, this remarkable mechanism has been clarified by Hatsugai \cite{Hatsugai}, who analyzes edge states of the discrete magnetic Laplacian explicitly, and by Fröhlich et al.\ \cite{FrKe,FrSt} (see also the more recent review in \cite{Froehlich}), who derive effective actions for the bulk and the edge physics using scaling limits and minimal assumptions on the microscopic model. This led to what is nowadays called \emph{bulk-edge correspondence} (B-EC) and which is the main footprint of several types of topological materials.

The current mathematical understanding of the bulk-edge correspondence has been built on several works which either focused on the analysis of the bulk problem, \textit{i.e.\ }a Hamiltonian defined on the whole plane, or on the analysis of the edge problem, \textit{i.e.\ }a Hamiltonian defined on the half plane or a variation of the latter. Regarding the bulk setting we mention the seminal works by Bellissard et al.\ \cite{Bellissard1986,Bellissard1986II,BES}, followed by \cite{AvronSeilerSimon,AizenmanGraf,GerminetKleinSchenker}, while for the edge setting we mention the first rigorous works by De Bièvre and Pulé \cite{DeBievrePule}, and by Fr\"ohlich, Graf and Walcher \cite{FroehlichGrafWalcher} followed by the more recent \cite{CombesGerminet,FeffermanLeeThorpWeinstein,HislopPopoffSoccorsi}.

There are several proofs of  the bulk-edge correspondence in the literature. Most of them deal with tight-binding Hamiltonians: they either make use of non-commutative geometric and $K$-theoretic techniques based on \cite{Bellissard1986,Bellissard1986II,BES}, see the works by Schulz-Baldes, Kellendonk and Richter \cite{SchulzBaldesKellendonkRichter,KellendonkRichterSchulzBaldes}, the monograph \cite{ProdanSchulzBaldes}, and the more recent~\cite{AllridgeMaxZirnbauer}; or they rely on a generalization of the index of pair of projections as in \cite{AvronSeilerSimon}, see the works by Graf~et~al.~ \cite{ElbauGraf, ElgartGrafSchenker, GrafPorta}. We also mention \cite{DeNittisSchulzBaldes} and \cite{AvilaSchulzBaldesVillegasBlas}, where, in the former, the bulk-edge correspondence for a tight-binding model is proved by using a spectral flow approach related to the mathematical description of Laughlin's flux insertion argument given in \cite{AvronSeilerSimon}, while the latter is based on the transfer matrix approach of \cite{Hatsugai}. 
The first proof of B-EC for continuous disordered systems, with constant magnetic field, is due to Kellendonk and Schulz-Baldes \cite{KellendonkSchulzBaldes}. Their setting is similar to ours, but they only allow compactly supported partition functions, thus restricting to zero temperature.  Their techniques also require a spectral gap in order to properly define the edge current and conductivity.

Later on,  the results for tight-binding Hamiltonians of \cite{ElgartGrafSchenker} have been extended in \cite{Taarabt} to the continuous setting with soft walls; which means that the Hamiltonian is defined on the whole space and there is a suitable scalar potential mimicking the Dirichlet boundary condition at the cut. More recently, several new proofs of B-EC for continuous clean systems modelling a soft junction between two different crystals appeared; they are based on semiclassical analysis of pseudodifferential operators \cite{Dr2019}, on Fredholm index theory \cite{Bal2019} or on a generalization of the Maslov index \cite{Gontier2021}. Nevertheless, they do not include constant magnetic fields, which is a strong  limitation in view of the physical models used in the quantum Hall effect. Moreover, we notice that several results presented in \cite{ProdanSchulzBaldes} have been recently extended to the continuous setting, see \cite{BourneRennie} and references therein. See also the recent approach to bulk-edge correspondence using Roe alegbras in \cite{Kubota}.

Furthermore, we want to stress that even though the previously cited proofs are based on several different mathematical approaches, they all share the same underlying strategy: first they define a certain edge conductance and then they show, by comparison, that it coincides with a topological index associated to the bulk Hamiltonian, which is intrinsically defined only at zero temperature. The constant magnetic field, when present, is treated just as a fixed parameter. Instead, our new approach, which is inspired by the thermodynamic argument of St{\v r}eda and Smr{\v c}a \cite{StredaSmrcka}, is based on a careful and detailed analysis of the magnetic dependence of the bulk and edge ``smeared" integrated densities of states, thus being temperature independent and not relying on any topological arguments. Additionally, our approach can be implemented both in the case of unbounded Schr\"odinger operators and of tight-binding models. In this paper we chose to work with the technically more demanding case of unbounded random Schr\"odinger operators and we postpone the analysis of the tight-binding setting to a future work (see Appendix \ref{subsec:Discrete} for more details). 

Our results can also be adapted to a situation where the Dirichlet edge is replaced by  a ``continuous interface" between two half planes, like in~\cite{Dr2019}. A detailed analysis of this case  will be presented in another paper.

As a final remark, we notice that our main result, and also the mathematical papers cited so far, rely on the independent electron approximation. The only results in the direction of a description of the bulk-edge correspondence for interacting electrons are \cite{AntinucciMastropietroPorta,MastropietroPorta}, which, together with the companion paper \cite{GiulianiMastropietroPorta}, implies a stability of the bulk-edge correspondence with respect to suitable small interactions. Moreover, notice that in \cite{AntinucciMastropietroPorta,MastropietroPorta} the edge conductance is defined by analyzing the linear response to a weak local field along the edge, while in our work the edge conductance is defined as the derivative of the transport current in equilibrium with respect to the chemical potential as it is customary in the analysis of the non-interacting bulk-edge correspondence, see for example \cite{SchulzBaldesKellendonkRichter}. Even though a general proof of the bulk-edge correspondence for interacting electrons is out of the scope of the present paper, we expect that our new approach based on magnetic perturbations could shed some new light also on this challenging problem  \cite{LampartMoscolariTeufelWesselFuture}.

\addtocontents{toc}{\protect\setcounter{tocdepth}{0}}

\subsection*{Acknowledgments}
The authors would like to thank J.~Bellissard, E.~Canc{\`e}s, G.~De~Nittis, J.~Fr\"ohlich, D.~Gontier, G.~M.~Graf, B.~Helffer, J.~Kellendonk, A.~Levitt, C.~Lubich, G.~Marcelli, D.~Monaco, G.~Nenciu, G.~Panati, M.~Porta,  H.~Schulz-Baldes, J.~Shapiro and C.~Tauber for helpful discussions and for sharing their insights.
H.C. gratefully acknowledges the financial
support from Grant 8021-00084B and 2032-00005B of the Danish Council for Independent Research $|$ Natural Sciences. The work of M.M. has been supported by a fellowship of the Alexander von Humboldt Foundation during his stay at the University of T\"ubingen, where this work initiated.
M.M. gratefully acknowledges the support of PNRR Italia Domani and Next Generation EU through the ICSC National Research Centre for High Performance Computing, Big Data and Quantum Computing and the support of the MUR grant Dipartimento di Eccellenza 2023–2027. S.T.\ acknowledges financial support by the Deutsche Forschungsgemeinschaft (DFG, German Research Foundation) – TRR 352 – Project-ID 470903074.

\addtocontents{toc}{\protect\setcounter{tocdepth}{2}}

\section{Proof of Proposition~\ref{AR1}}\label{sect4}

  To simplify notation, let us replace $F'$ with any Schwartz function $G$. Let us start by analyzing the difference of operators
    \begin{equation}
    \label{eq:diffFprime}
     G(H^E_{\omega,b})-\chi_{E} \, G(H_{\omega,b}) \chi_{E}.
    \end{equation}
    We want to rewrite  \eqref{eq:diffFprime} with the help of the  Davies-Droste-Dynkin-Helffer-Sj\"ostrand formula \eqref{dc22}, which is a generalization of the Cauchy-Green-Pompeiu formula from complex analysis, and we simply refer to it as the Helffer-Sj\"ostrand (H-S) formula \cite{HelfferSjostrand1989} from now on. If $z=z_1+\iu z_2\in \mathbb{C}$ and $N\in \mathbb{N}$ is arbitrary, let $G_N(z_1,z_2)$ denote the almost analytic extension of $G$ with support in $|z_2|\leq 1$ such that 
    \begin{equation}\label{dc27}
     \forall \, z_1\in \R,\quad \forall \, 0<|z_2|\leq 1,\quad \zeta:=\langle z_1\rangle/ |z_2|\geq 1,\quad |\overline{\partial}G_N(z)|\leq C_N \zeta^{-N}.
    \end{equation}
    With $\mathcal{D}=\R\times [-1,1]$ we have:
    $$
    \begin{aligned}
    &G(H^E_{\omega,b})-\chi_{E} \, G(H_{\omega,b}) \chi_{E}=-\frac{1}{\pi} \int_{\mathcal{D}} \mathrm{d}z_1\mathrm{d}z_2\, \bar{\partial} {G_N}(z) \left(\left(H^E_{\omega,b}-z\right)^{-1} - \chi_{E} \left(H_{\omega,b}-z\right)^{-1}\chi_{E} \right).
    \end{aligned}
    $$
    In the following we will use the shorthand notation $H_{\omega,b}=:H$ and $H^E_{\omega,b}=:H^E$. As it will be evident from the proof, the proof does not depend on the specific value of $b$, and all the estimates are uniform with respect to $\omega \in \Par$. Therefore, we get for free that the traces exist and are uniformly bounded with respect to $\omega$.
    
    By using geometric perturbation theory as described in Section~\ref{subsec:GPT}, we can write the resolvent of the edge Hamiltonian as a sum of terms for which we can control the localization properties in the direction perpendicular to the boundary. 
    
    Consider \eqref{gpt1} for a fixed positive number $\ell>2$, that is 
    $$
    \left(H^E - z \right)^{-1}=U_{\ell}(z)  - \left(H^E- z \right)^{-1}W_{\ell}(z) .
    $$
     Using $\chi_{E}=\eta_\ell +\eta_0$ we have
    \begin{equation}
    \label{eq:TClass1}
    \begin{aligned}
    &(H^E-z)^{-1}  - \chi_{E}(H-z)^{-1}\chi_{E}=(\widetilde{\eta}_{\ell}-\eta_\ell) \left(H - z \right)^{-1} \eta_{\ell}  - \eta_0 \left(H - z \right)^{-1} \eta_{\ell} \\ 
    &-\chi_{E}\left(H - z \right)^{-1} \eta_{0} 
    +\widetilde{\eta}_0 \left( H^E - z \right)^{-1} \eta_0 - \left(H^E - z \right)^{-1}W_{\ell}(z) .
    \end{aligned}
    \end{equation}
    
    Using that $[H,X_1]=[H^E,X_1]$ are both ``tangential velocities" which  commute with functions only depending on $x_2$, the above identity also holds when we apply $[H^E,X_1]$ or $[H,X_1]$ on the left-hand side. 
    
    The first four terms have the same structure, in the sense that they have at least one factor localizing near the boundary, and can be treated in the same way. Let us first consider the third term, and then we will separately analyze the very last one, which has a different nature.
    
      First, we have  $\chi_{\infty}\chi_{E}=\chi_{\infty}$. Second, we observe that  whenever we have an operator family $A(z)$ which is analytic on $\mathcal{D}$ and with a polynomial growth in $z_1$ at infinity, then 
      $$\int_{\mathcal{D}} \mathrm{d}z_1\mathrm{d}z_2 \, \bar{\partial} {G_N}(z) A(z)=0.$$
      Third, using the previous observation and applying the first resolvent identity twice, we obtain:  
    \begin{equation}
    \label{eq:GreenthmInfty}
    \begin{aligned}
     &\int_{\mathcal{D}} \mathrm{d}z_1\mathrm{d}z_2 \, \bar{\partial} {G_N}(z) \chi_{\infty}  \left[H,X_1\right] \left( H - z \right)^{-1}  \eta_0 \\
     &= \int_{\mathcal{D}} \mathrm{d}z_1\mathrm{d}z_2 \, \bar{\partial} G_N(z)  (z-\iu)^2  \chi_{\infty}\left[H,X_1\right] \left( H - z \right)^{-1} \left( H - \iu \right)^{-2} \eta_0 .
     \end{aligned}
    \end{equation}
     Let us argue that the operator 
    $$ \chi_{\infty}  \left[H,X_1\right] \left( H - z \right)^{-1} \left( H - \iu \right)^{-2} \eta_0 $$ is trace class, with a trace norm which is bounded up to a constant by $\zeta^N$ (see \eqref{dc27} for the definition of $\zeta$) where $z\in \mathcal{D}$ is arbitrary while $N$ is a sufficiently large fixed number.  Indeed, we can decompose the integral kernel in the product of two operators 
    \begin{equation}
    \label{eq:aux8}
    \left(\chi_{\infty}  \left[H,X_1\right] \left( H - z \right)^{-1} \e^{\frac{2\beta}{\zeta} \langle X_1 \rangle }\right) \left( \e^{-\frac{2\beta}{\zeta} \langle X_1 \rangle }\left( H - \iu \right)^{-2} \eta_0\right)=: A B. 
    \end{equation}

    By using the estimate \eqref{IntRes2}, we get
    $$
    |A(\x;\y)|\leq \zeta^6 \e^{\frac{2\beta}{\zeta} } \left(1+ \frac{1}{\|\x-\y\|}\right) \e^{-\frac{(\delta-2\beta)}{\zeta}\|\x-\y\|}
    $$
    which, together with a Schur estimate, implies that $A$ is a bounded operator whose norm grows like $\zeta^N$, for some $N \in \N$. Then, it is enough to show that $B$ is a trace class operator whose norm grows polynomially in $\zeta$. Indeed, for $\beta>0$, we have
    \begin{equation}
    \label{eq:aux2TraceClass}
    \begin{aligned}
     B&=\e^{-\frac{2\beta}{\zeta} \langle X_1 \rangle }\left( H - \iu \right)^{-2} \eta_0\\
    &=\left(\e^{- \frac{2\beta}{\zeta} \langle X_1 \rangle }\left( H - \iu \right)^{-1}\e^{ \frac{2\beta}{\zeta} \langle X_1 \rangle }e^{ -\frac{\beta}{\zeta} \langle X_1 \rangle } \e^{-\frac{\beta}{\zeta}\langle X_2 \rangle }\right) \\
    &\;\;\quad \cdot \left( \e^{-\beta \frac{\eta}{\zeta} \langle X_2 \rangle } \e^{- \frac{\beta}{\zeta} \langle X_1 \rangle } \e^{ \frac{2\beta}{\zeta} \langle X_2 \rangle }\left( H - \iu \right)^{-1} \e^{-\frac{2\beta}{\zeta} \langle X_2 \rangle } \right) \e^{\frac{2\beta}{\zeta} \langle X_2 \rangle } \eta_0 \\
    & =: C \, D \, \e^{\frac{2\beta}{\zeta} \langle X_2 \rangle } \eta_0 .
    \end{aligned}
    \end{equation}
    Using \eqref{IntKRes1} together with the triangle inequality, we get
    $$
    |C(\x;\y)| \leq K  \zeta^4 \left(1+ \left| \ln \left\|\x-\y\right\|  \right| \right) \e^{-\frac{(\delta-2\beta)}{\zeta} \left\|\x-\y\right\| } \e^{- \frac{\beta}{\zeta} \langle  y_1 \rangle } \e^{- \frac{\beta}{\zeta}\langle  y_2 \rangle }.
    $$
    Then, by choosing $\beta<\delta/2$, it shows that $C$ is a Hilbert-Schmidt operator. A similar estimate shows that $D$ is a Hilbert-Schmidt operator, too. We also have that  $\e^{\frac{2\beta}{\zeta} \langle X_2 \rangle }\eta_0$  is bounded  uniformly in $z$ because $\zeta\geq 1$. Thus, both $CD$ and $B$ are trace class operators with a trace norm growing polynomially in $\zeta$.

    Therefore, we have
    \begin{equation}
    \label{eq:aux9}
    \begin{aligned}
    &\left\| \int_{\mathcal{D}} \mathrm{d}z_1\mathrm{d}z_2 \, \bar{\partial} {G_N}(z) (z-\iu)^2 \chi_{\infty}  \left[H,X_1\right] \left( H - z \right)^{-1} \left( H - \iu \right)^{-2} \eta_0 \right\|_1 \\
    &\quad \leq   \int_{\mathcal{D}} \mathrm{d}z_1\mathrm{d}z_2 \, \left|\bar{\partial} {G_N}(z) (z-\iu)^2\right| \left\| \chi_{\infty}  \left[H,X_1\right] \left( H - z \right)^{-1} \left( H - \iu \right)^{-2} \eta_0 \right\|_1 \\
    &\quad \leq  K\int_{\mathcal{D}} \mathrm{d}z_1\mathrm{d}z_2\,  \left|\bar{\partial} {G_N}(z) \right| |z-\iu| \zeta^M < \infty
    \end{aligned}
    \end{equation}
    where $\| \cdot \|_1$ denotes the trace class norm, and in the last inequality we have chosen $N$ large enough. All the other terms in \eqref{eq:TClass1}, with the exception of the last one, can be treated in an analogous way. 
    Thus it remains to prove that the operator
    $$
    \frac{\iu}{\pi} \int_{\mathcal{D}} \mathrm{d}z_1\mathrm{d}z_2\,   \bar{\partial} {G_N}(z)  \chi_{\infty} \left[H^E,X_1\right] \left(H^E - z \right)^{-1}W_{\ell}(z)
    $$
    is trace class. Repeating the same strategy that led us to \eqref{eq:GreenthmInfty}, by adding and subtracting $\chi_{\infty} \left[H^E,X_1\right] \left(H^E - \iu \right)^{-1}W_{\ell}(\iu)$, the operator becomes:
    \begin{equation}
    \label{eq:auxTraceClass}
    \begin{aligned}
    &\frac{\iu}{\pi} \int_{\mathcal{D}} \mathrm{d}z_1\mathrm{d}z_2\,  \bar{\partial} {G_N}(z)  (z-\iu) \chi_{\infty} \left[H^E,X_1\right] \left(H^E-z \right)^{-1} \left(H^E-\iu \right)^{-1} W_{\ell}(z) \\
    &+\frac{\iu}{\pi} \int_{\mathcal{D}} \mathrm{d}z_1\mathrm{d}z_2\,  \bar{\partial} {G_N}(z)  (z-\iu) \chi_{\infty} \left[H^E,X_1\right] \left(H^E-\iu \right)^{-1} \left( W_{\ell}(z) - W_{\ell}(\iu) \right) .
    \end{aligned}
    \end{equation}
    
    The first term in \eqref{eq:auxTraceClass} can be treated using the same strategy in \eqref{eq:aux8}. Indeed, $W_{\ell}(z)=\chi_{Q}W_{\ell}(z)$, where $Q:=\R \times [0,2\ell]$,  $W_{\ell}(z)$ satisfies \eqref{eq:WKern}. Regarding the second term, we can treat the four terms in $W_{\ell}(z) - W_{\ell}(\iu) $ separately, and it is enough to show how to treat the one containing $[H^E,X_2]$ and $\eta_\ell$. By using the resolvent identity, we get
    
    \begin{equation*}
    \begin{aligned}
    &\frac{1}{\pi} \int_{\mathcal{D}} \mathrm{d}z_1\mathrm{d}z_2  \bar{\partial} {G_N}(z)  (z-\iu)^2 \chi_{\infty} \left[H^E,X_1\right] \left(H^E-\iu \right)^{-1}  \partial_2\widetilde{\eta}_\ell  \left[H^E,X_2\right]  \left(H^E-z \right)^{-1} (H^E-\iu )^{-1} \eta_{\ell} .
    \end{aligned}
    \end{equation*}
    Repeating the same strategy as before, we gain a higher power of $\left(H^E-\iu \right)^{-1}$, that is
    \begin{equation}
    \label{eq:aux3TC}
    \begin{aligned}
    &\frac{1}{\pi} \int_{\mathcal{D}} \mathrm{d}z_1\mathrm{d}z_2  \bar{\partial} {G_N}(z)  (z-\iu)^3 \chi_{\infty} \left[H^E,X_1\right] \left(H^E-\iu \right)^{-1}  \partial_2\widetilde{\eta}_\ell  \left[H^E,X_2\right]  \left(H^E-z \right)^{-1} \left(H^E-\iu \right)^{-2} \eta_{\ell} .
    \end{aligned}
    \end{equation}
    Then, we have to propagate the exponential localization as in \eqref{eq:aux8}, that is we decompose the integrand in \eqref{eq:aux3TC} as the product of five operators:
    $$
    \begin{aligned}
    &\left(\chi_{\infty}\e^{\frac{2\beta}{\zeta} \langle X_1 \rangle}\right )\cdot 
    \left ( \e^{\frac{-2\beta}{\zeta} \langle X_1 \rangle } \left[H^E,X_1\right] \left(H^E-\iu \right)^{-1} \e^{\frac{2\beta}{\zeta} \langle X_1 \rangle } \right) \cdot \left ((\iu \partial_2\widetilde{\eta}_\ell) \, \e^{\frac{2\beta}{\zeta} \langle X_2 \rangle } \right ) \\
    &\quad \cdot \left( \e^{-\frac{2\beta}{\zeta} \langle X_2 \rangle }\e^{-\frac{2\beta}{\zeta} \langle X_1 \rangle }   \left[H^E,X_2\right]  \left(H^E-z \right)^{-1} \e^{\frac{2\beta}{\zeta} \langle X_2 \rangle } \e^{\frac{2\beta}{\zeta} \langle X_1 \rangle } \right) \\
    & \quad \quad \cdot\left(\e^{-\frac{2\beta}{\zeta} \langle X_1 \rangle } \e^{-\frac{2\beta}{\zeta} \langle X_2 \rangle } \left(H^E-\iu \right)^{-2}\right) := ABCDE .
    \end{aligned}
    $$
    Using the triangle inequality and the exponential localization given by \eqref{IntKernelResolvent2}, one can show that the first four factors are bounded and have a norm which grows at most polynomially in $\zeta$. It is thus sufficient to investigate the trace norm of $E$. Rewrite $E$ as:
    $$E=
    \left(\e^{-\frac{\beta}{\zeta} \langle X_1 \rangle } \e^{-\frac{\beta}{\zeta} \langle X_2 \rangle }\e^{-\frac{\beta}{\zeta} \langle X_1 \rangle }  \e^{-\frac{\beta}{\zeta} \langle X_2 \rangle } \left(H^E-\iu \right)^{-1} \e^{\frac{\beta}{\zeta}\langle X_1 \rangle } \e^{\frac{\beta}{\zeta} \langle X_2 \rangle } \right) \cdot \left(\e^{-\frac{\beta}{\zeta} \langle X_1 \rangle } \e^{-\frac{\beta}{\zeta} \langle X_2 \rangle } \left(H^E-\iu \right)^{-1}\right).
    $$
    An argument similar to \eqref{eq:aux2TraceClass} shows that $E$ is trace class with a trace class norm that grows at most polynomially in $\zeta$. Finally, an estimate like \eqref{eq:aux9} concludes the proof that  $\mathcal{I}_{\omega}(b)$ is trace class.
    
    Point (ii) of this proposition is proved (among other things) in Lemma \ref{lemma:IntTrClassEstimateF}. 
    
    Now let us prove point (iii). From the proof of point (i) we see that we can write $\mathcal{I}_{\omega,b}$ as a finite sum of products of Hilbert-Schmidt operators of the type $\mathcal{I}_{\omega,b}=\sum_{j=1}^M A_{j,\omega} B_{j,\omega}.$ This implies:
    $$I_b(x_1,x_2)=\mathbb{E}\left(\sum_{j=1}^M \int_{E} \mathrm{d}\y \,A_{j,\dotw}(x_1,x_2;\y) B_{j,\dotw} (\y; x_1,x_2)\right). $$
    Note that this function is supported on the strip $\mathcal{S}_\infty$. After a  double application of the Cauchy-Schwarz inequality with respect to both $\x$ and $\y$ and using that the Hilbert-Schmidt norms of $A_{j,\omega}$ and $B_{j,\omega}$ are uniformly bounded in $\omega$, the proof follows.   \hfill \qedsymbol

    \medskip

\section{Proof of Lemma~\ref{AR2}}\label{sect5}

We start with a preliminary technical lemma regarding a generalized trace cyclicity in our setting. 

\begin{lemma}
	\label{lemma:cycle}
	Let $\left\{A_{\omega}\right\}_{\omega \in \Par},\left\{B_{\omega}\right\}_{\omega \in \Par}$ be two families of integral operators with integral kernels that are jointly continuous in $\R^4 \setminus \{(\x;\x), \;\; \x \in \R^2\}$ and which satisfy either \eqref{IntKRes1} or \eqref{IntKernelResolvent2}. 
	
	Let $C_\omega$ denote either $A_\omega$ or $B_\omega$ and let $b\in \R$. 
	
	\noindent {\rm (i)} Assume that 
	\begin{equation}
    \label{eq:covariance1DAux}
    \tau_{b,\gamma}C_{\omega}\tau_{b,-\gamma}=C_{T(\gamma_1,0)\omega}, \quad \forall \, (\gamma_1,0) \in \Z^2.
    \end{equation}
	Then 
	for every $L>0$
	\begin{align}\label{dc24}
	\E \left(\int_{[0,1]\times \R_+ }\di \x \int_{\R \times [0,L]} \di \y A_{\dotw}(\x;\y) B_{\dotw}(\y;\x)\right)=\E \left(\int_{[0,1]\times [0,L]}\di \y \int_{\R \times \R_+} \di \x  B_{\dotw}(\y;\x)A_{\dotw}(\x;\y)\right) .
	\end{align}

	\noindent {\rm (ii)} Assume that 
	\begin{equation}
    \label{eq:covariance2DAux}
    \tau_{b,\gamma}C_{\omega}\tau_{b,-\gamma}=C_{T(\gamma)\omega},\quad \forall \, \gamma \in \Z^2.
    \end{equation}
	Let $\Omega$ be the unit square. Then both $\chi_{\Omega} A_{\omega} B_{\omega}$ and $\chi_{\Omega} B_{\omega} A_{\omega}$  are trace class and
	\begin{equation}\label{dc28}
	\E\left(\Tr\left(\chi_{\Omega}A_{\dotw}B_{\dotw}\right)\right)=\E\left(\Tr\left(\chi_{\Omega}B_{\dotw}A_{\dotw}\right)\right).
	\end{equation}
	
	\end{lemma}
	\begin{proof}
	Let us start with case (i). Since both families of operators satisfy  \eqref{eq:covariance1DAux}, we have that 
	\begin{equation}
	\label{eq:auxC1}
	A_{\omega}(\x-{ (\gamma_1,0)};\y-{ (\gamma_1,0)}) B_{\omega}(\y-{ (\gamma_1,0)};\x-{ (\gamma_1,0)})=A_{T_{\gamma}\omega}(\x;\y) B_{T_{\gamma}\omega}(\y;\x),\,  \forall  \,  (\gamma_1,0) \in \Z^2.
	\end{equation}
	
	Decomposing $\R=\bigcup_{\gamma_1 \in \Z}\left([0,1] + \gamma_1\right)$ and considering \eqref{eq:auxC1}, we get
	\begin{equation}
	\label{eq:auxC100}
	\begin{aligned}
	&\E\left(\int_{[0,1]\times \R_+}\di \x \sum_{\gamma_1 \in \Z}\int_{0}^{1}\di y_1 \int_{0}^L \di y_2   A_{\dotw}(\x;y_1-\gamma_1,y_2) B_{\dotw}(y_1-\gamma_1,y_2;\x)\right)\\
	&= \E\left(\int_{[0,1]\times \R_+}\di \x \sum_{\gamma_1 \in \Z}\int_{0}^{1}\di y_1 \int_{0}^L \di y_2 A_{T_{\gamma}\dotw}(x_1+\gamma_1,x_2;y_1,y_2) B_{T_{\gamma}\dotw}(y_1,y_2;x_1+\gamma_1,x_2)\right) \\
	&=\E\left(\int_{[0,1]\times \R_+}\di \x \sum_{\gamma_1 \in \Z}\int_{0}^{1}\di y_1 \int_{0}^L \di y_2 B_{\dotw}(y_1,y_2;x_1+\gamma_1,x_2) A_{\dotw}(x_1+\gamma_1,x_2;y_1,y_2) \right)\\
	&=\E\left( \int_{0}^1 \di y_1 \int_0^L \di y_2  \int_{\R \times \R_+} \di \x  B_{\dotw}(y_1,y_2;\x) A_{\dotw}(\x;y_1,y_2) \right)
	\end{aligned}
	\end{equation}
	where first we have used the exponential decay of the integral kernels to exchange the integral and the expectation over the disorder configurations, and then we have used that the expectation is invariant over the action of the translation group.

	Case (ii) can be proved in an analogous way. Indeed, the exponential decay of the integral kernels assures that $\chi_{\Omega} A_{\omega} B_{\omega}$ and $\chi_{\Omega} B_{\omega} A_{\omega}$ are trace class (one can write them as a product of two  Hilbert-Schmidt operators). Since they have a jointly continuous integral kernel we have that 
	$$
	\Tr(\chi_{\Omega} A_{\omega} B_{\omega}) = \int_{\Omega} \mathrm{d}\xx \int_{\R^2} \mathrm{d}\y A_{\omega}(\xx;\y) B_{\omega}(\y;\xx).
	$$
	From \eqref{eq:covariance2DAux} we have that 
	\begin{equation}
	\label{aux:Tinv}
	A_{\omega}(\xx-\gamma;\y-\gamma) B_{\omega}(\y-\gamma;\xx-\gamma)=A_{T_{\gamma}\omega}(\xx;\y) B_{T_{\gamma}\omega}(\y;\xx) \qquad \forall \gamma \in \Z^2.
	\end{equation}
	Then, decomposing $\R^2= \bigcup_{\gamma \in \Z^2}\left(\Omega + \gamma\right)$ and considering \eqref{aux:Tinv}, we can conclude in the same way as we have done in \eqref{eq:auxC1}, that is
	$$
	\begin{aligned}
	&\E\left(\Tr(\chi_{\Omega} A_{\dotw} B_{\dotw})\right) = \E\left(\int_{\Omega}\mathrm{d}\xx \sum_{\gamma \in \Z^2} \int_{\Omega} \mathrm{d}\yy A_{\dotw}(\xx;\yy - \gamma) B_{\dotw}(\yy-\gamma;\xx)\right) \\
	&=\E\left(\int_{\Omega} \mathrm{d}\xx \sum_{\gamma \in \Z^2}  \int_{\Omega} \mathrm{d}\yy A_{T_{\gamma}\dotw}(\xx+\gamma;\yy) B_{T_{\gamma}\dotw}(\yy;\xx+\gamma) \right) \\
	&\int_{\Omega}\mathrm{d}\yy  \sum_{\gamma \in \Z^2} \E\left( \int_{\Omega} \mathrm{d}\xx B_{\dotw}(\yy;\xx+\gamma)  A_{\dotw}(\xx+\gamma;\yy) \right) = \E\left(\Tr(\chi_{\Omega} B_{\dotw} A_{\dotw} ) \right).
	\end{aligned}
	$$
	\end{proof}

	\begin{proof}[Proof of Lemma~\ref{AR2}]
	The strategy is to reduce the statement to the setting of \cite[Proposition~3]{SchulzBaldesTeufel13}, in which the result is proved for bounded Hamiltonians.
	
	Let $a<0$ with $a/4 <\inf \Sigma(0)$. We may assume that $F$ is a Schwartz function supported in $(a/3,\infty)$. Moreover, let us use the shorthand notation $H_{\omega,b}=H_\omega$, since the dependence on $b$ is not playing any role here. We have that $(x-a)^{N}F(x)$ is a bounded function, which implies by functional calculus that 
	$$ \left(H_{\omega}-a\right)^{N} F(H_{\omega})$$
	is a bounded operator for all $N$. 
	
	 Any operator of the type $\chi_\Omega (H_\omega-a)^{-N} X_j H_\omega$ or $\chi_\Omega H_\omega X_j  (H_\omega-a)^{-N}$ is trace class for all $N\geq 3$; this can be proved using the same idea of sending some exponential decay over the resolvents and using that $\chi_\Omega$ is compactly supported, see the proof of Lemma~\ref{AR1} for more details. From Lemma~\ref{lemma:cycle} we get
	\begin{align*}
	&\E\left({\rm Tr}\left(\chi_\Omega  \iu \left[H_{\dotw},X_i\right]F(H_{\dotw})\right) \right )= \E\left({\rm Tr}\left (\chi_\Omega  \iu \left[H_{\dotw},X_i\right]\left(H_{\dotw}-a\right)^{-(N+1)} \left(H_{\dotw}-a\right)^{N+1} F(H_{\dotw})\right )\right) \\
	&= \frac{1}{N}\E\left({\rm Tr}\left (\chi_\Omega  \iu \left[\left(H_{\dotw}-a\right)^{-N},X_i\right] \left(H_{\dotw}-a\right)^{N+1} F(H_{\dotw})\right) \right).
	\end{align*}
	Define the Schwartz function $f(x):=(x-a)^{N+1} F(x)$, and $f (H_{\omega}) = \left(H_{\omega}-a\right)^{N+1} F(H_{\omega}).$

	The function $f$ can be uniformly approximated by functions that are smooth and with compact support in  $(a/2,\infty)$, that is there exists $\{g_n\}_{n \in \N}$ such that $g_n \in C^{\infty}_0( (a/2,\infty))$ and $$\lim_{n \to \infty} \sup_{x \in \R} |g_n(x)-f(x)| =0.$$ If $N \geq 2$, one can prove as before that $\chi_\Omega  \iu [\left(H_{\omega}-a\right)^{-N},X_i]$ is a trace class operator, hence we have
	$$
	\left|{\rm Tr}\left(\chi_\Omega  \iu \left[\left(H_{\omega}-a\right)^{-N},X_i\right] \left(g_n(H_{\omega})- f(H_{\omega})\right)\right)\right| \leq \left\|\chi_\Omega  \iu \left[\left(H_{\omega}-a\right)^{-N},X_i\right] \right\|_{1} \left\|g_n(H_{\omega})-f(H_{\omega})\right\| ,
	$$
    which means that is enough to prove 
    $$
    {\rm Tr}\left(\chi_\Omega  \iu [\left(H_{\omega}-a\right)^{-N},X_i] g(H_{\omega})\right)=0 \qquad \forall \, g \in \C^{\infty}_0( (a/2,\infty)).
    $$
    
    Assume that $N\geq 2$. Define $\tilde{g}(x):=g(x^{-1/N}+a)$ if $x>0$. We see that $\tilde{g}(x)=0$ if $x$ is close enough to $0$, hence we may smoothly extend it by zero to  $(-\infty,0]$. If $x$ becomes large then $\tilde{g}$ is again zero because $a<a/2<0$. Hence $\tilde{g}\in C_0^\infty(\R)$. Then  $\tilde{g}\left(\left(H_{\omega}-a\right)^{-N}\right)=g(H_{\omega})$ and the operator $\left(H_{\omega}-a\right)^{-N}=:\mathfrak{H}$ can be seen as a bounded Hamiltonian. Then we would like to prove 
    $$ {\rm Tr}\left(\chi_\Omega  \iu \left[\left(H_{\omega}-a\right)^{-N},X_i\right] g(H_{\omega})\right)=
    {\rm Tr}\left(\chi_\Omega  \iu [\mathfrak{H},X_i] \tilde{g}(\mathfrak{H})\right)=0,
    $$
    where the second equality can be obtained by mimicking the method in \cite[Proposition 3]{SchulzBaldesTeufel13}. The idea is as follows: since $\chi_\Omega[\mathfrak{H},X_i]$ is trace class and $\tilde{g}$ has compact support, then up to a use of the Stone-Weierstrass theorem it is enough to prove the identity when $\tilde{g}(\mathfrak{H})$ is replaced by powers of $\mathfrak{H}$. If $M \in \N$, by writing  
    $$\chi_\Omega \mathfrak{H}^M X_i-\chi_\Omega X_i \mathfrak{H}^M=\sum_{j=0}^{M-1}\chi_\Omega  \mathfrak{H}^j[\mathfrak{H},X_i]\mathfrak{H}^{M-j}$$
    where each operator is trace class, we may use the usual trace ciclicity on the left, and \eqref{dc28} on the right to get
    $$0=M\; \mathbb{E}\big ( \Tr\left(\chi_\Omega  [\mathfrak{H},X_i]\mathfrak{H}^M\right)\big ),\quad \forall M\,\geq 1.$$
Furthermore, since the family $(H_{\omega,b})_{\omega \in \Par}$ is ergodic with respect to the lattice $\Z^2$, by applying Birkhoff's ergodic theorem we can conclude from \eqref{eq:EqNoCurrent}  that for almost all $\omega \in \Par$ the limit in \eqref{eq:TDLcurrent} holds true.

	\end{proof}

\section{Proof of Theorem \ref{thm-positive}}
\label{sect6}
The proof can be divided in three main steps. 
\subsection{Step 1: Edge integrated density of states}

Let $\chi_{L}$ denote as before the characteristic function of the set $[0,1]\times [0,L]$. For every $b$ and $L>0$, we have
\begin{equation}
\label{eq:DefNLepsilon}
\rho_{L,\omega}(b)= \frac{1}{L} \Tr \left(\chi_{L}  F(H^E_{\omega,b})\right)= \frac{1}{L} \int_{0}^{1} \mathrm{d} x_1  \int_{0}^{L} \mathrm{d} x_2  \, F(H^E_{\omega,b})(x_1,x_2;x_1,x_2).
\end{equation}

\begin{proposition}
For almost every $\omega \in \Par$ we have that
\label{step1}
\begin{equation}
\label{step1:mainEq}
\lim_{L \to \infty} \rho_{L,\omega}(b) =B_F(b) \, .
\end{equation}
\end{proposition}
\begin{proof}
The idea of the proof is rather standard and consists in showing that by substituting the integral kernel of the resolvent of the edge Hamiltonian with the integral kernel of the resolvent of the bulk Hamiltonian in the trace in \eqref{eq:DefNLepsilon}, the error goes to zero with $L$. We note that by employing the same methods used in the proof of Lemma~\ref{lemma:IntTrClassEstimateF}, one can show that all the operators appearing below are trace class. Using \eqref{gpt1}, we can decompose the trace in \eqref{step1:mainEq} in two terms:
\begin{equation}\label{dc50}
\begin{aligned}
\Tr \left(\chi_{L} F(H^E_{\omega,b})\right)&= -\frac{1}{\pi} \Tr\left(\chi_{L}\int_{\mathcal{D}} \mathrm{d}z_1\mathrm{d}z_2 \, \bar{\partial} {F_N}(z) U_{L,\omega}(z)\right) \\
&\frac{1}{\pi}\Tr\left(\chi_{L} \int_{\mathcal{D}} \mathrm{d}z_1\mathrm{d}z_2 \, \bar{\partial} {F_N}(z)  \left(H_{\omega,b} - z \right)^{-1}W_{L,\omega}(z)\right) \\
&=: A_{1,\omega}(L) + A_{2,\omega}(L).
\end{aligned}
\end{equation}
Use \eqref{gpt1} in order to write 
$$A_{1,\omega}(L)= \chi_{L} \widetilde{\eta}_{0} F(H^E_{\omega,b}) \eta_{0} + \chi_{L}\widetilde{\eta}_{L} F(H_{\omega,b}) \eta_{L}.$$
From \eqref{eq:PolLocalizationF} and \eqref{eq:defEta} we get
$$
\left|  \frac{1}{L}\Tr \left(\eta_0 \chi_{L} F\left(H^E_{\omega,b}\right) \right) \right| \leq \frac{C}{\sqrt{L}},
$$
where $C$ is independent of $\omega$. Thus this term will go to zero when $L\to \infty$.

Considering the term with $\widetilde{\eta}_L$ 
in $A_{1,\omega}(L)$ we obtain 
$$
\begin{aligned}
\Tr \left(\eta_L \chi_{L} F(H_{\omega,b}) \right) & =  \int_{0}^{1} \mathrm{d} x_1 \int_{0}^{{2\sqrt{L}}} \mathrm{d} x_2  \big (\eta_{L}(\x)-1\big ) F\left(H_{\omega,b}\right)(\x;\x) \\
& \phantom{=}+  \int_{0}^{1} \mathrm{d} x_1 \int_{0}^{L} \mathrm{d} x_2  F(H_{\omega,b})(\x;\x) =: A_{{11,\omega}}(L) + A_{{12,\omega}}(L) .
\end{aligned}
$$
Due to \eqref{eq:PolLocalizationF}, we have that $|A_{11,\omega}(L)|\leq C \sqrt{L}$, where again $C$ is a constant that does not depend on $\omega$. Using the covariance of the integral kernel of $F(H_{\omega,b})$ and Birkhoff's ergodic theorem coupled with the isotropic ergodicity of the disorder, we obtain that for almost every $\omega$
$$
\lim_{L \to \infty} \frac{1}{ L} A_{1,\omega}(L) =
\E\left( \int_{\Omega} {\mathrm{d}} \xx \, F(H_{\dotw,b})(\xx;\xx) \right)=B_F(b) . 
$$
To conclude the proof, it is enough to show that $\left|A_{2,\omega}(L)\right|$ is a boundary term that becomes negligible in the limit $L\to \infty$. Indeed, considering the exponential space localization and the fast decay in $L$ of the integral kernel of $W_{L,\omega}(z)$, see \eqref{eq:WKern}, and reusing the method of Lemma~\ref{lemma:IntTrClassEstimateF} which transforms the localization of the integral kernels near the diagonal into trace norm estimates, we obtain that for every $M>0$ there exists $C_M$ such that
$$
\left|\Tr\left(\chi_{L} \int_{\mathcal{D}} \mathrm{d}z_1\mathrm{d}z_2 \, \bar{\partial} {F_N}(z)  \left(H^E_{\omega,b} - z \right)^{-1} W_{L,\omega}(z)\right)\right| \leq C_M L^{-M}.
$$
Hence this term decays faster than any power of $L$ and the proof is over.
\end{proof}

\medskip

\subsection{Step 2: The thermodynamic limit of the magnetic derivative}

\begin{proposition}
\label{prop:TDLMagneticD}
The functions $\rho_{L,\omega}(\cdot)$ and $B_F(\cdot)$  are differentiable at every $b\in \R$ and for almost every $\omega\in \Par$ we have: 
\begin{equation}
\label{eq:aux10}
\lim_{L \to \infty} \frac{\mathrm{d} \rho_{L,\omega}}{\mathrm{d} b} (b)=\lim_{L \to \infty} \E\left(\frac{\mathrm{d}\rho_{L,\dotw}}{\mathrm{d} b} (b)\right) = B'_{F}(b) .
\end{equation}
\end{proposition}

\begin{proof}
Let us fix $b\in \R$ and show that \eqref{eq:aux10} holds true. In order to shorten the next formulas we denote by  $P_{y_i}$, $i \in \left\{1,2\right\}$ the operator $\left(-\iu\frac{\partial}{\partial y_i} - b A_i- \mathcal{A}_i \right)$ and we introduce the shortcuts $R^E_{\omega,z}:=
(H^E_{\omega,b}-z)^{-1}$ and $R_{\omega,z}:=(H_{\omega,b}-z)^{-1}$.  We start by showing that the following two expansions in $\epsilon$, with $|\epsilon|<1$, hold true: 
	\begin{equation}
	\label{eq:expansionNL}
	\begin{aligned}
	&\rho_{L,\omega}(b+\epsilon)= \rho_{L,\omega}(b)\\
	&+\frac{\epsilon}{\pi L}  \int_{0}^1 \di x_1\int_0^{L}  \di x_2 \int_{\mathcal{D}} \mathrm{d}z_1\mathrm{d}z_2 \,  \bar{\partial} {F_N}(z) \int_{E} \mathrm{d}\mathrm{\y} R^E_{\omega,z}(\x;\y) \left(y_2-x_2\right) P_{y_1} R^E_{\omega,z}(\y;\x)  \\
	& + \mathcal{O}(\epsilon^2),
	\end{aligned}
	\end{equation}
	(where the remainder  $\mathcal{O}(\epsilon^2)$  {\it is uniformly bounded in $L$ and $\omega$}),
	and 
	\begin{equation}
	\label{eq:expansionRhoF}
	\begin{aligned}
	&B_F(b+\epsilon)= B_F(b)\\
	&+ \frac{\epsilon}{\pi } \E\left( \int_{\Omega_0} \mathrm{d}\xx  \int_{\mathcal{D}} \mathrm{d}z_1\mathrm{d}z_2 \,  \bar{\partial} {F_N}(z) \int_{\R^2} \mathrm{d}\mathrm{\y} R_{\dotw,z}(\xx;\y) \left(y_2-\xx_2\right) P_{y_1} R_{\dotw,z}(\y;\xx) \right)   + \mathcal{O}(\epsilon^2).
	\end{aligned}
	\end{equation}

	The main tool will be the magnetic perturbation theory \cite{CorneanNenciu1998,CorneanNenciu09,CorneanMonacoMoscolari2018}, but adapted to operators defined on domains and with a modified asymmetric magnetic phase, as we will explain next. 
	
	Let 
	\begin{equation}\label{dc40}
	    \varphi_2(\x,\y):=(y_1-x_1)y_2,\quad \nabla_{\x}\varphi_2(\x,\y)=(-y_2,0),\quad \varphi_2(\x,\x)=0.
	\end{equation}
	This phase has two important properties. First, it commutes with $P_{x_2}$ because it is independent of $x_2$, and 
	\begin{equation}\label{dc41}
	  \big (P_{x_1}-\epsilon A_1(\x)\big )  \e^{\iu \epsilon \varphi_2(\x,\y)}=\e^{\iu \epsilon \varphi_2(\x,\y)}\big (P_{x_1}-\epsilon A_1(\x-\y)\big ).
	\end{equation}
	Second, it obeys the composition rule 
	\begin{equation}\label{dc42}
	\begin{aligned}
	  &\Phi_2(\x,\x',\y):=\varphi_2(\x,\x')+\varphi_2(\x',\y)-\varphi_2(\x,\y)=(x_1-x_1')(y_2-x_2'),\\
	  &|\Phi_2(\x,\x',\y)|\leq \Vert \x-\x'\Vert \; \Vert \x'-\y\Vert. 
	  \end{aligned}
	\end{equation}
	
	Let $S^E_{\omega,\epsilon,z}$ be the operator with the integral kernel 
	$$S^E_{\omega,\epsilon,z}(\x;\y):=\e^{\iu \epsilon \varphi_2(\x,\y)} R^E_{\omega,z}(\x;\y). $$
	Using the estimate in \eqref{IntKRes1} and the Schur test one can show that this operator is bounded. Moreover, its range lies in the domain of $H^E_{\omega,b+\epsilon}$ and we have the identity (here \eqref{dc41} plays a crucial role)
	\begin{align*}
	(H^E_{\omega,b+\epsilon}-z)  S^E_{\omega,\epsilon,z}=1 + T^E_{\omega,\epsilon,z},  
	\end{align*}
	where the operator $T^E_{\omega,\epsilon,z}$ has an integral kernel given by
	$$T^E_{\omega,\epsilon,z}(\x;\y)=\e^{\iu \epsilon \varphi_2(\x,\y)}\Big (-\epsilon A_1(\x-\y)P_{x_1} R^E_{\omega,z}(\x;\y) +2^{-1}\epsilon^2 |A_1(\x-\y)|^2R^E_{\omega,z}(\x;\y) \Big ). $$
	This kernel is strongly localized around the diagonal and defines a bounded operator. Thus we may write a variant of the ``second resolvent identity" as:
	\begin{align*}
	(H^E_{\omega,b+\epsilon}-z)^{-1}&= S^E_{\omega,\epsilon,z}-(H^E_{\omega,b+\epsilon}-z)^{-1}T^E_{\omega,\epsilon,z}\\
	&=S^E_{\omega,\epsilon,z}-S^E_{\omega,\epsilon,z}T^E_{\omega,\epsilon,z}+(H^E_{\omega,b+\epsilon}-z)^{-1}\left(T^E_{\omega,\epsilon,z}\right)^2.
	\end{align*}
	We are interested in writing a ``nice" expansion up to errors of order $\epsilon^2$ for the integral kernel of the resolvent on the left-hand side. Using the composition rule in \eqref{dc42} we may write:
	\begin{equation}
	\label{eq:aux1}
	\begin{aligned}
	&(H^E_{\omega,b+\epsilon}-z)^{-1}(\x;\y)=\e^{\iu \epsilon \varphi_2(\x,\y)} R^E_{\omega,z}(\x;\y) \\
	& - \epsilon \, \e^{\iu \epsilon \varphi_2(\x,\y)} \int_{E} \mathrm{d}\tilde{\x}  R^E_{\omega,z}(\x;\tilde{\x})  \e^{\iu \epsilon \Phi_2(\x,\tilde{\x},\y)} \left(\tilde{x}_2-y_2\right) P_{\tilde{x}_1} R^E_{\omega,z}(\tilde{\x};\y)+ \Reminder^E_{\omega,\epsilon,z}(\x;\y) .
	\end{aligned} 
	\end{equation}	Before analyzing the remainder $\Reminder^E_{\omega,\epsilon,z}(\x;\y)$, let us focus on the first two terms. Substituting the previous equation in \eqref{eq:DefNLepsilon}, and using that $\varphi_2(\x,\x)=0$, we get that the first term on the right-hand side of \eqref{eq:aux1} gives rise to $\rho_{L,\omega}(b)$. In the second term of \eqref{eq:aux1} we may replace $e^{\iu \epsilon \Phi_2(\x,\tilde{\x},\y)}$ with $1$ at a price of an error of order $\epsilon \Vert \x-\tilde{\x}\Vert \; \Vert \tilde{\x}-\y\Vert$. But polynomial growth in either $\Vert \x-\tilde{\x}\Vert $ or $\Vert \tilde{\x}-\y\Vert$ is taken care of by the exponential decay of the nearby resolvents, at a price of some polynomially growing factors in $\langle z_1 \rangle/|z_2|$, which are taken care of by $\overline{\partial}F_N$ as in \eqref{dc27}. To conclude, the only first order contribution to $\rho_{L,\omega}(b+\epsilon)$ coming from the first two terms on the right-hand side of \eqref{eq:aux1} is
	$$
	 \frac{\epsilon}{\pi L} \int_{0}^1 \di x_1\int_{0}^{L} \di x_2 \int_{\mathcal{D}} \mathrm{d}z_1\mathrm{d}z_2 \, \bar{\partial} {F_N}(z) \int_{E} \mathrm{d}\tilde{\x}  R^E_{\omega,z}(\x;\tilde{\x}) \left(\tilde{x}_2-{x}_2\right) \left(P_{\tilde{x}_1} R^E_{\omega,z}\right)(\tilde{\x};\x) ,
	$$ 
	which is exactly the second term in \eqref{eq:expansionNL}.
	
	Let us now focus on the remainder, which also contains several terms. One of these terms is
	$$
	\begin{aligned}
 \frac{\epsilon^2}{2} \int_{E}\mathrm{d}\tilde{\x} \, \e^{\iu \epsilon \varphi_2(\x,\tilde{\x})} R^E_{\omega,z}(\x;\tilde{\x}) \e^{\iu \epsilon \varphi_2(\tilde{\x},\y)} |A_1(\tilde{\x}-\y)|^2R^E_{\omega,z}(\tilde{\x},\y).
	\end{aligned}
	$$
The growth of $(\tilde{x}_2-y_2)^2$ contained in the vector potential is again taken care of the exponential localization of the nearby resolvent. By using the pointwise estimates in Appendix~\ref{appendix:IntKernels} one can show that the above term is pointwise bounded by $\epsilon^2$ times a constant which is  uniform in $\x$, $\y$ and $\omega$, but which grows polynomially in $\langle z_1 \rangle/|z_2|$. This last growth is again taken care of by $\overline{\partial}F_N$ when one integrates over $\mathcal{D}$. Thus, its contribution to \eqref{eq:expansionNL} will be of order $\epsilon^2$ with a constant which is uniform in $L$ and $\omega$. The other terms of the remainder can be treated in a similar way, and we do not give further details.

	In particular, \eqref{eq:expansionNL} also shows that  $\rho_{L,\omega}$ is differentiable and its derivative is given by
	\begin{align}
	\label{eq:derivativeNL}
	&\frac{\mathrm{d}\rho_{L,\omega}}{\mathrm{d}b}(b)\\
	&=  \frac{1}{\pi L} \int_{0}^1 \di x_1 \int_0^{L} \di x_2\;  \int_{\mathcal{D}} \mathrm{d}z_1\mathrm{d}z_2\, \bar{\partial} {F_N}(z) \int_{E} \mathrm{d}\mathrm{\y} R^E_{\omega,z}(\x;\y) \left(y_2-x_2\right) P_{y_1} R^E_{\omega,z}(\y;\x) . \nonumber 
	\end{align}
	The next step is to prove that we can replace the resolvent of the edge Hamiltonian with the resolvent of the bulk Hamiltonian using \eqref{gpt1} by showing that the error terms goe to zero as $L$ goes to $+\infty$. The terms that contain $\eta_{0}$ and $\widetilde{\eta}_0$, together with their derivatives, give a contribution of order at most ${L^{-1/2}}$. Moreover, the terms sandwiched between derivatives of $\widetilde{\eta}_{L}$ and $\eta_L$ decay faster than any power of $L$ by an estimate analogue to \eqref{eq:WKern} (notice that $\widetilde{\eta}_L(\x)=\widetilde{\eta}_L(x_2)$).

	Thus, the only term not going to zero is the one generated by $\widetilde{\eta}_{L} \left(H_{\omega} - z \right)^{-1} \eta_{L}$.  More precisely, using that $\widetilde{\eta}_{L} {\eta}_{L}={\eta}_{L}$:
	$$
	\begin{aligned}
	&\frac{\mathrm{d}\rho_{L,\omega}}{\mathrm{d}b}(b)\\
	&= \frac{1}{\pi L} \int_{0}^1 \di x_1 \int_0^{L} \di x_2\;  \int_{\mathcal{D}} \mathrm{d}z_1\mathrm{d}z_2 \bar{\partial} {F_N}(z) \int_{E} \mathrm{d}\mathrm{\y} \eta_{L}(\x) \eta_{L}(\y) R_{\omega,z}(\x;\y)  \left(y_2-x_2\right) P_{y_1} R_{\omega,z}(\y;\x)\\
	& \quad + \mathcal{O}\left(L^{-1/2}\right),
	\end{aligned}
	$$
    where the remainder is also uniformly bounded in $\omega$.
	Consider the following integral kernel
	\begin{equation}
	\label{eq:aux6}
	 R_{\omega,z}(\x;\y)  \left(y_2-x_2\right) \left(P_{y_1} R_{\omega,z}\right)(\y;\x) .
	\end{equation}
	Using the exponential localization in $\Vert \x-\y\Vert$ of both factors (see \eqref{IntKRes1} and \eqref{IntKernelResolvent2}), the error we make by replacing both $\eta_L(\x)$ and $\eta_L(\y)$ with $1$ is of order $L^{-1/2}$ and uniformly bounded in $\omega$. Thus:
    \begin{equation}
    \label{eq:aux7}
    \begin{aligned}
    &\frac{\mathrm{d}\rho_{L,\omega}}{\mathrm{d}b}(b)\\
    &= \frac{1}{\pi L} \int_{0}^1\di x_1 \int_0^{L} \di x_2\;  \int_{\mathcal{D}} \mathrm{d}z_1\mathrm{d}z_2\, \bar{\partial} {F_N}(z) \int_{E} \mathrm{d}\mathrm{\y}  R_{\omega,z}(\x;\y)  \left(y_2-x_2\right) P_{y_1} R_{\omega,z}(\y;\x)\\
    &\quad + \mathcal{O}\left(L^{-1/2}\right).
    \end{aligned}
    \end{equation}
	We want to extend the $\y$ integral over the whole $\R^2$ plane while keeping $\x$ in $E$. In the extra contribution where $\y\in \R^2\setminus E$, we have $x_2>0>y_2$ and $|x_2-y_2|\geq  |x_2-y_2|/2 + x_2/2 $. Due to the exponential localization in $\Vert \x-\y\Vert $ of \eqref{eq:aux6}, the integrand of the extra term gains an exponentially localized term in $|x_2|$, which makes the integral over $x_2$ to be bounded. Thus the error we get by extending the $\y$ integral over the whole of $\R^2$ is of order $L^{-1}$ and uniformly bounded in $\omega$, hence
	\begin{align}
	\label{eq:aux11}
	&\frac{\mathrm{d}\rho_{L,\omega}}{\mathrm{d}b}(b)\\
	&= \frac{1}{\pi L} \int_{0}^1 \di x_1 \int_0^{L} \di x_2\;  \int_{\mathcal{D}} \mathrm{d}z_1\mathrm{d}z_2 \, \bar{\partial} {F_N}(z) \int_{\R^2} \mathrm{d}\mathrm{\y}  R_{\omega,z}(\x;\y)  \left(y_2-x_2\right) P_{y_1} R_{\omega,z}(\y;\x) \nonumber \\
		&+ \mathcal{O}\left(L^{-1/2}\right) \nonumber.
	\end{align}
	Since the operator corresponding to the integral kernel from \eqref{eq:aux6}  satisfies the covariance relation, from \eqref{eq:aux11} we get that for almost every $\omega\in \Par$ we have
	\begin{equation}\label{dc51}
	\begin{aligned}
	&\lim_{L \to \infty} \frac{\mathrm{d}\rho_{L,\omega}}{\mathrm{d}b}(b)\\
	& = \frac{1}{\pi }\E\left(\int_{\Omega} \mathrm{d} \xx \int_{D} \mathrm{d}z_1\mathrm{d}z_2  \bar{\partial} {F_N}(z) \int_{\R^2} \mathrm{d}\mathrm{\y} \,  R_{\dotw,z}(\xx;\y)  \left(y_2-\xx_2\right) P_{y_1} R_{\dotw,z}(\y;\xx)\right).
	\end{aligned}
	\end{equation}
    Moreover, since the right-hand side of \eqref{eq:aux11} is also uniformly bounded in $L$ and $\omega$, we can exchange the limit $L\to \infty$ with the expectation with respect to $\omega$ by using Lebesgue dominated convergence theorem.

	The last step is to show that the above expression also equals $B_F'(b)$. One way is to prove it directly applying magnetic perturbation theory in exactly the same way but for the bulk operator. A much easier way is to exploit the fact that the remainder in $\epsilon^2$ from \eqref{eq:expansionNL} is uniform in $L$ and $\omega$, hence by choosing some $\omega$ for which both Proposition~\ref{step1} and \eqref{dc51} hold true we have 
	$$B_F(b+\epsilon)=B_F(b)+\epsilon \lim_{L\to\infty} \frac{\mathrm{d}\rho_{L,\omega}}{\mathrm{d}b}(b) +\mathcal{O}(\epsilon^2),$$
	which is exactly \eqref{eq:expansionRhoF}.
\end{proof}

\subsection{Step 3: {Connecting the magnetic derivative to the total edge current}}

We start with a technical lemma. 
	
	\begin{lemma}\label{CurrentIndipG}
	Let $g\in C^1([0,1])$ be such that  $g(0)=1$ and $g(1)=0$. Then the limit 
	 
\begin{align}
\label{eq:GCurrent}
\mathcal{L}_b:=	\lim_{L\to\infty} \E \left( \int_0^1 \mathrm{d}x_1 \int_{0}^{L} \mathrm{d}x_2\;  g\left (\frac{x_2}{L}\right ) \; \big\{\iu [H_{\dotw,b}^E,X_1]F'(H_{\dotw,b}^E)\big\}(x_1,x_2;x_1,x_2) \right)
	\end{align}

exists and is independent of $g$. 
	\end{lemma}
	
	\begin{proof}
	By using the ergodicity of $(H_{\omega,b})_{\omega \in \Par}$ and  Proposition~\ref{AR1:point2} we get that 
	$$K_b(x_1,x_2):=\E\left( \big\{\iu [H_{\dotw,b},X_1] F'(H_{\dotw,b}\big\}(x_1,x_2;x_1,x_2)\right)$$
is continuous and $\Z^2$-periodic. Moreover, due to Lemma \ref{AR2} we have
	\begin{align}\label{AG4}
	\int_0^1 \mathrm{d}x_1 \int_{0}^{1} \mathrm{d}x_2\; K_b(x_1,x_2)=0.
	\end{align}

Let us start by adding and subtracting to the right-hand side of \eqref{eq:GCurrent} the same expression with $H^E_{\dotw,b}$ substituted by $H_{\dotw,b}$. Then, by using \eqref{eq:currentL1} 
and the Lebesgue dominated convergence theorem, the lemma follows from the proof of the following identity:
	
\begin{align}\label{AG1}
	\lim_{L\to\infty} &\int_0^1 \mathrm{d}x_1 \int_{0}^{L} \mathrm{d}x_2\;  g\left (\frac{x_2}{L}\right ) K_b(x_1,x_2)
	=-\int_0^1 \mathrm{d}x_1 \int_{0}^{1} \mathrm{d}x_2\;  x_2\;   K_b(x_1,x_2). 
	\end{align}
	Let us now prove \eqref{AG1}. Let $N$ be the integer part of $L$. Because $g$ is continuous and $g(1)=0$ we have 
	\begin{align*}
	\int_0^1 \mathrm{d}x_1 \int_{0}^{L} \mathrm{d}x_2\;  g\left (\frac{x_2}{L}\right ) K_b(x_1,x_2)
	=\int_0^1 \mathrm{d}x_1 \int_{0}^{N} \mathrm{d}x_2\;  g\left (\frac{x_2}{L}\right ) K_b(x_1,x_2)+o(1).
	\end{align*}
	Using the periodicity of $K_b$ in the second variable we have
	\begin{align}\label{AG3}
	\int_0^1 \mathrm{d}x_1 & \int_{0}^{L} \mathrm{d}x_2\;  g\left (\frac{x_2}{L}\right ) K_b(x_1,x_2)=\sum_{k=0}^{N-1}\int_0^1 \mathrm{d}x_1 \int_{0}^{1} \mathrm{d}x_2\;  g\left (\frac{x_2+k}{L}\right ) K_b(x_1,x_2)+o(1).
	\end{align}
	Using \eqref{AG4},
	we may replace $g\left (\frac{x_2+k}{L}\right )$ in \eqref{AG3} with 
	$$g\left (\frac{x_2+k}{L}\right )-g\left (\frac{k}{L}\right )=\frac{x_2}{L}g'(t_k),\quad \frac{k}{L}< t_k< \frac{x_2+k}{L}\leq  \frac{1+k}{L}.$$
	We have the identity
	$$\lim_{L\to\infty}\sum_{k=0}^{N-1}\Big ( g\left (\frac{x_2+k}{L}\right )-g\left (\frac{k}{L}\right )\Big)=x_2\lim_{L\to\infty}\frac{1}{L}\sum_{k=0}^{N-1}g'(t_k)=x_2\int_0^1 g'(t) \di t=-x_2,$$
	which introduced in \eqref{AG3} proves \eqref{AG1}. 
	\end{proof}

    \begin{proposition} \label{prop:TDLMagneticEdge}
	If $\mathcal{L}_b$ is the limit from Lemma \ref{CurrentIndipG}, then we have:
	$$
	\lim_{L \to \infty} \E\left(\frac{\di \rho_{L,\dotw}}{\di b}(b)\right)=-\mathcal{L}_b.
	$$
	\end{proposition}

	\begin{proof}
    Recall the notation $\chi_{L}$ for the finite strip $[0,1]\times [0,L]$ and the shortcuts $H^\star_{\omega,b}=H^\star_{\omega}$, $R^\star_{\omega,z}=
(H^\star_{\omega,b}-z)^{-1}$ with $\star \in \left\{ E,\varnothing \right\}$ . Using trace cyclicity we have:
	\begin{equation}\label{dc43}
	\begin{aligned}
	0&= \Tr\left(\iu \left  [X_1,\chi_{L}
	F(H^E_{\omega})\chi_{L}\right]\right)\\
	&=-\frac{1}{\pi} \int_{\R^2} \mathrm{d}\mathrm{\x} \chi_{L}(\x) \int_{\mathcal{D}} \mathrm{d}z_1\mathrm{d}z_2 \,  \bar{\partial} {F_N}(z) \int_{E} \mathrm{d}\mathrm{\y}  R^E_{\omega,z}(\x;\y) P_{y_1} R^E_{\omega,z}(\y;\x),\\
	0&= \Tr\left(\iu \left [X_1,X_2 \chi_{L} F(H^E_{\omega})\chi_{L}\right]\right)  \\
	&=-\frac{1}{\pi}\int_{\R^2} \mathrm{d}\mathrm{\x} \chi_{L}(\x) \; x_2\int_{\mathcal{D}} \mathrm{d}z_1\mathrm{d}z_2 \,  \bar{\partial} {F_N}(z) \int_{E} \mathrm{d}\mathrm{\y}  R^E_{\omega,z}(\x;\y)  P_{y_1} R^E_{\omega,z}(\y;\x).
	\end{aligned}
	\end{equation} 
	 Use \eqref{dc43} in \eqref{eq:derivativeNL} by replacing  $x_2$ with $L$: 
	\begin{align*}
	&\frac{\di \rho_{L,\omega}}{\di b}(b)\\
	&\quad =\frac{1}{\pi L} \int_{0}^1 \di x_1\int_{0}^{L} \di x_2 \int_{\mathcal{D}} \mathrm{d}z_1\mathrm{d}z_2 \,  \bar{\partial} {F_N}(z) \int_{E} \mathrm{d}\mathrm{\y}  R^E_{\omega,z}(\x;\y) \left(y_2-L\right) P_{y_1} R^E_{\omega,z}(\y;\x) .
	\end{align*}

	Now we separate the contribution of the integration over $E$ into two parts: one on the strip $0\leq y_2\leq L$, and the other one on the region $y_2>L$. 
	
	When we integrate over the second region we have $x_2\leq L <y_2$ and 
	$$\Vert \x-\y\Vert\geq \Vert \x-\y\Vert /2 +|x_2-L|/2 +|y_2-L|/2.$$
	Using the exponential localization of the integral kernels in $\Vert \x-\y\Vert$ we can isolate an extra exponential decay in $|x_2-L|$ which makes the integral over $x_2$ convergent, and also an exponential decay in $|y_2-L|$ which takes care of the growth of the factor $y_2-L$. The price we pay is  some polynomial extra growth in $\langle z_1\rangle/|z_2|$. This is a similar argument with the one used in treating the integral of \eqref{eq:aux7}. Hence the contribution from $y_2>L$ is of order $L^{-1}$ and 
	$$
	\begin{aligned}
	&\frac{\di \rho_{L,\omega}}{\di b}(b)\\
	&=-\frac{1}{\pi L} \int_{0}^1 \di x_1\int_{0}^{L} \di x_2 \int_{\mathcal{D}} \mathrm{d}z_1\mathrm{d}z_2 \, \bar{\partial} {F_N}(z) \int_{\R} \mathrm{d}y_1 \int_{0}^{L} \mathrm{d}y_2 R^E_{\omega,z}(\x;\y) \left(L -y_2\right) P_{y_1} R^E_{\omega,z}(\y;\x)\\
	& + \mathcal{O}\left(L^{-1}\right).
	\end{aligned}
	$$
    Now the $y_2$ variable is so that $y_2\leq L$. Reversing the roles of $x_2$ and $y_2$ from the previous argument, we can extend the integral in $x_2$ over the region $x_2>L$,  making only an error of order $L^{-1}$. Thus we get
	\begin{align}
	\label{eq:auxiliary}
	&\frac{\di \rho_{L,\omega}}{\di b}(b)\\	&=-\frac{1}{\pi L} \int_{0}^1 \di x_1 \int_{0}^{\infty} \di x_2  \int_{\mathcal{D}} \mathrm{d}z_1\mathrm{d}z_2 \, \bar{\partial} {F_N}(z) \int_{\R} \mathrm{d}y_1 \int_{0}^{L} \mathrm{d}y_2 R^E_{\omega,z}(\x;\y) \left(L-y_2\right) P_{y_1} R^E_{\omega,z}(\y;\x)\nonumber \\
	& + \mathcal{O}\left(L^{-1}\right) .\nonumber 
	\end{align}
	Take the expectation $\mathbb{E}$ over the random variables. We now apply the ``partial trace cyclicity" from \eqref{dc24}. The integral over $\x$ produces the operator $(R^E_{\omega,z})^2$, which is turned into $F'(H^E_{\omega})$ by the integral over $z$, namely  
	$$
	\frac{1}{\pi}\int_{\mathcal{D}} \mathrm{d}z_1\mathrm{d}z_2 \, \bar{\partial} {F_N}(z) \int_{E} \mathrm{d}\mathrm{\x}  R^E_{\omega,z}(\x;\y)  \left(P_{y_1} R^E_{\omega,z}\right)(\y;\x) = \big ( P_{y_1}F'(H^E_{\omega,b})\big )(\y;\y).
	$$
	
	Thus we obtain: 
	\begin{equation*}
	\begin{aligned}
	&\E\left(\frac{\di \rho_{L,\dotw}}{\di b}(b)\right)=- \E\left(\int_{0}^{1} \mathrm{d}y_1 \int_{0}^{L} \mathrm{d}y_2 \left(1-\frac{y_2}{L}\right) \left\{\iu [H_{\dotw,b}^E,X_1] F'(H_{\dotw,b}^E)\right\} (\y;\y)\right)+ \mathcal{O} \left( L^{-1} \right). 
\end{aligned}
	\end{equation*}
	The function $g:[0,1]\mapsto [0,1]$ given by $g(x_2)=1-x_2$ belongs to $C^1([0,1])$ and obeys $g(0)=1$ and $g(1)=0$. Using this function in  Lemma \ref{CurrentIndipG}, the proof is over.  
	\end{proof}

\section{Proof of Proposition~\ref{prophc1}}\label{sect9}
	Let us assume that $F$ is the Fermi-Dirac distribution
	$$
	F_{\mu,T}(x)=\frac{1}{\e^{(x-\mu)/T}+1}
	$$
	where $\mu \in \R$ is the Fermi energy, and $T>0$ is the temperature in a suitable unit system. $F_{\mu,T}$ is not a Schwartz function, however, since the almost sure spectrum of the family of Hamiltonians $H_{\omega,b}$, that is $\Sigma(b)$, is bounded from  below by some $a<0$ such that $a+2\leq  \inf \Sigma(0)$, we can replace $F_{\mu,T}$ with the function $F_{\mu,T}(x)\varphi_{>a}(x)$, where $0\leq \varphi_{>a}\leq 1$ is a smooth cut-off function which equals $1$ if $x>a+1$, and $0$ if $x<a$. For simplicity of notation we will still denote the new function by $F_{\mu,T}$.

	Consider a function $\varphi \in C^\infty_0(\R)$ such that $\varphi=1$ on $[\mu-\delta,\mu+\delta]$ and $\varphi'$ is supported in $[\mu-2\delta,\mu+2\delta]\subset (E_-,E_+)$, for some $\delta>0$. To ease the notation, in the following we will omit the subscript $\mu$, that is $F_{\mu,T}=F_{T}$. Then we write
	\begin{equation}
	\label{eq:GapLinearity}
	F_{T}= F_{T}\varphi + F_{T}(1-\varphi)\chi_{(-\infty,\mu]} +  F_{T}(1-\varphi)\chi_{(\mu, + \infty)}=:F_T\varphi +F^{(<\mu)}+F^{(>\mu)},
	\end{equation}
	and notice that all the three terms on the right are Schwartz functions. Since the spectral gaps are preserved when $\epsilon:= |b-b_0|$ is small enough  \cite{Cornean10,CorneanMonacoMoscolari2018}, then  $(F_T\varphi)(H_{\omega,b})=0$ if $\epsilon$ is small enough. Also: 
	$$
	\begin{aligned}
	&B_{F_T}(b)=\E\left(\int_{\Omega} \mathrm{d}\xx \,  F^{(<\mu)} (H_{\dotw,b})(\xx;\xx)\right) + \E \left( \int_{{ \Omega}} \mathrm{d}\xx \, F^{(>\mu)} (H_{\dotw,b})(\xx;\xx)\right).
	\end{aligned}
	$$
	Using \eqref{eq:expansionRhoF} and the shorthand notation $\mathcal{R}_{\omega,z}:=\left(H_{\omega,b_0}-z\right)^{-1}$ we may write 
	\begin{equation}
	\label{eq:aux4}
	\begin{aligned}
	&\pi B'_{F_T}(b_0)\\
	& = \E \left( \int_{\Omega} \mathrm{d} \xx \int_{D} \mathrm{d}z_1\mathrm{d}z_2\bar{\partial} {F^{(<\mu)}_N}(z) \int_{\R^2} \mathrm{d}\mathrm{\y}  \mathcal{R}_{\dotw,z}(\xx;\y)  \left(y_2-x_2\right) \left(P_{y_1} \mathcal{R}_{\dotw,z}\right)(\y;\xx) \right) \\
	& \quad + \E \left(\int_{\Omega} \mathrm{d} \xx \int_{D} \mathrm{d}z_1\mathrm{d}z_2 \,\bar{\partial} {F^{(>\mu)}_N}(z) \int_{\R^2} \mathrm{d}\mathrm{\y}  \mathcal{R}_{\dotw,z}(\xx;\y)  \left(y_2-x_2\right)\left(P_{y_1} \mathcal{R}_{\dotw,z}\right)(\y;\xx)\right).
	\end{aligned}
	\end{equation}
	By using the pointwise decay estimates \eqref{IntKRes1} and \eqref{IntKernelResolvent2} for the integral kernels appearing in the above integrals with respect to $\y$, one can show that there exist some $M>0$ and a constant $C$ such that for all $\underline{x}$ and $\omega$ we have 
	\begin{equation}
	\label{eq:aux5}
	\left|\int_{\R^2} \mathrm{d}\mathrm{\y}  \mathcal{R}_{\omega,z}(\xx;\y)  \left(y_2-x_2\right)\left(P_{y_1} \mathcal{R}_{\omega,z}\right)(\y;\xx)\right| \leq C \frac{\langle z_1\rangle ^M}{|z_2|^M}. 
	\end{equation}
	The function $F^{(>\mu)}$ is supported in $[\mu+\delta,\infty)$, hence: 
	\begin{equation}
	\label{eq:FermiDerivatives}
	\sup_{|x-\mu|\geq \delta}\left|\frac{\partial^n F_{\mu,T}}{\partial x^n}(x)\right| \leq C_{n,\delta} \, \e^{- \frac{\delta}{3T}} \e^{- \frac{|x-\mu|}{3T}}.
	\end{equation}
	Let us remind the explicit construction of the almost analytic extension of $F^{>\mu}$. Let $0\leq g(y)\leq 1$ with $g\in C_0^\infty(\R)$ such that $g(y)=1$ if $|y|\leq 1/2$ and $g(y)=0$ if $|y|>1$. Fix some $N>M+2$, where $M$ is the number appearing in \eqref{eq:aux5}. Define 
	\begin{equation*}
	F^{(>\mu)}_N(x+\iu y) = g(y) \sum_{j=0}^N \frac{1}{j!} \frac{\partial^j F^{(>\mu)}}{\partial x^j}(x) (\iu y)^j.
	\end{equation*}
	Then for all $0<T\leq 1$ we have
	\begin{equation}
	\label{eq:auxleqmu}
	\sup_{z\in \mathcal{D}} \left| \bar{\partial} F^{(>\mu)}_N (z) \right| \leq C_{N,\mu,\delta} \frac{|z_2|^N}{\langle z_1\rangle ^N} \e^{-\frac{\delta}{4T}}.
	\end{equation}
	Hence, from \eqref{eq:auxleqmu} and \eqref{eq:aux5} we get:
	\begin{equation}
	\label{eq:zeroT1}
	\begin{aligned}
	&\left|\E \left( \int_{\Omega} \mathrm{d} \xx \int_{D} \mathrm{d}z_1\mathrm{d}z_2\,\bar{\partial} {F^{(>\mu)}_N}(z) \int_{\R^2} \mathrm{d}\mathrm{\y}  \mathcal{R}_{\dotw,z}(\xx;\y)  \left(y_2-x_2\right)\left(P_{y_1} \mathcal{R}_{\dotw,z}\right)(\y;\xx) \right) \right| \leq C  \e^{- \frac{\delta}{4T}}
	\end{aligned}
	\end{equation}
	which implies that the second term on the right-hand side of \eqref{eq:aux4} vanishes when $T$ goes to zero.
	
	Let us denote by $	F_0:= \varphi_{>a}(1-\varphi)\chi_{(-\infty,\mu]}$. 
	Let us now prove that the first term on the right-hand side of \eqref{eq:aux4} converges exponentially fast to $B'_{F_0}(b_0)$. First, $F^{(<\mu)}$ is supported in the interval $[a,\mu-\delta]$, hence the integral over $\mathcal{D}$ reduces to an integral over $\mathcal{D}_{<\mu}:= [a,\mu - \delta] \times [-1,1]$. Second, we have that 
	$$
	\sup_{z \in \mathcal{D}_{<\mu}} \left| \bar{\partial} F^{(<\mu)}_N (z) -  \bar{\partial}F_{0,N}(z) \right| \leq C_{N,\delta} \frac{|z_2|^N}{\langle z_1\rangle ^N} \e^{-\frac{\delta}{4T}} .
	$$
	The previous estimate, together with \eqref{eq:aux5}, implies:
	$$
	\left| \E\left( \int_{\Omega} \mathrm{d} \xx \int_{\mathcal{D}} \mathrm{d}z_1\mathrm{d}z_2\bar{\partial} {F^{(<\mu)}_N}(z) \int_{\R^2} \mathrm{d}\mathrm{\y}  \mathcal{R}_{\dotw,z}(\xx;\y)  (y_2-x_2) \left(P_{y_1} \mathcal{R}_{\dotw,z}\right ) (\y;\xx) \right)- B'_{F_0}(b_0) \right|\leq  C \e^{-\frac{\delta}{4T}}. 
	$$
	Introducing this estimate and \eqref{eq:zeroT1} in \eqref{eq:aux4},  we get \eqref{hc4}.
	
	Then, since $\textrm{supp}(F_0') \subset \R \setminus \Sigma(b_0)$, we have that $F'_0(H_{\omega,b_0})=0$ for almost every $\omega$. Therefore, by exploiting also the regularity of the integral kernels given by Proposition~\ref{AR1}(ii) to exchange the trace with the expectation with the disorder,  we get
	$$
	\begin{aligned}
    B'_{F_0}(b_0) &=  -\lim_{L\to \infty} \mathbb{E}\left(\Tr  \left( \widetilde{\chi}_{L} \chi_{\infty} \left(\iu \left[H_{\dotw,b}^E,X_1\right]  F'_0(H_{\dotw,b}^E) - \chi_{E} \iu \left[H_{\dotw,b},X_1\right]  F'_0(H_{\dotw,b}) \chi_{E} \right)\right) \right) \\
    &= -\lim_{L\to \infty} \mathbb{E}\left(\Tr\left( \widetilde{\chi}_{L} \mathcal{I}_{\dotw,b} \right) \right) =-\lim_{L\to \infty} \int_0^1 \mathrm{d}x_1 \int_{0}^{L} \mathrm{d}x_2 \,  g\left(\frac{x_2}{L}\right) I_{b}(x_1,x_2) . \\ 
	\end{aligned}
	$$
    Finally, the zero-temperature limit of the bulk-edge correspondence, i.e.  \eqref{eq:ZeroTLimit}, follows by using \eqref{eq:currentL1} and the Lebesgue dominated converge theorem. \hfill \qedsymbol

\appendix

\section{Analysis of integral kernels}

\subsection{Geometric perturbation theory}
\label{subsec:GPT}
In the following it will be important to relate the resolvent of the edge Hamiltonian, $H^E_{\omega,b}$, with the resolvent of the bulk Hamiltonian  $H_{\omega,b}$. 

Let $L\geq1$ and consider the set $\Xi_L(t):=\left\{\x \in E \, | \, \mathrm{dist}(\x,\partial E) \leq t \sqrt{L} \right\}$, $t>0$.  Then let $0\leq \eta_0, \eta_L \leq 1$ be two smooth non-negative functions only depending on $x_2$ such that $\eta_0(\x)+\eta_L(\x)=1$ for every $\x \in E$. Moreover, we assume that 
\begin{equation}
\label{eq:defEta}
\begin{aligned}
&\mathrm{supp}(\eta_0) \subset \Xi_L(2) ,\\
&\mathrm{supp}(\eta_L) \subset E \setminus \Xi_L(1) , \\
&\| \partial^{n}_2\eta_i \|_{\infty} \simeq L^{-\frac{n}{2}},\quad n\geq 1, \quad i \in \left\{0,L\right\} .
\end{aligned}
\end{equation}
We now introduce another couple of  non-negative functions  $0\leq \widetilde{\eta}_0, \widetilde{\eta}_L \leq 1$ again only depending on $x_2$, with the properties:
\begin{equation}
\label{eq:defEta2}
\begin{aligned}
& \mathrm{supp}(\widetilde{\eta}_0) \subset  \Xi_L\left(\frac{11}{4}\right), \\
& \mathrm{supp}(\widetilde{\eta}_L) \subset E \setminus \Xi_L\left(\frac{1}{4}\right), \\
& \widetilde{\eta}_i \eta_i = \eta_i , \quad i \in \left\{0,L\right\} ,  \\
& \mathrm{dist}\left(\mathrm{supp}(\partial_2 \widetilde{\eta}_i), \mathrm{supp}(\eta_i) \right) \simeq \sqrt{L} , \\
&\| \partial^{n}_2{ \widetilde{\eta}_i} \|_{\infty} \simeq L^{-\frac{n}{2}},\quad n\geq 1, \quad i \in \left\{0,L\right\} .
\end{aligned}
\end{equation}
The functions $\widetilde{\eta}_0, \widetilde{\eta}_L$ are a sort of stretched version of ${\eta}_0$ and ${\eta}_L$. 

The space $L^2(E)$ can be canonically identified with a subset of $L^2(\R^2)$, where the identification operator $I_{E}: L^2(E)  \to L^2(\R^2)$ is given by 
$$
I_{E}(\psi):= \chi_{E} \psi \in L^2(\R^2) .
$$ 
In this setting, the multiplication operators $\eta_{L}, \widetilde{\eta}_L$,  can be viewed as operators from $L^2(\R^2)$ to $L^2(E)$, and vice versa.  Moreover, we have the following identity valid on the domain of $H_{\omega,b}$:
$$
\left(H^E_{\omega,b}- z \right)  \tilde{\eta}_{L} =  \left(H_{\omega,b} - z \right)  \tilde{\eta}_{L}.
$$

Define the bounded operator in $L^2(E)$:
$$
U_{L,\omega}(z):=\widetilde{\eta}_{L} \left(H_{\omega,b} - z \right)^{-1} \eta_{L}  +   \widetilde{\eta}_0 \left( H^E_{\omega,b} - z \right)^{-1} \eta_0.
$$
Then, we have
$$
\left(H^E_{\omega,b} - z \right)U_{L,\omega}(z)=1 + W_{L,\omega}(z) \, ,
$$
where $W_{L,\omega}(z)$ is the bounded operator given by 
$$
\begin{aligned}
W_{L,\omega}(z)&:= \left(-2 \iu \nabla \widetilde{\eta}_L \cdot \left(-\iu\nabla - \mathcal{A} - b A\right) - (\Delta \widetilde{\eta}_L) \right)   \left(H_{\omega,b} - z \right)^{-1} \eta_{L}  \\
& \left(-2 \iu \nabla \widetilde{\eta}_0 \cdot \left(-\iu\nabla - \mathcal{A} - b A\right) - (\Delta \widetilde{\eta}_0) \right)   \left( H^E_{\omega,b} - z \right)^{-1} \eta_{0} .
\end{aligned}
$$
Therefore, the resolvent of the edge Hamiltonian obeys the identity:
\begin{equation}
\label{gpt1}
\left(H^E_{\omega,b} - z \right)^{-1}=U_{L,\omega}(z)  - \left(H^E_{\omega,b} - z \right)^{-1}W_{L,\omega}(z) .
\end{equation}

\subsection{Integral and trace-class estimates for bulk and edge Hamiltonian}
\label{appendix:IntKernels}
In this section we collect the important estimates on the integral kernel that we use in the proof of the main result. The proofs relies on the control, with respect to $z \in \C \setminus \R$, of the integral kernels of three type of operators: $\left(H^\star_{\omega,b}- z\right)^{-1}$, $\left[H^\star_{\omega,b}, X_i\right]\left(H^\star_{\omega,b}- z\right)^{-1}$ and $F(H^\star_{\omega,b})$, where $\star \in \{E, \varnothing\}$ and $i \in \{1,2\}$. No gap assumption is required in the following.

Consider the set $\mathcal{D}:=\R \times [-1,1] \subset \C$ and let $z_1:=\mathrm{Re}(z)$, $z_2:=\mathrm{Im}(z)$ and $\zeta:=\langle z_1 \rangle |z_2|^{-1}$. Since the scalar potential $V$ is smooth and periodic, and the random potential $V_\omega$ is smooth, uniformly bounded and with uniformly bounded derivatives, from \cite[Proposition B.8]{CorneanFournaisFrankHelffer} we get that there exists some $\delta>0$ such that for every $z \in \mathcal{D}$ and for all $\x\neq \x' \in \R^2$ we have 
\begin{equation}
\label{IntKRes1}
\sup_{\omega \in \Par}\left|\left(H^\star_{\omega,b}-z\right)^{-1}(\x;\x')\right| \leq C \zeta^4 \left(1+ \left| \ln \left\|\x-\x'\right\| \right| \right) \e^{-\frac{\delta}{\zeta} \left\|\x-\x'\right\| } \qquad \star \in \{E, \varnothing\}.
\end{equation}

Moreover, in view of the regularity of the magnetic potential $\mathcal{A}$ and the definition of $A$, we get a control on the integral kernel of the bulk operators $\left[H_{\omega,b}, X_i\right]\left(H_{\omega,b}- z\right)^{-1}$. For $i \in \left\{1,2\right\}$, it holds true \cite[Proposition B.9]{CorneanFournaisFrankHelffer}, for all $\x\neq \x' \in \R^2$, that
\begin{equation}
\label{IntKernelResolvent2}
\sup_{\omega \in \Par}\left| \left(\left(-\iu\nabla_{\x} - \mathcal{A} - b A \right)_i(H_{\omega,b}-z)^{-1}\right)(\x;\x') \right| \leq C \zeta^4 \left(1+\frac{1}{\|\x-\x'\|}\right)\e^{-\frac{\delta}{\zeta}\|\x-\x'\|} .
\end{equation}
Notice that, when $\star = \varnothing$ and $z \in \C \setminus \Sigma(b)$, the integral kernels in \eqref{IntKRes1} and \eqref{IntKernelResolvent2} are jointly continuous for $\x \neq \x' \in \R^2$,  see \cite[Proposition 3.1]{CorneanNenciu09}.

Let us now prove an estimate analogous to \eqref{IntKernelResolvent2} for the edge Hamiltonian. We will exploit gauge covariant magnetic perturbation theory \cite{CorneanNenciu1998,Nenciu02,CorneanNenciu09,CorneanMonacoMoscolari2018} together with the explicit formula for the integral kernel of the free Laplacian defined with Dirichlet boundary condition on $E$.

Consider the operator $(-\Delta_E+\lambda)^{-1}$, where $\Delta_E$ is the Laplace operator restricted to the upper half plane $\R \times [0,+\infty)$ with Dirichlet boundary conditions.  Using the reflection method we can explicitly write the integral kernel of $(-\Delta_E+\lambda)^{-1}$ :
\begin{equation}
\label{eq:freeLaplacianEdge}
(-\Delta_E+\lambda)^{-1}(\x;\y)= (-\Delta+\lambda)^{-1}(\x;\y) - (-\Delta+\lambda)^{-1}(\x;\y_*) ,
\end{equation}
where $\lambda >0$, $\y_*:=(y_1,-y_2)$, and $(-\Delta+\lambda)^{-1}(\x;\y)$ is the integral kernel of the resolvent of the Laplace operator on $\R^2$, that is 
$$
(-\Delta+\lambda)^{-1}(\x;\x')=\frac{1}{2\pi} K_0 (\sqrt{\lambda} \|\x-\y\|)
$$
where $K_\nu$, $\nu \in \Z$ denotes the Macdonald function. Let us briefly remind here some useful properties of the Macdonald functions:
\begin{align*}
&K'_0(x)=-K_1(x),\\
&K_\nu(x)=\sqrt{\frac{\pi}{2 x}} e^{-x} \left(1+ \mathcal{O}\left(\frac{1}{x}\right)\right) \text{ as } x \to +\infty ,\\
&K_0(x)= -\log(x) + \mathcal{O}(1)  \text{ as } x \to 0, \\
&x^{\nu}K_\nu(x)=2^{\nu-1} \Gamma(\nu) + o(1)  \text{ as } x \to 0, \nu \neq 0.
\end{align*} Consider the operator $S(-\lambda)$, which is the bounded and selfadjoint operator associated to the integral kernel
$$
S(-\lambda)(\x;\y):=\e^{\iu  \varphi(\x,\y)} (-\Delta_E+\lambda)^{-1}(\x;\y) ,
$$
where $\varphi$ is the Peierls phase defined by 
$$
\varphi(\x;\y):= -b\int_{\x}^{\y} A(\mathbf{s}) \cdot \mathrm{d}\mathbf{s} =-\frac{b}{2}(x_1-y_1)(x_2+y_2)
$$
for a fixed $b \in \R$, while $A(\x)=(-x_2,0)$.  If $A_t(\x):=\frac{1}{2}(-x_2,x_1)$ denotes the transverse gauge, then we have the important identity 
\begin{equation}\label{dc20}
bA(\x)-\nabla_{\x} \varphi(\x,\y)=bA_t(\x-\y).
\end{equation}
 
 The operator $S(-\lambda)$ maps the Hilbert space $L^2(E)$ into the domain of the operator $\left(-\iu\nabla - b A\right)^2$ with Dirichlet boundary conditions at $x_2=0$. Moreover, we have
$$
\left(\left(-\iu\nabla - b A\right)^2 + \lambda\right) S(-\lambda) = 1+ T(-\lambda),
$$
where $T(-\lambda)$ is the bounded operator associated to the integral kernel
$$
\begin{aligned}
T(-\lambda)(\x;\y):=-2\e^{\iu  \varphi (\x,\y)}b\;  A_t(\x-\y)\cdot ({-\iu\nabla}_\x) \left((-\Delta+\lambda)^{-1}(\x;\y)  - (-\Delta+\lambda)^{-1}(\x;\y_*)\right)\\
+\e^{\iu  \varphi (\x,\y)}b^2\; |  A_t(\x-\y)|^2\; \left((-\Delta+\lambda)^{-1}(\x;\y) -  (-\Delta+\lambda)^{-1}(\x;\y_*)\right).
\end{aligned}
$$
We note that the kernel $T(-\lambda)(\x;\y)$ is bounded near the diagonal because of the presence of $A_t(\x-\y)$.
For $\lambda>0$, the operator $\left(\left(-\iu\nabla - b A \right)^2 + \lambda\right)$ with Dirichlet boundary condition at $x_2=0$ is invertible, therefore we get
\begin{equation}\label{jul1}
\left(\left(-\iu\nabla - b A\right)^2 + \lambda\right)^{-1}=S(-\lambda) - \left(\left(-\iu\nabla - b A\right)^2 + \lambda\right)^{-1} T(-\lambda) .
\end{equation}
Moreover, if $\lambda$ is sufficiently large, then using the Schur test one can show that the norm of $T(-\lambda)$ goes to zero and we may write 
$$\big ((-i\nabla -bA)^2+\lambda\big )^{-1}=S(-\lambda)\sum_{n= 0}^{\infty} (-1)^n T(-\lambda)^n.$$
Using the theory of polar operators, in particular \cite[Lemma 2, p.213]{Vladimirov} (see also \cite{MoscolariStottrup}), one can show that the above series consists of integral operators with more and more regular integral kernels; in fact only the term with $n=0$ has a logarithmic singularity near the diagonal. Moreover, the series defines a smooth integral kernel outside the diagonal.

Taking the adjoint on both sides of \eqref{jul1} we have:
$$
\left(\left(-\iu\nabla - b A \right)^2 + \lambda\right)^{-1}=S(-\lambda) -T(-\lambda)^* \left(\left(-\iu\nabla - b A\right)^2 + \lambda\right)^{-1} . 
$$
 Exploiting  \eqref{IntKRes1}, \eqref{eq:freeLaplacianEdge} and \eqref{dc20}, we get also for $\lambda$ large enough:
$$
\begin{aligned}
\big |\left(-\iu\nabla_{\x} - b A(\x)  \right)_i & \left(\left(-\iu\nabla - b A\right)^2 + \lambda\right)^{-1}(\x;\y) \big | \\
& \leq \left| \left(-\iu\nabla_{\x} - b A \right)_i S(-\lambda)(\x;\y)\right|  \\
& + \int_{E} \mathrm{d} \tilde{\x}\left| \left(-\iu\nabla_{\x} - b A  \right)_i T(-\lambda)^*(\x; \tilde{\x}) \left(\left(-\iu\nabla - b A\right)^2 + \lambda\right)^{-1}(\tilde{\x};\y) \right| \\
& \leq C  \left(1+ \frac{1}{\|\x-\y\|}\right) \e^{-\frac{\sqrt{\lambda}}{2}\|\x-\y\|} .
\end{aligned}
$$

Then, because all the other scalar and magnetic potentials are bounded, by using the second resolvent identity we get
\begin{equation}
\label{eq:derivativeRLambda}
\begin{aligned}
\sup_{\omega \in \Par}\left|\left(-\iu\nabla_{\x} - b A - \mathcal{A} \right)_i\left( H^E_{\omega,b} + \lambda\right)^{-1}(\x;\y) \right| \leq C'  \left(1+ \frac{1}{\|\x-\y\|}\right) \e^{-\frac{\sqrt{\lambda}}{2}\|\x-\y\|} \, .
\end{aligned}
\end{equation}
Eventually, exploiting the first resolvent identity, together with \eqref{eq:derivativeRLambda} and \eqref{IntKRes1}, we obtain that there exist $\delta, C >0$ such that for every $z \in \mathcal{D}$ it holds true that
\begin{equation}
\label{IntRes2}
\sup_{\omega \in \Par}\left|\left(-\iu\nabla_{\x} - b A - \mathcal{A} \right)_i\left(H^E_{\omega,b}- z\right)^{-1}(\x;\y) \right| \leq C' \zeta^6 \left(1+ \frac{1}{\|\x-\y\|}\right) \e^{-\frac{\delta}{\zeta}\|\x-\y\|} .
\end{equation}

\medskip

In order to shorten our formulas, in the rest of the section we use the shorthand notation $H^\star_{\omega,b}=H^\star_{\omega}$, $\star \in \{E,\varnothing\}$.
\begin{lemma}
\label{lemma:L2toContinuous}
Let $g_0$ and $g$ be two smooth compactly supported functions. We also assume that $\text{supp} (g_0) \subset E$, $\text{dist}\left(\text{supp} (g_0),\partial E\right) >0$. Then, for $\delta>0$ small enough, the operators  $g_0 \iu \left[H^E_{\omega},X_i\right]\left(H^E_{\omega}-z\right)^{-2} \e^{\delta|X|}$ and $g\iu \left[H_{\omega},X_i\right]\left(H_{\omega}-z\right)^{-2} \e^{\delta|X|}$, with $i \in \{1,2\}$ and $z \in \C \setminus \R$, map $L^2(E)$, respectively $L^2(\R^2)$, into the space of continuous functions on $\R^2$. 
\end{lemma}
\begin{proof}
The integral kernel of $\left[H^E_{\omega},X_i\right]\left(H^E_{\omega}-z\right)^{-1}(\x;\x') $ and $\left(H^E_{\omega}-z\right)^{-1}(\x;\x')$ are jointly continuous for $\x\neq\x' \in \R^2$ and they have a singularity that goes like $\|\x-\x'\|^{-1}$, and $\ln{\|\x-\x'\|}$ respectively. Hence, by applying the same result in the theory of polar integral operators  as before, \cite[Lemma 2, p.213]{Vladimirov}, one can show that, if $\mathcal{U}_1,  \mathcal{U}_2\subset E$ are compact sets, the integral kernel  $\left[H^E_{\omega},X_i\right]\left(H^E_{\omega}-z\right)^{-2} (\x;\x')$ is jointly continuous for $(\x;\x') \in \mathcal{U}_1\times\mathcal{U}_2$. By using Lebesgue's dominated convergence theorem together with \eqref{IntKRes1} and \eqref{IntKernelResolvent2}, this  implies that for every $\psi \in L^2(E)$, $i \in \{1,2\}$ and $z \in \C \setminus \Sigma_{E}$,   
$$
\left(g_{0} \iu \left[H^E_{\omega},X_i\right]\left(H^E_{\omega}-z\right)^{-2} \e^{\delta|X|}\psi\right)(\x)=\int_{\R^2} \di \y g_{0}(\x)  \left(\iu\left[H^E_{\omega},X_i\right]\left(H^E_{\omega}-z\right)^{-2}\right)(\x;\y)\e^{\delta|\y|}\psi(\y)
$$ is a continuous function.
\end{proof}

It remains to analyze the integral kernel of $F(H^\star_{\omega})$. We do that by using the H-S formula \cite{HelfferSjostrand1989}.

Consider a function $F$ as in Theorem~\ref{thm-positive}. We still denote by $F$ its extension to a Schwartz function. By following the strategy of \cite{CorneanFournaisFrankHelffer}, we can write $F(H^\star_{\omega})$ as
\begin{equation}\label{dc22}
F(H^\star_{\omega})=-\frac{1}{\pi} \int_{\mathcal{D}} \mathrm{d}z_1\mathrm{d}z_2 \, \bar{\partial} {F_N}(z) (H^\star_{\omega}-z)^{-1} \, , \quad z=z_1+\iu z_2,
\end{equation}
where $\mathrm{Re}z=z_1, \mathrm{Im}z=z_2 \in \R$, $\mathrm{d}z_1\mathrm{d}z_2$ is the Lebesgue measure of $\C\cong \R \times \R$, $\mathcal{D}:=\R \times [-1,1]$, and $F_N$ is an almost analytic extension of $F$ constructed as follows: Let $0\leq g(y)\leq 1$ with $g\in C_0^\infty(\R)$ such that $g(y)=1$ if $|y|\leq 1/2$ and $g(y)=0$ if $|y|>1$. Fix some $N\geq 2$ and define 
	\begin{equation*}
	F_N(z_1+\iu z_2) := g(z_2) \sum_{j=0}^N \frac{1}{j!} \frac{\partial^j F}{\partial {z_1}^j}(z_1) (\iu z_2)^j.
	\end{equation*}
Then the following properties hold true:
\begin{itemize}
	\item $F_N(x)=F(x) \qquad \forall \, x \in \R$;
	\item $\mathrm{supp}(F_N) \subset \mathcal{D}$;
	\item there exists a positive constant $C_N$ ( dependent on $N$) such that
	$$
	\left|\bar{\partial} {F_N}(z)\right| \leq C_N \frac{|z_2|^N}{\langle z_1\rangle^N} \, .
	$$
\end{itemize}
\smallskip

\begin{lemma}
	\label{lemma:IntTrClassEstimateF}
	For all $g_1,g_2\in C_0^\infty (E)$, and all $h_1,h_2\in C_0^\infty(\R^2)$, the operators  $g_1F(H^E_{\omega})g_2$,  $g_1\left(\iu\left[H^E_{\omega},X_i\right]F(H^E_{\omega})\right)g_2$,  $h_1F(H_{\omega})h_2$ and $h_1\left(\iu\left[H_{\omega},X_i\right]F(H_{\omega})\right)h_2$ have a jointly continuous integral kernel. Also, for $\star \in \{E,\varnothing\}$ and every $M>0$, there exists a positive constant $C$ such that
	\begin{equation}
	\label{eq:PolLocalizationF}
	\sup_{\omega \in \Par}\left|F(H^\star_{\omega} )(\x;\y) \right| \leq C \langle \x-\y\rangle^{-M} 
	\end{equation}
	and 
	\begin{equation}
	\label{eq:PolLocalizationFprime}
	\sup_{\omega \in \Par}\left|\left(\iu\left[H^\star_{\omega},X_i\right]F(H^\star_{\omega})\right)(\x;\y) \right| \leq C \langle \x-\y\rangle^{-M} \qquad i \in \{1,2\}.
	\end{equation}
	Moreover, $\chi_{L} F(H^\star_{\omega})$ is trace class and
	\begin{equation}
	\label{eq:TraceIntegral}
	\Tr \left(\chi_{L} F(H^\star_{\omega})\right)  = \int_{\R^2} \mathrm{d}\x \chi_{L}(\x) F(H^\star_{\omega})(\x;\x)  .
	\end{equation}
\end{lemma}
\begin{proof}
	Consider the operator $F(H^E_{\omega})$. In view of the integrability of the integral kernel of the resolvent of $H^E_{\omega}$, see \eqref{IntKRes1}, by using Fubini-Tonelli's theorem together with a Schur estimate, we can prove that $F(H^E_{\omega})$ is a bounded operator from $L^2(E)$ to $L^\infty(E)$ and
	\begin{equation}
	\label{ManyResolvents}
	F(H^E_{\omega})(\x;\y)=-\frac{1}{\pi} \int_{\mathcal{D}} \mathrm{d}z_1\mathrm{d}z_2 \,  \bar{\partial} {F_N}(z) (H^E_{\omega}-z)^{-1}(\x;\y).
	\end{equation}
	
	Moreover, by using Green's theorem and the first resolvent formula, we get
	$$
	F(H^E_{\omega})=-\frac{1}{\pi} \int_{\mathcal{D}} \mathrm{d}z_1\mathrm{d}z_2 \, (z-\iu) \bar{\partial} {F_N}(z) (H^E_{\omega}-z)^{-1} (H^E_{\omega}-\iu)^{-1} \, .
	$$
	Hence, for every $M\geq 2$, by using \eqref{IntKRes1} and by choosing $N$ large enough in \eqref{dc22}, we get that there exist $a>1$ and a constant $C$ such that
	\begin{align}
	\label{PolLoc1}
	&\sup_{\omega \in \Par}\left| \langle \x-\x' \rangle ^{M} F(H^E_{\omega})(\x;\x') \right| \nonumber \\
	&\leq \sup_{\omega \in \Par}\frac{1}{\pi} \int_{\mathcal{D}} \mathrm{d}z_1\mathrm{d}z_2 \, \left|\bar{\partial} {F_N}(z) \right| \left| \langle \x-\x' \rangle ^{M} \int_{E} \mathrm{d} \y (H^E_{\omega}-z)^{-1}(\x;\y)(H^E_{\omega}-\iu)^{-1}(\y;\x')\right| \nonumber  \\
	&\leq C\;  \int_{\mathcal{D}} \mathrm{d}z_1\mathrm{d}z_2 \, \frac{1}{\langle z_1 \rangle^a}<\infty.  
	\end{align}
	 Notice that in the second inequality we have used Cauchy-Schwarz inequality together with \eqref{IntKRes1}, and that $|z_2|\leq 1$. 
    The proof of \eqref{eq:PolLocalizationFprime} can be done in the same way, using \eqref{IntKRes1}. 
	
    Let us now investigate the continuity of the integral kernels by mimicking the arguments in \cite{CorneanNenciu09,CorneanFournaisFrankHelffer}. We only consider $\iu\left[H^E_\omega,X_i\right]F(H^E_\omega)$, since the other operators can be treated in a similar way.
	
	We have
	$$
	\begin{aligned}
	&g_1 \iu\left[H^E_{\omega},X_i\right]F(H^E_{\omega}) g_2\\
	&=\left(g_1 \iu\left[H^E_{\omega},X_i\right] \left(H^E_{\omega}-\iu\right)^{-2} \e^{\delta |X|}\right) \left(\e^{-\delta |X|}\left(H^E_{\omega}-\iu\right)^{-1} F(H^E_{\omega})\left(H^E_{\omega}-\iu\right)^{5} \right) \left(H^E_{\omega}-\iu\right)^{-2} g_2 \\
	&:= A B C \, .
	\end{aligned}
	$$
	By functional calculus, $F(H^E_{\omega})\left(H^E_{\omega}-\iu\right)^{5}$ is a bounded operator, and since $\e^{-\delta |X|}\left(H^E_{\omega}-\iu\right)^{-1}$ is a Hilbert-Schmidt operator, the same holds true for $B$. From the proof of Lemma~\ref{lemma:L2toContinuous} we get that both operators $A$ and $C$ have a bounded integral kernel which may be dominated by a term that goes like constant times $\e^{-\epsilon \Vert \x-\y\Vert}$ for some small enough $\epsilon$. Using the Cauchy-Schwarz inequality, we get that there exists a constant $C_{g_1,g_2}$ such that
	\begin{equation}
	\label{eq:IntF}
	\sup_{\x,\y \in \R^2}\left| g_1(\x) \left(\iu\left[H^E_{\omega},X_i\right]F(H^E_{\omega})\right)(\x;\y) g_2(\y)\right|  \leq C_{g_1,g_2} \|B\|_{2}\, ,
	\end{equation}
	where $\|B\|_2$ denotes the Hilbert-Schmidt norm of $B$. We can now mimic the argument from \cite[Appendix C]{CorneanFournaisFrankHelffer} which uses that the integral kernel $B(\cdot;\cdot)\in L^2(E \times E)$ can be approximated with the kernel of a finite rank operator $B_N=\sum_{j,i=1}^{N}\beta_{i,j}\;  f_i \otimes f_j$, where $\{f_i\}_{i \in \N}$ is an orthonormal basis of $L^2(E)$. 
	Indeed, by applying Lemma~\ref{lemma:L2toContinuous} to the operator $A$ and $C^*$, we get that both operators maps $L^2(E)$ to the space of continuous functions. Since the integral kernel of $AB_NC$ can be written as
	$$
	\sum_{i,j=1}^{N} \beta_{i,j} \left(A f_i\right)(\x) \overline{\left(C^* f_j\right)(\y)},
	$$
	we get that it is jointly continuous.
	
	Since $B_N$ converges to $B$ in the Hilbert-Schmidt norm, together with the estimate \eqref{eq:IntF},  we can uniformly approximate the integral kernel of $g_1 \left(\iu\left[H^E_{\omega},X_i\right]F(H^E_{\omega})\right) g_2$ with the integral kernel of the operator $AB_NC$, which has a jointly continuous integral kernel. In view of the generality of $g_1$ and $g_2$, we obtain that $\iu\left[H^E_{\omega},X_i\right]F(H^E_{\omega})$ has a jointly continuous integral kernel on every compact set $\mathcal{V}_1 \times \mathcal{V}_2 \subset E \times E$.

	Let us now show that $\chi_{L} F(H^E_{\omega})$ is a trace class operator in $L^2(E)$ for every $L>0$. Consider
	\begin{equation}
	\label{eq:DGDeco}
	\chi_{L} F(H^E_{\omega})=\left(\chi_{L} \left(H^E_{\omega}-\iu\right)^{-2} \e^{\delta |X|}\right) \left(\e^{-\delta |X|}\left(H^E_{\omega}-\iu\right)^{-1} F(H^E_{\omega})\left(H^E_{\omega}-\iu\right)^{3} \right):= D G.
	\end{equation}
	As it is done in the proof of Lemma~\ref{AR1}, one can show that $D$ is an Hilbert-Schmidt operator, while the operator $G$ can be proved to be Hilbert-Schmidt in the same way we have done for the operator $B$ above. Therefore, $\chi_{L} F(H^E_{\omega})$ is the product of two Hilbert-Schmidt operators, thus it is trace class.

	Let us now show \eqref{eq:TraceIntegral}. Let $ g_n$ be smooth functions, uniformly bounded and compactly supported in $E$, such that $ g_n^2 \to \chi_{L}$  strongly in the $L^2(E)$ norm and $g_n\chi_{L}=g_n$. The operator $g_n F(H^E_{\omega})$ can be written as a product of two Hilbert-Schmidt operators in the same way as we have done in \eqref{eq:DGDeco}, hence it is a trace class operator. Then:
	$$
	\left| \Tr\left(\chi_{L} F(H^E_{\omega})\right) - \Tr\left(g_n F(H^E_{\omega}) g_n \right)  \right|\leq \left\|\left( g_n^2 -\chi_{L}\right) \chi_{L} F(H^E_{\omega}) \right\|_1   \to 0,
	$$
   where we used the well known fact that if $U_n$ is a sequence of bounded operators strongly convergent to zero and $V$ is trace class, then $U_n V$ converges to zero in the trace norm. 
	Moreover, since $g_n$ is compactly supported in $E$, we have that $g_n F(H_{\omega})g_n$ has a jointly continuous integral kernel, hence we can compute its trace simply by integrating the diagonal of the integral kernel. Putting all this together, we get that
	$$
	\Tr\left(\chi_{L} F(H_{\omega})\right) = \lim_{n \to \infty} \int_{\R^2} \di\x \; g_n^2(\x) F(H_{\omega})(\x;\x).
	$$
	Eventually, by using the estimate \eqref{eq:PolLocalizationF} together with Lebesgue dominated convergence theorem we obtain \eqref{eq:TraceIntegral}.
	The proof for $F(H_{\omega})$ is the same.
\end{proof}

\begin{remark}For the bulk operator $F(H_\omega)$ one may prove a much higher regularity of its integral kernel, using tools from the magnetic psudodifferential calculus developed in \cite{IMP, CHP-Beals, C4}. This will be done somewhere else. 
\end{remark}

\begin{lemma}
\label{lemma:WKern}
The operator $W_{L,\omega}(z)$ in \eqref{gpt1} has a smooth integral kernel that is exponentially localized around the diagonal  and decays faster than any polynomial in $L$; more precisely, for every $M>1$,  there exist $C,\delta_0>0$ and $a\geq 1$ such that for every   $L\geq 1$ and $\zeta:= \langle z_1 \rangle |z_2|^{-1}$ we have: 
\begin{equation}
\label{eq:WKern}
\sup_{\omega \in \Par} \left|W_{L,\omega}(z)(\x;\y) \right| \leq C  \frac{\zeta^a}{L^{M}}     \e^{-\frac{\delta_0}{\zeta}\|\x-\y\|} .
\end{equation}
\end{lemma}
\begin{proof}
$W_{L,\omega}(z)$ contains four terms and each of them can be treated separately, proving that it satisfies \eqref{eq:WKern}. Consider, for example, the term
$$
 \nabla \widetilde{\eta}_L \cdot \left(-\iu\nabla - \mathcal{A} - b A\right)   \left(H_{\omega} - z \right)^{-1} \eta_{L} = \partial_2 \widetilde{\eta}_L \left[H_{\omega}, X_2 \right] \left(H_{\omega} - z \right)^{-1} \eta_{L}
$$
since the other ones can be treated in a similar way. 
The main idea is to couple the exponential localization around the diagonal of the integral kernel of $\left[H_{\omega}, X_2 \right] \left(H_{\omega} - z \right)^{-1}$, see   \eqref{IntKernelResolvent2}, with the fact that the distance between the support of $\nabla \widetilde{\eta}_L$ and the support of $\eta_L$ is of order $\sqrt{L}$, see \eqref{eq:defEta2}. Then, by using the maximum of the function $x^{2M}\e^{-|x|}$ we have:
$$ \e^{-c\sqrt{L}}\leq \frac{(2M)^{2M}}{\e^{2M}} c^{-2M} L^{-M},\quad \forall M>1.$$ 
Hence we can trade off a part of the exponential decay of the integral kernel with any polynomial decay in $L$, at the price of getting a prefactor that grows polynomially in $\zeta$. 
\end{proof}

\section{A short argument for the quantization of the edge conductivity}\label{appendix1}

In the absence of the random potential one may use the so-called edge states in order to directly show that the right-hand side of \eqref{EdgeCurrent} must be an integer. The argument below is well known to physicists and it already appeared in the mathematical literature \cite{CombesGerminet,GrafPorta,DeNittisSchulzBaldes}, but we sketch it here for completeness and adapted to our setting. Let us use the shorthand notation $H^E:=H^E_{b_0}$  for the edge Hamiltonian.

First, let us notice that the operator $H^E$ is $\mathbb{Z}$-periodic in the $x_1$-direction and to a Bloch-Floquet-Zak transform, it is unitarily equivalent to 
$\int_{[-\pi,\pi]}^\oplus \di k_1 \, h^E(k_1) $, where the fiber operator $$h^E(k_1)=\frac{1}{2} (-\iu \partial_{x_1}-\mathcal{A}_1(x_1,x_2)+b_0x_2 +k_1)^2 +\frac{1}{2} (-\iu \partial_{x_2}-\mathcal{A}_2(x_1,x_2))^2+V(x_1,x_2)$$
 is densely defined in $L^2(\mathcal{S}_\infty)$ with periodic boundary conditions on the lateral lines $x_1\in\{0,1\}$ restricted to $E$, and with a Dirichlet boundary condition at the bottom $x_2=0$. The bulk operator $H_{b_0}$ can be written in similar manner, but its fiber $h(k_1)$ will be defined in $L^2([0,1]\times \R)$ with periodic boundary conditions on the infinite lines $x_1\in\{0,1\}$. 
 
 Second, we note that the bulk fiber $h(k_1)$ does not have spectrum inside the gap $[E_-,E_+]$, regardless of which $k_1$ we choose. Moreover, one can prove that $(h^E(k_1)-z)^{-1} -\chi_\infty (h(k_1)-z)^{-1}\chi_\infty$ is compact for any $z$ with ${\rm Im}(z)\neq 0$, thus the spectral projection of $h^E(k_1)$ inside $[E_-,E_+]$ is compact, hence $h^E(k_1)$ can only have finitely many discrete eigenvalues which can be inside $[E_-,E_+]$. They are the {\em edge states}, corresponding to eigenfunctions exponentially localized near the boundary $x_2=0$. According to Rellich's theorem, each edge state eigenvalue $\lambda(k_1)$ can be analytically followed as a function of $k_1$ as long as its value belongs to $[E_-,E_+]$. 
 
 It is important in what follows to note, that the left-hand side of the formula in the next lemma (or the right-hand side of \eqref{EdgeCurrent}) remains the same if we slightly move $E_-$ to the left or $E_+$ to the right. This is because $F_0'$ is supported in $[E_-,E_+]$. Thus we may assume that $E_\pm$ do not coincide with some edge state eigenvalues when $k_1\in \pm \pi$.
 
 \begin{lemma}\label{lemahc}  Let $N<\infty$ be the total number of edge state eigenvalues which can enter the interval $[E_-,E_+]$. Without loss of generality we may assume that no such eigenvalue starts or ends at $E_\pm$, i.e. $\lambda_n(\pm \pi)\not \in \{E_-,E_+\}$. Then: 
 $$-2\pi {\rm{Tr}}\left( \chi_{\infty} \iu \left[H^{E},X_1\right] F_0'(H^{E}) \right)=\sum_{n=1}^N \left(F_0(\lambda_n(-\pi))-F_0(\lambda_n(\pi)) \right),$$
and the right-hand side is an integer.
 \end{lemma}
 
 \begin{proof}
 In the Bloch-Floquet-Zak picture, the fiber of $\iu \left[H^{E},X_1\right] F_0'(H^{E})$ is a finite rank operator given by 
 $$(\partial_{k_1} h^E(k_1)) F_0'(h^E(k_1)).$$
 By direct computation we can prove
 \begin{align*}
 2\pi {\rm{Tr}}\left( \chi_{\infty} \iu \left[H^{E},X_1\right] F_0'(H^{E}) \right)&= \int_{-\pi}^{\pi} \di k_1 \; \Tr \left((\partial_{k_1} h^E(k_1)) F_0'(h^E(k_1))\right) \\
 &=\sum_{n=1}^N \int_{-\pi}^{\pi} \di k_1 \; F_0'(\lambda_n(k_1))\langle \psi_n(k_1),(\partial_{k_1} h^E(k_1))\psi_n(k_1)\rangle,
 \end{align*}
 where the eigenvectors $\psi_n(k_1)$ may be chosen to be smooth in $k_1$ due to Rellich's theorem. Then using the Feynman-Hellmann formula and the chain rule, the formula is proved.

 The  contribution coming from edge states for which $\lambda_n(\pm \pi)\not\in [E_-,E_+]$ consists of $0$'s and $\pm 1$'s. Let us also remark that the sets 
 $$M_+:=\{\lambda_j(\pi) \in (E_-,E_+)\},\quad M_-:=\{\lambda_k(-\pi) \in (E_-,E_+)\}$$
 must be equal because both the spectrum of $h_E(k_1)$ as a set, and the multiplicity of its eigenvalues are $2\pi$-periodic and hence:
 $$\sum_{\lambda\in  M_+}F_0(\lambda)- \sum_{\mu \in M_-}F_0(\mu)=0.$$
 \end{proof}
 
 Thus the Chern  character of $P_0$ can be seen as the number of edge state eigenvalues which initially enter the interval $[E_-,E_+]$ by crossing $E_-$, minus the number of edge states which finally exit the interval $[E_-,E_+]$ by also crossing $E_-$.

 \section{Positive temperature bulk-edge correspondence for the Landau Hamiltonian}
 \label{appendix:Kubo}
Let us consider the pure Landau Hamiltonian $H_{\omega,b} \equiv H_b=\frac{1}{2}\left(-\iu \nabla - b A\right)^2$, with $b> 0$, whose spectrum is given by the Landau levels  $\sigma(H_b)=\left\{e_{n,b}=b(n+\frac{1}{2}) \, |\, n \in \N\right\}$ (for a recent review on the Landau Hamiltonian see \cite{MoscolariPanati}). We denote by $\Pi_{n,b}$ the spectral projection onto the infinite dimensional eigenspace associated to $e_{n,b}$. It is well-known (see for example \cite{DeNittisMoscolariGomi}) that the integrated density of states associated to each Landau level is just given by the value of the diagonal of the integral kernel of $\Pi_{n,b}$ at the point $\x=0$, that is $\Pi_{n,b}(0;0)=\frac{b}{2\pi}$.

Assume now that the Fermi energy does not coincide with any of the Landau levels, that is  $\mu \notin \sigma(H_b)$, for simplicity let us assume that $e_{0,b} <\mu<e_{1,b}$. In \cite{CorneanNenciuPedersen06} the authors analyze the linear response to a constant electric field of a fermionic non-interacting systems described by the Landau Hamiltonian. It turns out that for this particular case  one can compute the transverse DC (direct current) conductivity at positive temperature by considering the Kubo formula and by taking the adiabatic limit:
 \begin{equation}
 \label{eq:LandauHallT}
 \lim_{\eta \to 0} \sigma_H(b,T,\eta)=e^2\frac{n^B(b,\mu,T)}{b}.
 \end{equation}

Furthermore, by explicit computation, see for example \cite{HaiduGummich}, we have that the bulk pressure and the bulk density of states are simply given by
\begin{align}
\label{eq:LandauPressure}
p^B(b,T,\mu)&=-\sum_{n=0}^{\infty} F_{\mu,T}(e_{n,b}) \Pi_{n,b}(0;0) =- \sum_{n=0}^{\infty} F_{\mu,T}(e_{n,b}) \frac{b}{2\pi} \\
\label{eq:LandauDensity}
n^B(b,T,\mu)&= \partial_{\mu} p^B(b,T,\mu) =\sum_{n=0}^{\infty} F_{FD}(e_{n,b}) \Pi_{n,b}(0;0) = \sum_{n=0}^{\infty} F_{FD}(e_{n,b}) \frac{b}{2\pi} .
\end{align}
By taking the derivative with respect to $\mathfrak{B}$ of the pressure \eqref{eq:LandauPressure} we get the bulk magnetization, that is
	$$
	\begin{aligned}
	m^B(b,T,\mu)(b)&= -e \partial_b p(b,T,\mu) =  e\sum_{n=1}^{\infty} F_{FD}(e_{n,b}) \frac{d e_{n,b} }{db}  \frac{b}{2\pi} - e \frac{p(b,T,\mu)}{b}\\
	&=: m^{B, \rm circ}(b,T,\mu) + m^{B, \rm res}(b,T,\mu)
	\end{aligned}
	$$
	where the first term $m^{B, \rm circ}(b,T,\mu)$ denotes the magnetization given by the local circulation of electrons and $m^{B, \rm res}(b,T,\mu)$  denotes the residual magnetization which captures the contribution of edge currents, see Section \ref{sec:magnetization}.

Again by explicit computation we have that
	\begin{equation}
	 \frac{\partial m^{B}(b,T,\mu) }{\partial \mu } =- \sum_{n=1}^{\infty} e F'_{FD}(e_{n,b}) \frac{d e_{n,b} }{db}  \frac{b}{2\pi} -e \frac{n^B(b,T,\mu)}{b} .
	\end{equation}
    Therefore, taking into account \eqref{eq:LandauHallT}, we get
    that 
    \begin{equation}
    \label{eq:BulkSide}
     \lim_{\eta \to 0} \sigma_H(b,T,\eta) =  e^2\frac{\partial n^B(b,T,\mu)}{\partial b} +e \frac{\partial m^{B, \rm circ}(b,T,\mu)}{\partial \mu } = -e \frac{\partial m^{B, \rm res}(b,T,\mu)}{\partial \mu }.
    \end{equation}
    Equation \eqref{eq:BulkSide} provides a complete picture of the bulk transport at positive temperature \cite{Streda,HaiduGummich,Resta2010}. In accordance with the derivation of St{\v r}eda \cite{Streda} we can recognize two contributions to the Hall conductivity: the first one, which is given by the magnetic derivative of the integrated density of states, which is also called the \emph{Fermi sea contribution} since it contains the Fermi-Dirac distribution, and the second one, which is given by the derivative w.r.t. of the chemical potential of $m^{circ}$, which contains the derivative of the Fermi-Dirac distribution and it is therefore called the \emph{Fermi surface contribution} to the conductivity. Actually the two names Fermi sea and Fermi surface contribution make real sense only at zero temperature, but we extend them to positive temperature with a little abuse of notation.
    
    Since $\mu$ is in a spectral gap of the bulk Hamiltonian, one can prove that $\lim_{T\to 0}  \frac{\partial m^{B, \rm circ}(b,T,\mu)}{\partial \mu } = 0$, hence
    \begin{equation*}
        \lim_{T\to 0}  \lim_{\eta \to 0} \sigma_H(b,T,\eta) = e^2\frac{\partial n^B(b,T,\mu)}{\partial b}.
    \end{equation*}
    where the right-hand side is proportional to the Chern character of the projection $\Pi_{0,b}$. Therefore, we have that the Hall conductivity converges to its zero-temperature quantized value. Let us now apply Theorem \ref{thm:Main} to this simple situation. We obtain that 
    \begin{equation}
    \label{eq:MEqI}
    m^{B}(b,T,\mu)=I^E_1(b,T,\mu)
    \end{equation}
    which, coupled with \eqref{eq:BulkSide}, implies that the limit to zero temperature of $\partial_{\mu}I_1^E$ converges to its quantized zero-temperature valued, hence establishing the well-known zero-temperature bulk-edge correspondence.
    
    In view of the discussion of Section \ref{sec:currentDef}, we can now subtract from \eqref{eq:MEqI} the contribution of the magnetization currents, thus obtaining
    \begin{equation}
        m^{B, \rm res}(b,T,\mu)=I^{E, \rm tr}_1(b,T,\mu).
    \end{equation}
    Finally, by taking a derivative with respect to $\mu$ and multiplying by $-e$, we obtain an equality between the bulk conductivity and edge conductance at any temperature for the pure Landau case:
    \begin{equation}
        \sigma_H(b,T,\mu)=-e\partial_\mu m^{B, \rm res}(b,T,\mu)=-e \partial_\mu I^{E, \rm tr}_1(b,T,\mu)=\sigma_E(b,T,\mu).
    \end{equation}


\section{Results for discrete systems}

    \label{subsec:Discrete}
    Our method can be adapted to tight-binding models as considered in \cite{ElgartGrafSchenker}. In what follows, we will sketch (without proofs, to be given elsewhere) how to interpret the results of \cite{ElgartGrafSchenker} in terms of our magnetic response approach to bulk-edge correspondence. 
    
    Let us consider a family of self-adjoint bulk Hamiltonians $\left(H_{\omega,b}\right)_{\omega \in \Par}$ on $\ell^2(\Z^2)$  such that the matrix elements of $H_{\omega,b}$ are given by 
    $$
    H_{\omega,b}(\x,\y)=\e^{\iu b \varphi(\x,\y)} T_{\omega}(\x-\y),
    $$
    with $\varphi(\x,\y):=-\frac{1}{2}(x_1-y_1)(x_2+y_2)$, and there exist $C,\delta>0$ such that 
    $$
   \left|T_{\omega}(\x) \right| \leq C \e^{-\delta\|\x\|}, \; \forall \, \x \in \Z^2. 
    $$
    
    We also assume that $\{H_{\omega,b}\}_{\omega \in \Par}$ satisfies the following covariant relation
\begin{equation}
\label{eq:covariance2D'}
\tau_{b,\gamma}H_{\omega,b}\tau^*_{b,\gamma}=H_{T(\gamma)\omega,b}, \quad  \left(\tau_{b,\gamma}\psi\right)(\x)= \e^{\iu b (\gamma_1-x_1)\gamma_2 }\psi(\x-\gamma)=\e^{\iu b\varphi_2(\x,\gamma)} \psi(\x-\gamma),\quad \forall \, \gamma \in \Z^2,
\end{equation}
where $\tau_{b,\gamma}$ is a family of magnetic translations compatible with the Landau gauge, while $T(\gamma)$ is the canonical action of $\Z^2$ on $\Theta$. Hence $(H_{\omega,b})_{\omega \in \Par}$ is ergodic with respect to the lattice $\Z^2$.

The family of edge Hamiltonians $(H^E_{\omega,b})_{\omega \in \Par}$ is defined by restricting the bulk operators to the set $\Z \times \N$. As in the continuous case, we have that this family is still ergodic with respect to the one-dimensional lattice generated by the vector $(1,0)$:

\begin{equation}
\label{eq:covariance1D'}
\tau_{b,\gamma}H^E_{\omega,b}\tau^*_{b,\gamma}=H^E_{T(\gamma)\omega,b} \qquad \forall \, \gamma=(\gamma_1,0) \in \Z^2 \,.
\end{equation}

In the discrete case, the unit cell simply reduces to one point, which simplifies things a lot. In particular, defining 
$$B_F(b):=\mathbb{E}\left(F\big (H_{\dotw,b}\big )({\bf 0},{\bf 0}) \right ),$$
the analogue of \eqref{ar1} becomes:
\begin{equation}
\label{eq:MainDiscrete1}
\frac{\di B_F}{\di b}(b)= -\lim_{N\to \infty} \mathbb{E}\left(\Tr  \left( \chi_{N}\iu \left[H_{\dotw,b}^E,X_1\right]  F'(H_{\dotw,b}^E)\right)\right).
\end{equation}
The proof in the discrete case follows the same step as in the continuous case. However, the presence of two different phases, namely $\varphi(\cdot,\cdot)$ in the definition of the tight-binding Hamiltonian and $\varphi_2(\cdot,\cdot)$ in the definition of the almost resolvent $S^E_{\omega,\epsilon,z}$, requires some additional care. In particular, there is an extra first order contribution to the magnetic derivative of the "smeared" integrated density of states $\rho_{L,\omega}(b)$ which can be proven to be equal to zero.
Moreover, notice that, in comparison to \eqref{ar1}, on the right-hand side only appears the function $\chi_N$, which now simply denotes the characteristic function of a  finite "strip" containing $N+1$ points located on the vertical axis: $\mathcal{S}_{N}:=\{0\times [0,N]\} \cap \Z^2$, $N\geq 1$. 

We will now relate our method and results  to the approach 
based on switch functions \cite{AvronSeilerSimon,ElgartGrafSchenker}, see also \cite{MarcelliPanatiTauber} for a recent discussion on the topic. Even though we postpone a detailed discussion and full proofs to a future paper, we will briefly sketch the main ideas.

Let $i \in \{1,2\}$ and $N \in \N$, and consider the set $E_i^N:=\{ \x \in \Z^2 \, | \,  x_i < N \}$. We define the switch function $\Lambda_i^{(N)}: \Z \to [0,1]$ to be the characteristic function of $E_i^N$. We also use the shorthand notation $\Lambda_i^{(0)} \equiv \Lambda_i$. By exploiting the key property of the switch function
\begin{equation}
\label{eq:SwitchKey}
\sum_{x_i \in \Z} \Lambda_i^{(N)}(x_i+y_i) - \Lambda_i^{(N)}(x_i) = -y_i
\end{equation}
we can show that 
\begin{equation}
\label{eq:switch-current}
\E\left(\Tr  \left( \chi_{N}\iu \left[H_{\dotw,b}^E,X_1\right]  F'(H_{\dotw,b}^E)\right)\right)=\E\left(\Tr\left( \Lambda_2^{(N)}\iu \left[\Lambda_1,H_{\dotw,b}^E\right]  F'(H_{\dotw,b}^E)\right)\right). 
\end{equation}
As a side remark we note that instead of changing the switch function, one could equivalently move the edge to ``minus infinity" as done in \cite{ElgartGrafSchenker}: let the operator $H_{\omega,b}^{(E,N)}$ be the truncation of the bulk Hamiltonian on the set $\Z \times [-N,+\infty)$, then by exploiting the covariance of the Hamiltonian we get
$$
\E\left(\Tr\left( \Lambda_2^{(N)}\iu \left[H_{\dotw,b}^E, \Lambda_1\right]  F'(H_{\dotw,b}^E)\right)\right) = \E\left(\Tr\left( \Lambda_2\iu \left[H_{\dotw,b}^{(E,N)}, \Lambda_1\right]  F'(H_{\dotw,b}^{(E,N)})\right)\right) \, .
$$
The last equality shows that the magnetic derivative of $B_F$ can be written in terms of a current written by using the switch functions. However, we can also understand the right-hand side of \eqref{eq:switch-current} by means of a suitable magnetic derivative, in the same spirit as we do in Theorem \ref{thm-positive} . Consider the vector ${\bf{\Lambda}}_\x=(\Lambda_1(\x),\Lambda_2(\x))$ and the magnetically perturbed bulk Hamiltonian 
$$
 \widetilde{H}_{\omega,b,\epsilon}(\x,\y)=\e^{\iu \epsilon \varphi\left(\bf{\Lambda}_\x,\bf{\Lambda}_\y\right)} H_{\omega,b}(\x,\y) \, ,
$$
together with its corresponding edge version, $\widetilde{H}_{\omega,b,\epsilon}^E$.
Notice that the magnetic perturbation depending on $\epsilon$ is not singular anymore because $\varphi\left(\bf{\Lambda}_\x,\bf{\Lambda}_\y\right)$ is bounded.
Let $\Omega_L:=[-L,L]^2 \cap \Z^2$ and $\chi_{\Omega_L}$ be its characteristic function. We consider the object 
$$
\mathfrak{B}_F^{(L)}(\epsilon):=\Tr\left(\chi_{\Omega_L} F(\widetilde{H}_{\omega,b,\epsilon}) \right) .
$$
It is not difficult to see that the above quantity diverges when $L\to \infty$. The interesting fact is that its derivative with respect to $\epsilon$ does have a limit.  Namely, using the magnetic perturbation theory adapted to the discrete setting we can show that
\begin{equation}
\label{eq:SwitchMagnetic}
\lim_{L \to \infty} \frac{\di \mathfrak{B}_F^{(L)}}{\di \epsilon}(\epsilon) = -\frac{1}{\pi}  \mathrm{Re} \int_{\mathcal{D}} \mathrm{d}z_1\mathrm{d}z_2\, \bar{\partial} {F_N}(z) \Tr\left( R_{\omega,z}  \left[{H}_{\omega,b}, \Lambda_1 \right] R_{\omega,z} \left[{H}_{\omega,b}, \Lambda_2\right]  R_{\omega,z} \right) \, ,
\end{equation}
where we used the Helffer-Sjostrand formula defined in \eqref{dc22} and the  shorthand notation $R_{\omega,z}\equiv({H}_{\omega,b}-z)^{-1}$. 

We notice two facts: first, the magnetic derivative and the limit $L \to \infty$ do not commute; second, \eqref{eq:SwitchMagnetic} holds true for every configuration of the disorder, namely no average is needed. The right-hand side of \eqref{eq:SwitchMagnetic} plays a key role in the proof of the bulk-edge correspondence as presented in \cite{ElgartGrafSchenker}. It is interesting to see how the expression \eqref{eq:SwitchMagnetic} can be connected to our approach to bulk-edge correspondence. First we need to take the expectation over the disorder and, then, by using twice the key property of the switch function \eqref{eq:SwitchKey} we obtain
\begin{equation}
\label{eq:ItyAnceEquiv}
\begin{aligned}
   &\E \left( \lim_{L \to \infty} \frac{\di \mathfrak{B}_F^{(L)}}{\di \epsilon}(0) \right)\\
   &=-\frac{1}{\pi}  \mathrm{Re} \int_{\mathcal{D}} \mathrm{d}z_1\mathrm{d}z_2\, \bar{\partial} {F_N}(z) \E \left(\Tr\left( R_{\dotw,z}  \left[{H}_{\dotw,b}, \Lambda_1 \right] R_{\dotw,z} \left[{H}_{\dotw,b}, \Lambda_2\right]  R_{\dotw,z} \right) \right) \\
   &=-\frac{1}{\pi}  \mathrm{Re} \int_{\mathcal{D}} \mathrm{d}z_1\mathrm{d}z_2\, \bar{\partial} {F_N}(z) \E \left(\Tr\left( \chi_0 R_{\dotw,z}  \left[{H}_{\dotw,b}, X_1 \right] R_{\dotw,z} \left[{H}_{\dotw,b}, X_2\right]  R_{\dotw,z} \chi_0\right) \right)
   \end{aligned}
\end{equation}
where $\chi_0$ is the characteristic function of the origin ${\bf 0} \in \Z^2$. One may recognize (after some non-trivial computations involving the "standard" magnetic perturbation theory) that the last term of the equality \eqref{eq:ItyAnceEquiv} is nothing but the derivative $B_F'(b)$. By putting together  \eqref{eq:ItyAnceEquiv}, \eqref{eq:SwitchMagnetic} and \eqref{eq:MainDiscrete1} we obtain: 
\begin{equation}
\label{eq:BECEquivalence}
\begin{aligned}
\lim_{N\to \infty} \mathbb{E}\left(\Tr  \left( \chi_{N}\iu \left[X_1,H_{\dotw,b}^E\right]  F'(H_{\dotw,b}^E)\right)\right) & = \frac{\di B_F}{\di b}(b)= \E \left( \lim_{L \to \infty} \frac{\di \mathfrak{B}_F^{(L)}}{\di \epsilon}(0) \right)  \\
&=\lim_{N \to \infty} \E\left(\Tr\left( \Lambda_2^{(N)}\iu \left[H_{\dotw,b}^E,\Lambda_1\right]  F'(H_{\dotw,b}^E)\right)\right) \, .
\end{aligned}
\end{equation}
We remark that the underlying mechanism behind these equations is the same that lies behind the equivalence of Hall conductivity and Hall conductance and it is an intrinsic two-dimensional phenomenon. The equality \eqref{eq:BECEquivalence} has two major implications:
\begin{itemize}
\item If $F$ is the smooth characteristic function of an almost sure spectral island of the family $(H_{\omega,b})_{\omega  \in \Par}$, then the right-hand side of \eqref{eq:SwitchMagnetic} is nothing but the Hall conductance, thus \eqref{eq:SwitchMagnetic} can be rewritten as
\begin{equation}
\label{eq:ExtStreda}
\lim_{L \to \infty} \frac{\di \mathfrak{B}_F^{(L)}}{\di \epsilon}(\epsilon) =  \frac{1}{2\pi} \Tr\left(F({H}_{\omega,b})\left[\left[ F({H}_{\omega,b}), \Lambda_1 \right],\left[ F({H}_{\omega,b}), \Lambda_2 \right] \right]\right)
\end{equation}
which has the same form of the usual St\v reda formula. We can interpret \eqref{eq:ExtStreda} as an \emph{extensive macroscopic St\v reda formula}: the derivative of the macroscopic charge with respect to a localized magnetic field (left-hand side of \eqref{eq:ExtStreda}) corresponds to the Hall conductance (right-hand side of \eqref{eq:ExtStreda}). Therefore, \eqref{eq:ItyAnceEquiv} is nothing but the equality between the disordered average of the Hall conductance (which is constant for almost every $\omega$) and the Hall conductivity, namely it is an equality between the Chern character written using the position operators and the switch functions. Or, in other words, \eqref{eq:ItyAnceEquiv} says that one can either compute the derivative of the density of states with respect to an extended magnetic field (the usual St\v reda formula) or compute the derivative of the macroscopic charge with respect to a localized magnetic field (the \emph{extensive St\v reda formula} \eqref{eq:ExtStreda}). Both procedures are equivalent and produce the Chern character. Finally, \eqref{eq:BECEquivalence} connects the two bulk quantities, on the one hand to the edge conductance written in term of position operators (the first term of the equality \eqref{eq:BECEquivalence}), and on the other hand to the edge conductance written in term of switch-functions (the last term of the equality \eqref{eq:BECEquivalence}), thus showing the equivalence between these two approaches.

\item If $F$ is a function such that $F'$ is supported in a region of dynamical localization for almost every $H_{\omega,b}$, then, for such $\omega$ we are in the setting described in \cite{ElgartGrafSchenker}. We notice that a true zero temperature limit would require to work with $F=F_{\rm FD}$ where the chemical potential $\mu$ belongs to the mobility gap, while taking $\beta$ to infinity. Instead, this procedure is circumvented in \cite{ElgartGrafSchenker} in the following way. First, they prove that the Chern character of the Fermi projection $P_\mu$ is independent on the chemical potential when $\mu$ belongs to a mobility gap. Second, they show that the difference 
$$2\pi \lim_{L\to\infty}\frac{\di \mathfrak{B}_F^{(L)}}{\di \epsilon}(0) - {\rm Chern}(P_\mu)$$
can be expressed only in terms of the pure point spectrum of $F'(H_{\omega,b})$ and the associated eigenfunctions. Third, using the previously mentioned non-trivial result of \cite{ElgartGrafSchenker} and \eqref{eq:BECEquivalence}, one can write 
\begin{equation} 
\begin{aligned}
\label{eq:discreteBECmGap}
\frac{{\rm Chern}(P_\mu)}{2\pi}=&\lim_{N \to \infty} \E\left(\Tr\left( \Lambda_2^{(N)}\iu \left[H_{\dotw,b}^E,\Lambda_1\right]  F'(H_{\dotw,b}^E)\right)\right)\\
&+\text{ term only depending on } F'(H_{\dotw,b}).
\end{aligned}
\end{equation}
The right-hand side of the above equation only depends on $F'$ and it is independent on the choice of $F$, provided that $F'$ is supported in the bulk mobility gap. The authors of \cite{ElgartGrafSchenker} interpret the right-hand side of \eqref{eq:discreteBECmGap} as the edge conductance, and they heuristically argue that the second term on the right hand side comes from the bulk contribution to the magnetization, which in our Section \ref{sec:magnetization} coincides to what we call $m^{B,\mathrm{circ}}(b,T,\mu)$. Furthermore, they also provide an alternative expression for the edge conductance only in terms of the edge Hamiltonian. We stress that in \cite{ElgartGrafSchenker} the authors do not need to take the average with respect to the disorder configuration.
\end{itemize}

\addtocontents{toc}{\protect\setcounter{tocdepth}{0}}

  
\addtocontents{toc}{\protect\setcounter{tocdepth}{2}}

\vfill

{
	
	\begin{tabular}{rl}
	    (H.~D. Cornean) & \textsc{Department of Mathematical Sciences, Aalborg University} \\
		&  Thomas Mannsvej 23, 9220 Aalborg, Denmark \\
		&  \textsl{E-mail address}: \href{mailto:cornean@math.aau.dk}{\texttt{cornean@math.aau.dk}} \\
		\\
		(M. Moscolari) & \textsc{Dipartimento di Matematica, Politecnico di Milano}\\
		& Piazza Leonardo da Vinci 32, 20133 Milano, Italy \\
		&  \textsl{E-mail address}: \href{mailto:massimo.moscolari@polimi.it}{\texttt{massimo.moscolari@polimi.it}} \\
		\\
		(S. Teufel) & \textsc{Fachbereich Mathematik, Eberhard-Karls-Universit\"at}\\
		& Auf der Morgenstelle 10, 72076 T\"ubingen, Germany \\
		&  \textsl{E-mail address}: \href{mailto:stefan.teufel@uni-tuebingen.de}{\texttt{stefan.teufel@uni-tuebingen.de}} \\
	\end{tabular}
	
}

\end{document}